\documentclass[pra,aps,onecolumn,nofootinbib,12pt]{article}
\usepackage{amsmath,amsthm,amssymb,euscript,epsf,epsfig}
\usepackage{array,hyperref}
\usepackage{mathtools}
\usepackage{youngtab}
\usepackage{tikz}
\usetikzlibrary{arrows.meta}
\usetikzlibrary{positioning}
\usepackage{color} 

\newcommand\tr{\operatorname{tr}}

\newcommand{\be}{\begin{equation}}
\newcommand{\ee}{\end{equation}}
\newcommand{\bea}{\begin{eqnarray}\displaystyle}
\newcommand{\eea}{\end{eqnarray}}

\makeatletter
\@addtoreset{equation}{section}
\makeatother

\def\one{{\hbox{ 1\kern-.8mm l}}}
\def\zero{{\hbox{ 0\kern-1.5mm 0}}}

\def\mC{ \mathbb{C}} 
\def\Mult{ {\rm{Mult}} } 
 
\def\sha{ {\rm shallow} }
\def\dep{ { \rm deep } }

\def\Dim{ {\rm{Dim}}}

\def\Vmn{ V_N^{ \otimes m } \otimes \overline{V}_N^{~ \!\!\otimes n } }
\def\barVN{ \overline{V}_N } 
\def\barZ{ \overline{Z}}

\newtheorem{theorem}{Theorem}
\newtheorem{lemma}[theorem]{Lemma}
\newtheorem{example}[theorem]{Example}
\newtheorem{proposition}[theorem]{Proposition}

\newtheorem{definition}[theorem]{Definition}

\newtheorem{observation}[theorem]{Observation}

\newtheorem{comment}[theorem]{Comment}

\def\BDB{\hbox{\rm BDB}}
\def\CBDB{ \hbox{\rm CBDB} } 

\def\End{ {\rm End} } 
\def\BRT{{\rm BRT}} 
\def\bfUnmod{ { \rm \bf Unmod}} 
\def\bfMod{ { \rm \bf Mod}} 
\def\Unmod{ { \rm Unmod} } 
\def\Mod{ { \rm Mod}} 
\def\Excl{ {\rm Excl } } 

\def\cA{{\cal A}}  
  
\def\cG{{\cal G}}

  \def\cR{{\cal R}}

 \def\cZ{{\cal Z}}

\usepackage{graphicx} 

\title{ \hfill { \normalsize QMUL-PH-25-18} \\[1em]   {\Large { \bf Simple harmonic oscillators from non-semisimple walled Brauer algebras } } }

\author{
	Sanjaye Ramgoolam$^{1}$\thanks{\href{mailto:s.ramgoolam@qmul.ac.uk}{s.ramgoolam@qmul.ac.uk}}, 
	Micha{\l} Studzi{\'n}ski$^{2}$\thanks{\href{mailto:michal.studzinski@ug.edu.pl}{michal.studzinski@ug.edu.pl}}\\[1em]
	\normalsize $^{1}$Centre for Theoretical Physics, Department of Physics and Astronomy,\\
	\normalsize  Queen Mary University of London, UK\\
	\normalsize $^{2}$International Centre for Theory of Quantum Technologies, \\ 
	\normalsize  University of Gdańsk, Poland
}

\begin{document}

\date{} 

\maketitle
\vspace{-2.0em} 
\begin{center} 
{ \bf ABSTRACT  } 
\end{center} 

\vspace{-0.5em} 

\normalsize{ 
Walled Brauer algebras $B_N ( m , n ) $ illuminate the combinatorics of mixed tensor representations of $U(N)$, with $m$ copies of the fundamental and $n$ copies of the anti-fundamental representation.  They lie at the intersection of  research in representation theory, AdS/CFT and quantum information theory. They have been used to study of correlators in multi-matrix models 
 motivated by  brane-anti-brane physics  in  AdS/CFT. They have  been applied in computing and optimising fidelities of port-based quantum  teleportation. There is a large $N$ regime, specifically $ N \ge  (m+n)$ where the algebras are semi-simple and their representation theory more tractable.  There are known combinatorial formulae for  dimensions of irreducible representations and associated reduction multiplicities.  The large $N$ regime has a stability property whereby these formulae are independent of $N$.   In this paper we initiate a systematic study of the combinatorics in the non-semisimple regime of $ N = m +n - l $, with positive $l$. We introduce restricted Bratteli diagrams (RBD)  which are useful as an instrument to process  known  data from the large $N$ regime to calculate representation theory data in the non-semisimple regime. We identify within the non-semisimple regime,  a region of $(m,n)$-stability, where $ \min ( m, n ) \ge ( 2l -3)  $ and the RBD  take a stable form depending on $l$ only and not the choice of $ m,n$ within the region. In this  regime, several aspects of the combinatorics of the RBD  are controlled by a  universal partition function for an infinite  tower of simple harmonic oscillators closely related, but not identical, to the partition function of  2D non-chiral free  scalar field theory. 
 } 
     
\newpage 

\tableofcontents

\section{Introduction}
\label{sec:intro}
Walled Brauer algebras have been  actively studied in connection with the representation  theory of unitary groups $U(N)$ or general linear groups $GL(N)$ in mixed tensor spaces, starting from  \cite{KOIKE198957}\cite{BENKART1994529}. These algebras and their $q$-deformations have applications in knot theory \cite{VGTuraev_1990}. They also have a wide range of applications in theoretical physics spanning quantum field theory, string theory and quantum information theory. In the context of gauge-string duality, they have  been used to find orthogonal bases of operators in  matrix quantum mechanics  \cite{KR,EHS}, with motivations coming from brane-anti-brane systems in the AdS/CFT correspondence \cite{malda,witten,gkp}. They have been used to study the one-loop dilatation operator in the quarter-BPS sector of $ N=4$ super-Yang-Mills theory \cite{KimuraQuarter} and the map to quarter-BPS geometries in the AdS dual \cite{KimuraLin}.
 Recently the Walled Brauer algebras has found applications to rapidly developing quantum information science in the context of so called mixed Schur-Weyl duality \cite{WBA2024,grinko2023gelfandtsetlinbasispartiallytransposed}. Its efficient implementation via quantum circuits~\cite{nguyen2023mixedschurtransformefficient} has led to  applications in quantum transmission protocols~\cite{grinko2024efficientquantumcircuitsportbased,fei2023efficientquantumalgorithmportbased,PRXQuantum.5.030354}, particularly in scenarios exhibiting underlying symmetries~\cite{grinko2025phd}. Among other applications in quantum information theory,  we mention here quantum teleportation~\cite{studzinski2017port,StudzinskiIEEE22}, higher-order quantum operations~\cite{Quintino2022deterministic}, quantum sampling problems~\cite{marcinska24} or symmetry reduction in semi-definite programs~\cite{Grinko2024Lin}.

\vskip.2cm 

 The representation  theory of mixed tensor space $ \Vmn$, where $V_N$ is the fundamental representation  of $U(N)$ and $ \barVN $  is the complex conjugate of the fundamental representation and $m,n$ are positive integers, is of interest in matrix theory and gauge-string duality as well as quantum information theory. 
Unitary matrices act as $ U^{ \otimes m } \otimes \overline{U}^{ \otimes n } $ and the decomposition into irreducible representations of $U(N)$  is related  by the mixed-tensor generalisation of Schur-Weyl duality to the representation theory of the walled Brauer algebra,  $B_N ( m , n )$. In the matrix theory context, $U(N)$ is a gauge symmetry and polynomial gauge invariant functions of a complex matrix $Z$  with degree $m$ in $Z $ and degree  $n$ in $ \barZ$ are constructed by composing $ Z^{ \otimes m } \otimes \barZ^{ \otimes n } $ with Brauer algebra elements and taking a trace \cite{KR}. In the quantum teleportation context, $N$ is the dimension of a Hilbert space $H$; Alice and Bob share  $m$ entangled pairs in $(H \otimes H)^{ \otimes m } $ and they use this as a resource to teleport   states in $ H^{\otimes n} $. The fidelity of quantum teleportation is expressed in terms of an appropriate trace of an element in $B_N ( m , n ) $~\cite{ishizaka_asymptotic_2008, studzinski2017port,mozrzymas2021optimal}.  

\vskip.2cm 

The representation theory of $U(N)$ in mixed tensor space, and of the dual algebra $B_N ( m , n  ) $ has a well-understood large $N$ regime, namely where $ N \ge  (m+n)$ and the algebra is semi-simple. There is  a more subtle regime, namely $ N < ( m+n)$  where $B_N ( m , n ) $ is non-semi-simple. The irreducible 
representations of $B_N ( m , n ) $ in the large $N$ regime are labelled by triples 
$ ( k , \gamma_+ , \gamma_- ) = \gamma $, where $  0 \le k \le \min ( m,n) $, $\gamma_+ $ is a partition of $ (m-k) $ and $ \gamma_- $ is a partition of $ (n-k)$. We refer to these triples 
as Brauer representation triples (abbreviated \BRT)  and the set of these triples for  fixed $(m,n)$ is independent of $N$ and is denoted $\BRT (m,n)$. For any $N$, there is a map $ \Gamma : ( \gamma, N ) \rightarrow \Gamma ( \gamma , N ) $ where $ \Gamma ( \gamma , N ) $ is a  mixed Young diagram with exactly $N $ rows, and with positive rows determined $ \gamma_+ $ and negative rows determined by $ \gamma_- $. 
 The explicit map is given in equation \eqref{mixedYD}. These row lengths determine the highest weight of the $U(N)$ irrep, and we will use $\Gamma ( \gamma , N ) $ to refer to the mixed Young diagram or the corresponding highest weight.  The decomposition of mixed tensor space, $ V_{ N }^{ \otimes m } \otimes \overline {V}_N^{~ \otimes n } $  into  irreducible representations $ V_{ \gamma }^{U(N)}  $  of $U(N)$ with highest weight $ \Gamma ( \gamma , N) $ has, in the semi-simple regime,  multiplicities  $ \Mult (V_{ \gamma }^{U(N)}  ) $  which are  equal to  dimensions $ d_{ m,n} ( \gamma ) $ of Brauer algebras $B_N ( m,n) $. Importantly these multiplicities/dimensions $ d_{ m,n} ( \gamma ) $  
are independent of $N$, which is referred to as a large $N$ stability property,
\begin{align}  
 & \hbox{ large  $N$-stability  for  }  N\ge (m+n)  : \cr 
 & ~~~~~ \Mult (V_{ \gamma }^{U(N)}  ) = d_{ m,n} ( \gamma )     \, . 
\end{align} 

The formula for $ d_{ m,n} ( \gamma ) $  is  \eqref{eq:dimWBA}.
In the non-semi-simple regime,  the multiplicities are  dimensions of irreducible representations of a semi-simple quotient $ \widehat B_N ( m,n)$ of $ B_N ( m,n)$. These multiplicities may in general have $N$-dependence, 
\begin{align} 
   \hbox{ Semi-simple}  & \hbox { quotient $ \widehat B_N ( m,n)$ for } ~~ N < (m+n) :  \cr 
 & \Mult ( V_{ \gamma }^{U(N)}   ) = \widehat{ d }_{ m,n ,  N } ( \gamma )   \, . 
\end{align} 
 At any given $ ( m,n,N)$ in the non-semi-simple regime, some subset of the Brauer triples for the specified $(m,n)$  will have $ \widehat{d}_{ m,n , N } ( \gamma )  = d_{ m,n } ( \gamma ) $, i,e. dimensions unmodified from the large $N$ regime. The complement will have 
 a modified dimension  $ \widehat{d}_{ m,n , N } ( \gamma ) $ which is smaller than $ d_{ m,n} ( \gamma )$.

This is reviewed with more technical detail in Section \ref{Sec:l=2} and key points are summarised  along with the terminology of this paper in Section \ref{secterm}. Finding general  formulae 
for  $ \widehat {d }_{  m , n, N } $ is an interesting open problem, which is one technical motivation for this paper. The main results in this paper are :   

\begin{enumerate}

\item For the non-semi-simple regime with general $m,n$ and  $ N = m+n-l$ with 
	 $ 1 \leq l \leq 4 $, we calculate the modified multiplicities/dimensions 
	 $ \widehat{ d }_{  m , n, N } $. The calculations are explained  and the results presented in section \ref{sec:DimExamples}.

\item As a tool for calculating these multiplicities, we introduce the notion of restricted Bratteli diagrams (RBD). These are graphs where the vertices are organised in layers labelled by a depth $d$ ranging over $ 0 \le d \le l-1 $. Their vertices are of two types, which we refer to as red and green nodes. The green nodes at $ d=0$ are associated with irreps of $U(N)$ appearing 
in the decomposition of $ \Vmn$  which have modified multiplicities compared to the large $N$ regime. The red nodes appear at depths $  1 \le d \le l $ 
and admit paths connecting them to green nodes at $ d=0$. Further description and examples of the RBD are given Sections \ref{sec:introRBDB}\ref{sec:DimExamples} and key properties are summarised in Section \ref{secterm}.

\item We show that the restricted Bratteli diagrams for $(m,n, l ) $ are independent of $(m,n)$ in the range $ m,n \ge ( 2l-3)$. We refer to this property as degree-stability, or $(m,n)$-stability. The degree-stability terminology is based on the analogy to the independence of $N$ in the representation theory of $ B_N ( m,n) $ for large enough $N$.  It is also motivated by the matrix theory application to gauge invariant polynomials \cite{KR}.  A consequence is that the 
separation of the set of irreps of $ B_N ( m, n )$ into those with modified and unmodified dimensions is independent of $(m,n)$ when $ m , n \ge (2l-3)$. This result is  in Section \ref{mnstability}. 

\item We show that the counting of red and green nodes in the restricted Bratteli diagram for 
$B_N( m , n)$ as a function of depth $d$, in the $(m,n)$-stable regime, is expressible in terms of an integer sequence and associated  generating function  $ Z_{ \rm univ } ( x ) $ (equation \eqref{genfunOsc})  which has a simple interpretation in terms of an infinite family of harmonic oscillators. This sequence is, somewhat surprisingly to us, already recorded in OEIS as A000714. The  connections to the representation theory of non-semisimple walled Brauer algebras are, as far as we know, novel. The links between this partition function and the counting of red and green nodes,  in the restricted Bratteli diagrams in $(m,n)$-stable regime, are developed in Sections \ref{sec:derSHO} and  \ref{sec:greennodes}. The key results  are in equations 
\eqref{ResDeepUniv} \eqref{RedsOscSum} \eqref{GreensUniv} \eqref{GreensUniv1}.  The generating function $ Z_{ \rm univ } ( x ) $ is close in form  to the partition function of a scalar field in two dimensions, which has been discussed in connection with low-dimensional large $N$ gauge-string duality (e.g. equation 
(3.3) in \cite{Douglas1993}), but it has the additional factors $ { 1 \over ( 1 - x) ( 1- x^2) } $. 
	   
\end{enumerate}

\vskip.2cm 

The paper is organised as follows. Section \ref{secterm} gives a summary of the key background and new concepts introduced in this paper, with the associated notation and terminology. We invite the reader to proceed from this introduction to the sections of interest, and to use Section \ref{secterm}  as a reference as needed. In section \ref{Sec:l=2} we give a technical description of the motivations we have outlined above. The Brauer algebras $B_N ( m , n ) $ are represented using  a map  $\rho_{ N  , m , n } $ to 
linear operators acting on mixed tensor space $ \Vmn $. The vector space of these linear operators is $ \End ( \Vmn) $.   The  image of this map is the commutant of, i.e.  the sub-algebra of $\End ( \Vmn) $ which commutes with, the operators representing $U(N)$ in the mixed tensor space. In the  regime of $ N < ( m+n)$, the map has a non-trivial kernel.  There is an analogous relation between the group algebra $ \mC ( S_{ m+n } ) $ of the symmetric group of permutations of $(m+n) $ distinct objects and the action of $U(N)$ on $V_N^{ \otimes (m+n)} $. 
In this case there is a map $ \rho_{ N , m+n } $ which also has a kernel for $ N < (m+n)$. 
In section \ref{sec:relkern} we explain that there is a very useful relation between the Kernels of these maps, given by an operation of partial transposition which has been useful both in  the matrix theory context \cite{KR,KRHol} and the quantum teleportation context \cite{Moz1,Yongzhang2021permutation}. 
Section \ref{sec:introRBDB} defines the restricted Bratteli diagrams (RBD), which have two types of nodes labelled red and green in our convention. Section \ref{sec:DimExamples} use the RBD to calculate the modified dimensions of irreducible representations of $B_{N } ( m  , n ) $, with $ N = m+n -l$ for examples with $ l \in \{ 1, 2, 3, 4 \} $, with general $(m,n)$. In section  \ref{Sec:countingReds} we study the distribution of red nodes in the RBD. We show
 that the RBD has $l$ layers, which we label with a depth variable $d$ in the range  $  0 \le  d \le (l-1) $. We establish the $(m,n)$-stability result and we give a counting formula for 
 $ \cR ( l, d ) $, the number of red nodes as a function of $ l $ and $d$ in the $(m,n)$ stable regime. In section \ref{sec:derSHO}  we establish relations between the counting of red nodes in  the $(m,n)$-stable regime and the oscillator partition function $ \cZ_{ \rm univ } (x) $ in equation \eqref{Zuniv}. In section \ref{sec:greennodes}, we prove  the equation \eqref{GreensUniv}  relating the counting of green nodes at $ d =0$ to $Z_{ \rm univ} ( x ) $ and give a simple argument to extend this to general $d$ in  \eqref{GreensUniv1}. 

\vskip.2cm

We hope  that this paper will be of interest  to mathematicians and theoretical physicists, particularly researchers working in AdS/CFT  and related models of gauge-string duality as well as researchers in quantum information theory. While the mathematical results are rigorous and the derivations complete, they are presented in informal physics style. The mathematica code which is used to construct the  RBD,  which motivated several of the mathematical results we prove and which can serve as useful source of examples for the reader, is described in the Appendix A and is available alongside the arxiv version~\cite{RamgoolamStudzinski2025WBAcode}. In Appendix B, we perform checks of the modified dimensions $ \widehat{d}_{ m,n,N} $ calculated in section \ref{sec:DimExamples} by directly verifying, for small values of $m,n$, the identity for dimensions of mixed tensor space which follows from its decomposition into irreducible representations of $ \widehat{B}_{ m,n, N } $. Appendix C contains some additional figures, restricted Bratteli diagrams, which are useful for the calculations of modified dimensions in section \ref{sec:DimExamples}.

\section{Key concepts and terminology }\label{secterm}

\vskip.2cm 
\noindent 
{\bf Irreducible representations} will be abbreviated as irreps. 

\vskip.2cm 
\noindent 
{\bf Partitions and Young diagrams} \\ 
A partition $\mu $  of a positive integer $n$ is a sequence of positive integers 
$ [ r_1 (\mu )    , r_2 ( \mu )  , \cdots , r_{ h} ( \mu )   ] $, or more briefly
 $ [ r_1     , r_2   , \cdots , r_{ h}   ] $, summing up to $n$ and listed according to $r_i \ge r_{ i+1} $.  We write $ \mu \vdash n $. We associate a Young diagram to $ \mu $, where $r_i $ are the row lengths.

\vskip.2cm 
\noindent 
{\bf Mixed tensor space }\\ $V_N$ is the fundamental or defining  representation of the unitary group  $U(N)$, while $ \overline{ V}_N$ is the anti-fundamental or complex conjugate representation of $V_N$. Mixed tensor space parameterised by $m,n,N$ is the tensor product $\Vmn $. Unitary group elements $U \in U(N)$ act as $ U^{ \otimes m } \otimes \overline{U}^{ \otimes n } $ on the tensor product.

\vskip.2cm 
\noindent 
{\bf Walled Brauer algebra $B_{N} ( m , n ) $ } \\
A diagram algebra with basis given by a set of diagrams. The diagrams have two rows, each with $(m+n)$ nodes. Each node has one incident line. It is useful to visualise a wall separating the first $m$ nodes from the subsequent $n$ nodes, and the rules of construction of the diagrams employ the separating wall between the $m$ and  the $n$ nodes on each row. The complete definition and examples are given at the start of section \ref{Sec:l=2}. 

\vskip.2cm 
\noindent 
{\bf Mixed tensor space as a representation of $B_N ( m , n) $ } \\ 
There is a homomorphism from $ \rho_{ N ,  m , n }  : B_N ( m , n )  \rightarrow \End ( \Vmn  ) $, defined by associating the lines of the Brauer diagrams to Kronecker delta-functions in tensor indices. This is described in section \ref{Sec:l=2}. The kernel of the homomorphism is studied in section~\ref{sec:relkern}.  

\vskip.2cm 
\noindent 
{\bf Walled Brauer diagram space $ B ( m , n ) $ } \\ The vector space with diagrams as basis, underlying the algebras $B_N ( m, n ) $ for all $N$. We use this notion in  section \ref{sec:relkern}.

\vskip.2cm
\noindent  
{\bf Partial transposition  $P^t_{ m,n} $ } \\ 
It is a map, which squares to one, and takes elements of $ \mC ( S_{ m+n} ) $ to $ B ( m , n ) $ and vice-versa. See equation \eqref{eq:mapPt}.

\vskip.2cm
\noindent  
{\bf The semisimple regime of parameters, also called the large $N$ regime. } \\ 
In the regime $ N \ge  ( m +n ) $, the sub-algebra of the $ \End ( \Vmn  ) $ commuting with $ U^{ \otimes m } \otimes \overline{U}^{ \otimes n } $ is  
$\rho_{ N , m , n } ( B_N ( m , n ) )  $. 

\vskip.2cm 
\noindent
{\bf The non-semisimple regime of parameters } \\ 
In the regime $ N  < (m+n)$,  the algebra $B_{N} ( m , n ) $ is non-semisimple. The homomorphism 
$ \rho_{ N , m , n } $  has a non-trivial kernel  $I_N ( m , n ) $. The quotient $ \widehat B_{N} ( m , n ) = B_{N} ( m , n )/I_N ( m , n ) $ is semisimple and is isomorphic to the commutant of $ U^{ \otimes m } \otimes \bar U^{ \otimes n } $ in $ \End ( V_N^{ \otimes m } \otimes \bar V_N^{ \otimes n } ) $, denoted by $\cA^N_{ m,n}  $. 

\vskip.2cm
\noindent  
{\bf Symmetric group algebra and tensor space } \\ 
For $U(N)$ acting on $ V_{N}^{ \otimes m +n } $, the commutant is  the image under a map $ \rho_{ N, m+n } $ from the group algebra $ \mC ( S_{ m+n} ) $ of the symmetric group $  S_{ m+n} $  of permutations of $\{ 1, 2, \cdots , m+n \} $ to  $ \End ( V_N^{ \otimes (n+n)} )  $.  Section~\ref{sec:relkern} gives an isomorphism between the kernels of $\rho_{ N , m+n }  $  and $ \rho_{ N, m,n} $.

\vskip.2cm
\noindent 
{\bf Brauer representation triples $  ( k , \gamma_+ , \gamma_- )  $  for the walled-Brauer pair $( m , n ) $ } \\ 
$k$ is an integer obeying $  0 \le  k \le \min ( m  , n ) $, $ \gamma_+ $ is a Young diagram with $m-k $ boxes and $ \gamma_- $ is  a Young diagram with $(n-k) $ boxes. Also abbreviated as 
Brauer triples. We will frequently use the abbreviation $ \gamma = ( k  , \gamma_+ , \gamma_- ) $. See the discussion around equation~\eqref{eq:WBAlabelling} in section~\ref{Sec:l=2} for further details.

\vskip.2cm 
\noindent 
{\bf Brauer  triples and the representations of $B_N ( m,n)$ in the large $N$ regime  } \\ 
 In the large $N$  regime $ N \ge  ( m +n ) $, the irreducible representations of $B_N ( m , n ) $ are in 1-1 correspondence with triples $ ( k , \gamma_+ , \gamma_- )$. This is described further in section~\ref{Sec:l=2}.
 
 \vskip.2cm 
 \noindent 
 {\bf First column lengths (Heights)  }\\ The first column lengths of $ \gamma_{\pm }\in \{\gamma_+,\gamma_-\} $ are denoted 
 $c_1 ( \gamma_{ \pm} ) $. We also refer to these as the heights of $ \gamma_{ \pm } $, and we define $ {\rm ht } ( \gamma ) = c_1 ( \gamma_{ + } ) + c_1 ( \gamma_{ - } )$.

\vskip.2cm 
\noindent 
{\bf Semisimple quotient } \\ 
$ \widehat {B}_{ N } ( m,n)   := B_N ( m , n ) / I_N ( m  , n ) $. $I_N (m,n) $ is the kernel of the map $ \rho_{ N , m , n } $. The quotient is the image in $ \End ( V_N^{ \otimes m } \otimes \bar V_N^{ \otimes n } ) $ of $ \rho_{ N , m , n } $ and is isomorphic to the sub-algebra of 
 $ \End ( \Vmn  )$ which commutes with $U^{ \otimes m } \otimes \overline {U}^{ \otimes n } $. 
 This commutant algebra is  denoted  as   $\mathcal{A}^N_{m,n}$, which is the algebra of partially transposed permutation operators in the terminology  of the quantum information literature.  See further discussion around \eqref{AquotB}. 

\vskip.2cm 
\noindent 
{\bf Irreducible representations of $ \widehat B_N ( m  , n ) = \cA^N_{ m,n}  $ and admissible Brauer  triples  for $B_N ( m , n ) $ }.\\
 The irreps of $ \widehat B_N ( m  , n ) $ are labelled by admissible triples which obey $ c_1 ( \gamma_+ ) + c_1 ( \gamma_- ) \le N $.

\vskip.2cm
\noindent  
{\bf  $\bfUnmod ( N , m , n )      $: Unmodified set of Brauer triples for $ ( N , m , n ) $.  }\\
This is the set of  admissible Brauer triples $ \gamma = ( k , \gamma_+ , \gamma_- ) $ for $ ( N , m , n ) $ which have dimensions, as representations of $  \widehat B_N ( m , n )= \cA^N_{ m,n} $, which are identical to the stable range dimension formula $ \dim_{ m , n } ( \gamma ) $.

\vskip.2cm 
\noindent 
{\bf  $\bfMod ( N , m , n )      $: Modified set of Brauer triples for $ ( N , m , n ) $.  } \\
This is the set of  admissible Brauer triples $ \gamma = ( k , \gamma_+ , \gamma_- ) $ for $ ( N , m , n ) $ which have dimensions, as representations of $  \widehat B_N ( m , n ) $, which are  
  $ \widehat{d}_{ m,n,N } ( \gamma ) = d_{ m,n} ( \gamma ) - \delta_{ m,n,N} $ with positive
  $\delta_{ m,n,N}$.

\vskip.2cm
\noindent  
{\bf Mixed Young diagram for $ B_N ( m , n )  $  } \\ 
Admissible Brauer representation triples  $ ( k , \gamma_+ , \gamma_- )$ for $ B_N ( m , n ) $ can be presented as mixed Young diagrams with $N$ rows. See the explanations around  \eqref{mixedYD}.

\vskip.2cm 
\noindent 
{\bf Bratteli diagram  for  walled-Brauer pair $ ( m , n ) $ as a layered graph} \\ 
Can be used to compute dimensions of irreps of the Brauer algebra $B_N (m,n)$ in stable large $N$ regime. A more detailed description is after \eqref{eq:dimWBA}.

\vskip.2cm 
\noindent 
{\bf Bratteli moves } \\ 
 As $L$ is increased from $0$ to $m$, a Bratteli move is the addition of a box to $ \gamma_+$ which produces a  valid Young diagram (weakly increasing row lengths). As $L $ is increased from $ (  m+1 ) $ to $(m+n)$, a Bratteli move is either the addition of a box to $ \gamma_-$  or the removal of a box from $ \gamma_+ $.

\vskip.2cm 
\noindent 
{\bf Large-$N$ stability or Rank-stability } $N$ is the rank of $U(N)$, hence the terminology rank-stability. With a view to applications in quantum information where $N$  is called $d$ this may also be called qdit-stability. This refers to the fact that representation theory data for $B_N ( m , n ) $ is indepdendent of $N$ for $ N \ge  (m+n)$. Examples of such representation theory data are dimensions of irreps labelled by Brauer triples $ (k, \gamma_+  , \gamma_- )$ (e.g. \eqref{eq:dimWBA}), decomposition multiplicities of irreps of $ B_N ( m , n ) $ into irreps of the sub-algebra $ \mC ( S_{ m) } \otimes \mC ( S_{ n) }  $. These formulae for the  $N$-stable regime are available in \cite{BENKART1994529} \cite{KOIKE198957}.

 \vskip.2cm 
 \noindent 
 {\bf  Young diagrams with first columns removed, $ ( \gamma_+ \setminus c_1 )$  and $ ( \gamma_ - \setminus c_1 )$  }  \\ 
The respective numbers of boxes are denoted as $ | \gamma_+ \setminus c_1 | $ and  $ | \gamma_- \setminus c_1 | $. This notation is used in  sections~\ref{Sec:countingReds},~\ref{sec:derSHO}, and~\ref{sec:greennodes}.

\vskip.2cm 
\noindent 
 {\bf  Young diagrams with first two columns removed, $ ( \gamma_+ \setminus  \{ c_1  , c_2 \} )$  and $ ( \gamma_ - \setminus \{ c_1 , c_2 \} )$ } \\ 
 The respective numbers of boxes are 
 $ | \gamma_+ \setminus  \{ c_1  , c_2 \}  |  $ and $  | \gamma_ - \setminus \{ c_1 , c_2 \} |  $. This notation is  used in section~\ref{mnstability}, where we study the stability region of RBDB$_N(m,n)$-diagrams.

\vskip.2cm 
\noindent
{\bf Coloured Bratteli diagrams (CBD)  for $B_N(m,n)$ } \\
Defined for the non-semisimple regime $ N < (m+n)$. See Definition \ref{CBDBNmn}.

\vskip.2cm 
\noindent
{\bf Level and depth in Bratteli diagrams } \\
The Bratteli diagram for the walled-Brauer pair $ ( m ,n ) $  has layers labelled by levels $  0 \le L \le (m+n)$. 
We define depth $d$ by the equation $d = m+n - L $. The depths for $\BDB (m,n) $ thus range from $0$ to $(m+n)$. The depths for $ \BDB (m , n ) $ in the $(m,n)$-stable regime ($m,n \ge (2l-3)$) range over $ 0 \le d \le (l-1) $. This notion is used in combinatorial considerations when counting number of red and green nodes among section~\ref{Sec:countingReds} and~\ref{sec:greennodes}.

\vskip.2cm 
\noindent 
{\bf Restricted Bratteli diagram (RBD) for  $B_N(m,n)$ } \\
A restriction of the coloured Bratteli diagrams defined by keeping at $ d=0$  only the nodes for 
Brauer representations of $ B_N (m,n)$ which have modified dimensions, and only the red  nodes $  d \ge 1 $  which link to the green nodes at $d=0$ along with any intermediate green nodes that appear in the paths from the red nodes to the greens at $ d=0$. 
See further description  in  \ref{DefRBDB}. The RBD are useful in calculating the modified dimensions $ \widehat{d}_{ m,n, N }  $ using as input the dimensions $ d_{ m,n } $ from the stable large $N$ regime. These calculations are in  section \ref{sec:DimExamples}. 

\vskip.2cm 
\noindent 
{\bf Admissible Brauer representation triples for $B_N(m,n)$  : green nodes } \\
The nodes associated with  admissible Brauer triples, with $ c_1 ( \gamma_+ ) + c_1  ( \gamma_- ) \le N $,  are coloured green in the CBD or the RBD of  $B_N(m,n)$. The problem of counting green nodes in the RBD is addressed in Section~\ref{sec:greennodes}.

\vskip.2cm 
\noindent 
{\bf Excluded  Brauer representation triples for $B_N(m,n)$ : red nodes } \\
The nodes associated with excluded triples,  $ c_1 ( \gamma_+ ) + c_1  ( \gamma_- ) >  N $
are coloured red in the CBD or RBD  $B_N(m,n)$  or  $B_N(m,n)$. The problem of counting red nodes is addressed in section~\ref{Sec:countingReds}. In section~\ref{sec:derSHO} we connect the counting problem with simple harmonic oscillators.

\vskip.2cm 
\noindent 
{\bf Degree-stability or  $(m,n)$-stability } \\ 
 For $ N  = (m+n-l)$, when $ m , n \ge (2l-3)$, the restricted Bratteli  diagram of $B_N ( m , n ) $ is independent of $m,n$. This stability property is explained with examples ad proved in section \ref{mnstability}.

\section{Technical background and motivation}
\label{Sec:l=2}
The walled Brauer algebra $B_{N}(m,n)$, where $m,n\geq 0$, and $N \in \mathbb{C}$, was introduced and studied  in~\cite{VGTuraev_1990,KOIKE198957,BENKART1994529,BEN96} and has been the subject of subsequent developments in representation theory, see e.g. for recent mathematical literature \cite{bulgakova:tel-02554375,Cox1}.  The abstract algebra $B_{N}(m,n)$ is composed of formal combinations of diagrams. Each diagram has two rows with $m+n$ nodes, associated with a vertical wall between the first $m$ and the last $n$ nodes. These nodes are connected up in pairs in such a way that:
\begin{enumerate}
	\item if the two nodes are in the same row, they must lie on different sides of the wall,
	\item if the two nodes are in different rows, they must lie on the same side of the wall.
\end{enumerate}
Notice that the dimension of $B_{N}(m,n)$ agrees with the number of elements in $S_{m+n}$ which is $(m+n)!$.
We illustrate the above construction with the notion of composition of such diagrams in Figure~\ref{fig:WBA}. 
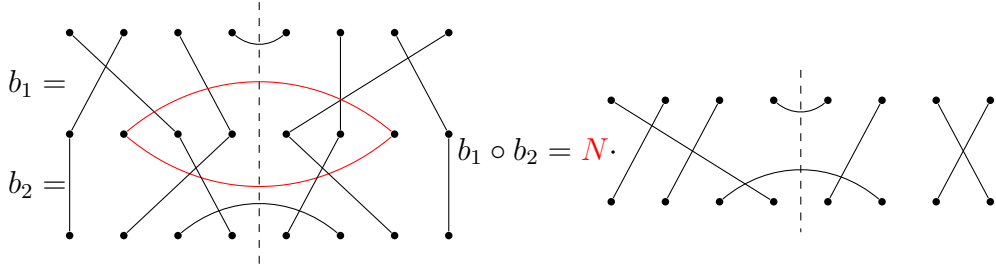
\begin{figure}[h!]
	\centering
	\begin{tikzpicture}[scale=0.9, every node/.style={inner sep=2pt}]
	\def\n{8}
	\def\sep{0.8}
	\def\h{1.5}
	
	\foreach \i in {1,...,8} {
		\node[circle, fill=black, inner sep=1pt] (s\i) at ({\i*\sep}, {2*\h}) {};
		\node[circle, fill=black, inner sep=1pt] (sp\i) at ({\i*\sep}, {\h}) {};
	}
	
	\foreach \i in {1,...,8} {
		\node[circle, fill=black, inner sep=1pt] (pp\i) at ({\i*\sep}, {0}) {};
	}
	
	\node at (0.5*\sep, 1.5*\h) {$
		b_1 = \ $};
	\node at (0.5*\sep, 0.5*\h) {$
		b_2 = \ $};
	
	\draw[dashed] (4.5*\sep, {2.3*\h}) -- (4.5*\sep, {-0.3*\h});
	
	\draw (s1) -- (sp3);
	\draw (s2) -- (sp1);
	\draw (s3) -- (sp4);
	\draw (s4) to[bend right=40] (s5);
	\draw (s6) -- (sp6);
	\draw (s7) -- (sp8);
	\draw (s8) -- (sp5);
	\draw[red] (sp2) to[bend left=40] (sp7);
	
	\draw (sp1) -- (pp1);
	\draw[red] (sp2) to[bend right=40] (sp7);
	\draw (sp3) -- (pp4); 
	\draw (sp4) -- (pp2);
	\draw (pp3) to[bend left=40] (pp6);
	\draw (sp5) -- (pp7);
	\draw (sp6) -- (pp5);
	\draw (sp8) -- (pp8);
	
	\def\compShift{-1.0}
	
	\foreach \i in {1,...,8} {
		\node[circle, fill=black, inner sep=1pt] (c\i) at ({(\i+10)*\sep}, {2*\h+ \compShift}) {};
		\node[circle, fill=black, inner sep=1pt] (cp\i) at ({(\i+10)*\sep}, {\h+ \compShift}) {};
	}
	
	\draw[dashed] (14.5*\sep, {2.3*\h+ \compShift}) -- (14.5*\sep, {0.7*\h+ \compShift});
	\node at (10.5*\sep, 1.5*\h+ \compShift) {$
		b_1 \circ b_2 =\textcolor{red}{N} \cdot \qquad \quad $};
	
	\draw (c1) -- (cp4);
	\draw (c2) -- (cp1);
	\draw (c3)-- (cp2);
	\draw (c4) to[bend right=40] (c5);
	\draw (c6) -- (cp5);
	\draw (c7) -- (cp8);
	\draw (c8) -- (cp7);
	\draw (cp3) to[bend left=40] (cp6);
	
	\end{tikzpicture}
	\caption{Example of graphical composition of two diagrams \(b_1,b_2 \in B_{N}(4,4)\). Identifying a closed loop (in red) results in multiplying the diagram by a scalar \(N \in \mathbb{C}\). We see that the composition \(b_1 \circ b_2\) remains within \(B_{N}(4,4)\).}
	\label{fig:WBA}
\end{figure}

When one introduces representation space $V_N^{\otimes m}\otimes \overline{V}_N^{\otimes n}$ every diagram from $B_{N}(m,n)$ can be viewed as a partially transposed permutation operator, where transposition is applied with respect to last $n$ systems, and when $N$ is the dimension of the mentioned representation space. In fact we can relate abstract diagrams from $B_{N}(m,n)$ with partially transposed permutation operators represented as diagrams, see Figure~\ref{S_to_WBA}. 
\begin{figure}[h!]
	\centering
	\[
	\left(
	\begin{tikzpicture}[scale=0.9, baseline={(current bounding box.center)}]
	\foreach \i in {1,...,6} {
		\node[circle, fill=black, inner sep=1pt] (t\i) at (\i,1) {};
		\node[circle, fill=black, inner sep=1pt] (b\i) at (\i,0) {};
		
		\node[above=2pt] at (t\i) {\tiny $\i$};
		\node[below=2pt] at (b\i) {\tiny $\i$};
	}
	
	\draw (t1) -- (b1);
	\draw (t2) -- (b6);
	\draw (t6) -- (b5);
	\draw (t5) -- (b4);
	\draw (t4) -- (b3);
	\draw (t3) -- (b2);
	\end{tikzpicture}
		\right)^{T_4 \circ T_5 \circ T_6}\quad =
	\quad
	\begin{tikzpicture}[scale=0.9, baseline={(current bounding box.center)}]
	\foreach \i in {1,...,6} {
		\node[circle, fill=black, inner sep=1pt] (t\i) at (\i,1) {};
		\node[circle, fill=black, inner sep=1pt] (b\i) at (\i,0) {};
		
		\node[above=2pt] at (t\i) {\tiny $\i$};
		\node[below=2pt] at (b\i) {\tiny $\i$};
	}
	
	\draw[dashed] (3.5,1.3) -- (3.5,-0.3);
	
	\draw (t1) -- (b1);
	\draw (t2) to[bend right=25] (t6);
	\draw (t3) -- (b2);
	\draw (t4) -- (b5);
	\draw (t5) -- (b6);
	\draw (b3) to[bend left=40] (b4);
	\end{tikzpicture}
	\]
	\caption{The diagram represent relation between diagrammatic representation of a cycle $\sigma=(26543)$ partially transposed with respect to systems $4,5,6$ and an element of the Walled Brauer algebra $B_N(3,3)$.}
	\label{S_to_WBA}
\end{figure}
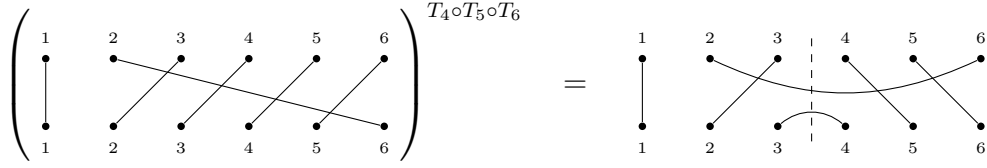

Collection of all such objects, with their linear combinations, gives the algebra of partially transposed permutation operators \( \mathcal{A}_{m,n}^{N} \).  In fact the algebra \( \mathcal{A}_{m,n}^{N} \) is equal to the commutant  sub-algebra in $ \End ( \Vmn ) $ of 
$ (U^{\otimes m}\otimes \overline{U}^{\otimes n})$
 and has been studied in many contexts~\cite{KOIKE198957,BENKART1994529,Halverson1996CharactersOT,PhysRevA.63.042111,Moz1}. 

As we pointed out in the introduction, the diagrammatic walled Brauer algebra $B_N(m,n)$ is well understood, while its matrix representation $\mathcal{A}_{m,n}^N$ still needs research.  The main reason for that is the following. The algebra $B_N(m,n)$ is non-semisimple for all $N< m+n$~\cite{Cox1}, so in principle it is not clear how to translate results from the diagrammatic algebra to its matrix representation $\mathcal{A}_{m,n}^N$ which is always semisimple.   More formally, in the general there is the following mapping $\rho_{ N, m , n }: B_N(m,n) \rightarrow \mathrm{End}(V_N^{\otimes m}\otimes \overline{V}_N^{\otimes n})$ which has for $N<m+n$ a non-trivial kernel $\mathrm{ker}(\rho_{ N, m , n })$. This kernel is in fact an ideal $I_N(m,n)\equiv \mathrm{ker}(\rho_{ N, m , n })$. By defining the following quotient
\begin{equation}\label{AquotB} 
\widehat{ B}_{ N } ( m,n):= B_N(m,n)/ I_N(m,n)
\end{equation}
it is clear that we have relation
\begin{equation}
\rho_{N,m,n}(B_{ N } ( m,n))=\widehat B_{ N } ( m,n) \cong \mathcal{A}_{m,n}^N
\end{equation}
The basis and generators for the ideal $I_N(m,n)$ can be found in~\cite{ANDERSEN_STROPPEL_TUBBENHAUER_2017}.
Even if we have a good description of the kernel in terms of diagrams, this does not immediately tell us about a basis in terms of Brauer representation theory data. To go to the representation theory picture, we need results from~\cite{BENKART1994529,Cox1,StollWerth2016, bowman2018cellularsecondfundamentaltheorem} -- which characterise the admissibility condition and their implications for matrix units. For applications, this abstract description must be made more explicit to give construction algorithms for the matrix units in the non-semisimple regime, which is an open problem. An intermediate step is to get explicit formulae for corrected dimensions. We make a step in this direction and find structural connections between these diagrams and Fock spaces (Hilbert spaces of infinitely many oscillators). This is discussed in further sections of this manuscript.

Following the discussion in \cite{BENKART1994529,KOIKE198957}, the representation space $V_N^{\otimes m}\otimes \overline{V}_N^{\otimes n}$ of the unitary group $U(N)$, where $U \in U(N)$ acts as $ U^{ \otimes m } \otimes \overline{U}^{ \otimes n } $, 
has a decomposition in terms of irreducible representations. By the double centralizer theorem, the  multiplicity of the irreducible representations is given by irreducible representations of the commutant algebra of $U(N)$ 
in $ \End ( \Vmn ) $,  denoted  as $ \cA^N_{m,n} $. For $ N \ge (m+n)$ this is the diagrammatic  Brauer algebra $B_N (m,n)$ and we may write 
\bea\label{eq:mixedSW1}
V_N^{\otimes m}\otimes \overline{V}_N^{\otimes n}=\bigoplus_{\gamma \in \BRT (m,n) } V_\gamma^{U(N)} \otimes V_{\gamma}^{B_N(m,n)}.
\eea  
The direct sum is labelled by combinatorial data $ \gamma $ which specifies  irreducible reps 
of $B_N(m,n)$. $\gamma$ consists of  an integer $ k $ ranging as 
$ 0 \le k \le \min ( m,n) $, along with a partition $ \gamma_+$  of $ (m-k)$ and a partition 
$ \gamma_- $ of $ (n-k)$. We thus write 
\begin{equation}
\label{eq:WBAlabelling}
\gamma=(k,\gamma_+\vdash (m-k), \gamma_-\vdash (n-k)) 
\end{equation}
We refer to the data $ \gamma$ as a Brauer representation triple (Brauer triple for short) and the set of Brauer triples for fixed $(m,n)$ is denoted 
 as $ \BRT (m,n)$.

In the non-semisimple regime $ N < (m+n)$, $   \cA^N_{ m,n}   $ is the quotient $ \widehat B_{ N } (m,n) $ of $ B_N ( m,n) $ and the double centralizer theorem implies 
\bea\label{eq:mixedSW} 
V_N^{\otimes m}\otimes \overline{V}_N^{\otimes n}=\bigoplus_{ \substack { \gamma \in \BRT ( m,n)   \\ c_1 ( \gamma_+ ) + c_1  ( \gamma_-  ) \le N }  } V_{ \gamma }^{U(N)} \otimes V_{\gamma}^{\widehat{B}_N(m,n)},
\eea

For convenience, we define 
\bea\label{defhBRTExcl}  
&& {\rm ht  } ( \gamma ) := \hbox{ height of } \gamma = c_1 ( \gamma_+ ) + c_1  ( \gamma_-  ) \cr 
&& \widehat{ \BRT} (  m,n , N )  := \{ \gamma \in \BRT (m,n) : {\rm ht  } ( \gamma )  \le N \} \cr 
&&  \Excl  (m,n, N ) := \{ \gamma \in \BRT (m,n) : {\rm ht } ( \gamma )  >  N \} 
\eea
For any $N$, there is a map  $\Gamma :  \Gamma( \gamma , N ) $ where $\Gamma( \gamma , N ) $ is the mixed Young diagram for $ V_\gamma^{U(N)}$, which determines the highest weight of the irrep $ V_{ \gamma }^{ U(N)} $. We will refer to elements $ \gamma \in \Excl ( m,n, N ) $  as $N$-excluded diagrams. The inequality on the height plays an important role in this paper 
\bea\label{finiteNconst} 
c_1 ( \gamma_+ ) + c_1  ( \gamma_-  )  \le N 
\eea 
and we refer to it as the ``finite $N$ constraint''. In the application of walled Brauer algebras to matrix invariants \cite{KR}, it is interpreted as a non-chiral stringy exclusion principle following qualitative similarities to a wide range of phenomena in AdS/CFT \cite{MalStrom}.  

For $ N \ge (m+n)$ all $\gamma \in \BRT (m,n)$ satisfy the constraint, and $ \widehat B_N ( m,n) = B_N ( m,n) $ so \eqref{eq:mixedSW} reduces to \eqref{eq:mixedSW1}. 
The decompositions of mixed tensor space in \eqref{eq:mixedSW} \eqref{eq:mixedSW1}  are referred to as  mixed Schur-Weyl duality.

\subsection{Bratteli diagrams } 
\label{sec:Brattdiag} 

In the stable large $N$ regime, $ \Dim ( V_{ \gamma}^{ B_N (m,n) } )$ is known 
and independent of $N$. We define $d_{m,n}(\gamma )$ to be this dimension from the stable range which is known \cite{BENKART1994529}\cite{KOIKE198957} to be 
\begin{equation}
\label{eq:dimWBA}
d_{m,n}(\gamma )=\frac{m!n!}{k!h(\gamma_+)h(\gamma_-)}.
\end{equation}
Here $h(\gamma_\pm)$ is the product of the Hook-lengths. There is a combinatorial interpretation of \eqref{eq:dimWBA} in terms of Bratteli diagrams, which we will find useful. 

The Bratteli diagram for  the walled-Brauer pair $ ( m , n ) $  is a graph consisting of vertices organised in layers labelled by levels ranging from $ L =0$ to $ L = (m+n) $. For $ 1 \le L \le m $, the vertices are associated with Brauer  triples for $ ( L , 0 ) $. For $ L =0$, the vertex is associated with the empty set and is the root of the graph.  For $ (m+1)  \le  L \le (m +n) $, the vertices are associated with the Brauer  triples of  the walled-Brauer pair $ ( m , L - m ) $.  Edges of the graph connect triples in adjacent layers, when these triples are  related by  combinatorial operations on Young diagrams, which are called Bratteli moves. As $L$ is increased from $0$ to $m$, a Bratteli move is the addition of a box to $ \gamma_+$ which produces a  valid Young diagram (weakly increasing row lengths). As $L $ is increased from $ (  m+1 ) $ to $(m+n)$, a Bratteli move is either the addition of a box to $ \gamma_-$  or the removal of a box from $ \gamma_+ $.   The dimension $d_{m,n}(\gamma )$ is equal to the number of paths in the Bratteli diagram of $B(m,n)$ from the root to a given Brauer triple  $\gamma \in \BRT (m,n)$.

It is also understood \cite{BENKART1994529,Cox1,StollWerth2016, bowman2018cellularsecondfundamentaltheorem}  how to  calculate the dimensions 
$ \widehat{d}_{ m,n , N } ( \gamma ) $ of the irreps $ V^{ \widehat{B}_{ N } (m,n) }$
in the non-semisimple regime $ N < (m+n)$ using a modification of Bratteli diagrams, which we will call coloured Bratteli diagrams (CBD) for $ (m,n,N)$. The CBD for $B_N(m,n)$ has 
nodes associated with $ \gamma $ obeying $ { \rm ht } (\gamma ) \le N$ coloured green and nodes associated with $ \gamma $ having $ { \rm ht } (\gamma ) > N$  coloured red. The green nodes in the final layer at $L=(m+n)$ correspond to irreps of $\widehat{B}_{ N } (m,n) $. 
Their dimension $ \widehat{d}_{ m,n , N } ( \gamma ) $  is equal to the number of paths from the root at $L=0$ which do not pass through red nodes at lower levels $L < (m+n)$. General formulae resulting from the application of this procedure are not available. Exposing hidden combinatorial structures, of physical interest, related to this algorithm and finding explicit formulae is the motivation which led to this paper. 

For a given $N < (m+n)$, some of the green nodes in the final layer will have no 
paths in the CBD which traverse red nodes at earlier layers. For these $ \widehat{d}_{ m,n , N } ( \gamma )   = d_{ m,n} ( \gamma )  $ and we refer to the set of these green nodes 
as $ \Unmod ( m,n , N )$. Other green nodes will have a subset of paths passing through red nodes at earlier stages and there is a modification of the dimension compared to the stable regime : 
\bea 
 \widehat{d}_{ m,n , N } ( \gamma )   = d_{ m,n} ( \gamma )  - \delta_{ m , n , N } 
\eea
where $ \delta_{ m,n< N } $ is a positive integer counting the number of paths arriving at $ \gamma $ from the root after passing through a red node. This set of 
nodes in the final layer at $ L = (m+n)$ is denoted $ \Mod ( m, n , N ) $. To summarise, 
each $N < (m+n) $ in the non-semisimple regime defines a partition of $ \BRT (m,n)$ 
\bea\label{DisjBRTmnN0}  
&& \BRT (m,n) = \Excl ( m,n , N ) \sqcup \Unmod (  m , n , N )  \sqcup \Mod (  m , n , N )  \cr 
&& \gamma \in \Excl ( m,n , N )  :  {\rm ht }  ( \gamma ) > N \cr 
&& \gamma \in \Unmod (  m , n , N )  : {\rm ht }  ( \gamma ) \le  N \hbox{ and } \widehat{d}_{ m,n} ( \gamma ) =  d_{ m,n} ( \gamma ) \cr 
&& \gamma \in \Mod (  m , n , N ) : {\rm ht }  ( \gamma ) \le  N \hbox{ and }   \widehat{d}_{ m,n} ( \gamma ) =  d_{ m,n} ( \gamma ) - \delta_{ m,n , N } ( \gamma ) \hbox{ with}~ \delta_{ m,n , N } ( \gamma ) >0  \cr 
&& 
\eea 
For convenience, we also introduce 
\bea\label{BRThat}  
\widehat{\BRT} ( m , n ) := \BRT (m,n) \setminus \Excl ( m,n , N ) = \Unmod (  m , n , N )  \sqcup \Mod (  m , n , N )
\eea 
The elements of $ \widehat{\BRT} ( m , n ) $ are in 1-1 correspondence with the irreps of 
$ \widehat{B}_{m,n} ( N ) $.

The coloured Bratteli diagram for $ B_{N} (m,n) = B_{2 } ( 3,2 ) $ is shown in 
 Figure~\ref{Fig0}.
	\begin{figure}[ht!] 
		\centering
		\includegraphics[scale=0.3]{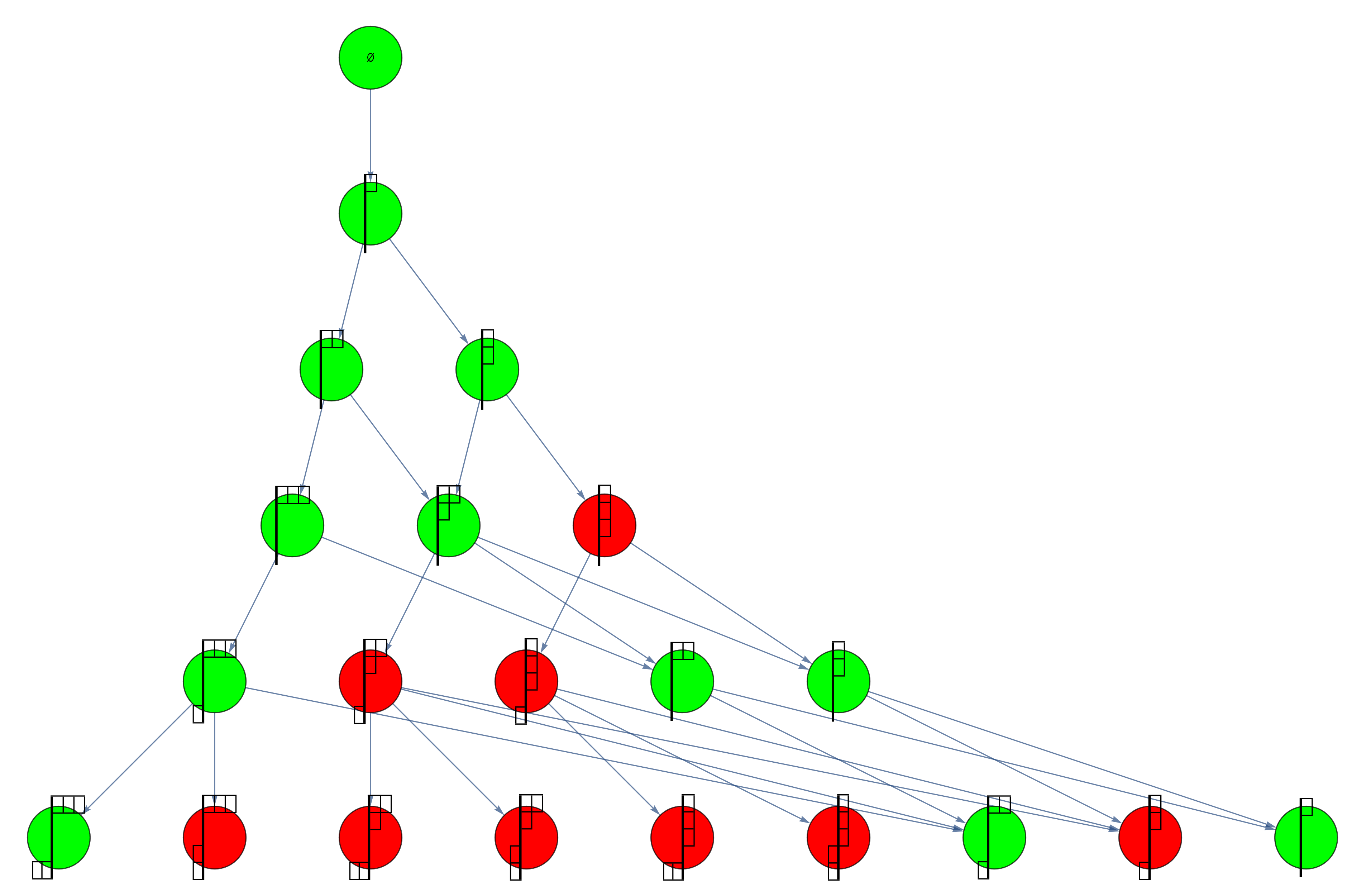}
		\caption{The graphic presents the coloured Bratteli diagram (CBD) for $m=3,n=2$ and $N=2$. The nodes in red are associated with Brauer representation triples which do not obey the finite $N$ constraint \eqref{finiteNconst}, while green nodes do obey the constraint.  } 
		\label{Fig0}
	\end{figure}
	
In section ~\ref{sec:introRBDB} we will introduce the definition  of restricted Bratteli diagrams (RBD)
for $B_N(  m,n) $ which will allow the efficient calculation of $ \delta_{ m,n, N } ( \gamma ) $ and hence $ \widehat{d}_{ m,n} ( \gamma ) $, using as input the dimension formulae for different sets of $ d_{ m,n} ( \gamma' ) $ from the stable regime. The RBD for $B_2(3,2)$ is shown in Figure~\ref{Gen_rest}. 
	\begin{figure}[ht!] 
		\centering
		\includegraphics[scale=0.3]{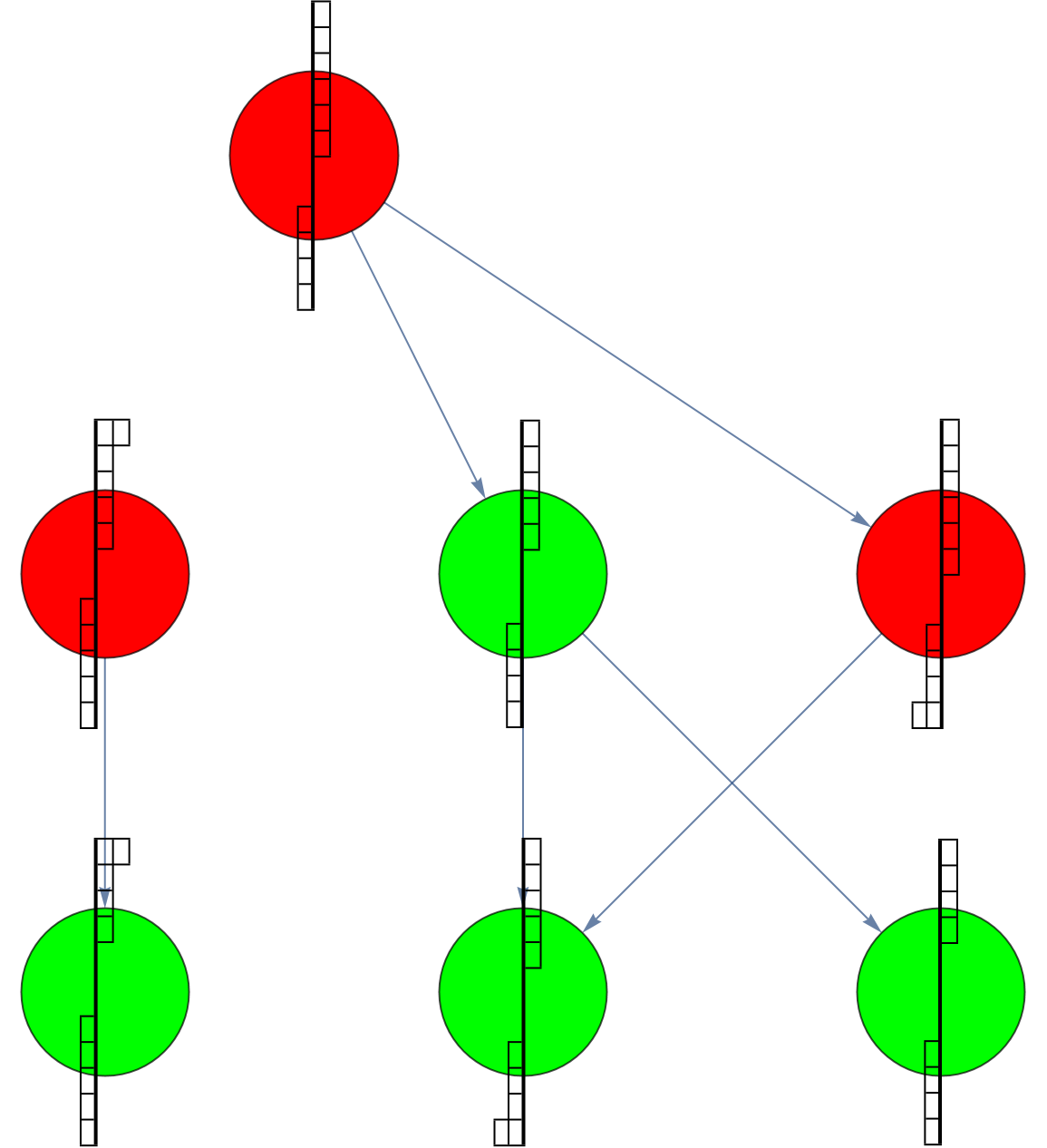}
		\caption{The graphic presents the   restricted Bratteli diagram (RBD) for $m=n=6$ and $N=9$. The red nodes are associated with Brauer triples which do not obey the finite $N$ constraint and admit paths to green nodes in the final layer. The green nodes  are associated with triples labelling irreps with modified dimensions. } 
		\label{Gen_rest}
	\end{figure}

\section{Relating Kernels : symmetric group algebras and Brauer algebras to tensor space}
\label{sec:relkern} 

Let us denote by $B(m,n)$ a linear space over $\mathbb{C}$ spanned by the  walled Brauer diagrams. 
The goal of this section is to characterise the kernel of  the map $ \rho_{N, m , n  }$ 
from $B(m,n)$ to  $\operatorname{End}(V_N^{\otimes m}\otimes \overline{V}_N^{\otimes n})$.
The main result in this section is also understood in the mathematical literature (Proposition 12 of \cite{StollWerth2016}). The present explanation and proof is close to the applications of walled Brauer algebras to matrix invariants (\cite{KR,KRHol}) and to quantum information tasks
 (\cite{Yongzhang2021permutation}\cite{MozJPA}) where partial transposition plays an important role. 

Let us start by considering the group algebra $\mathbb{C}[S_{m+n}]$ with corresponding irreducible matrix units $Q^{\lambda}_{IJ}$, associated with an irrep $\lambda$, where $1\leq I,J\leq \dim_{m+n}(\lambda)$. These operators span every irreducible block labelled by $\lambda \vdash (m+n)$ and satisfy the following matrix multiplication rules:
\begin{align}
\label{MatUnits_S_n}
Q_{IJ}^{\lambda}Q_{I'J'}^{\lambda'}=\delta^{\lambda \lambda'}\delta_{JI'}Q^{\lambda}_{IJ'}.
\end{align}
For the algebra  $\mathbb{C}[S_{m+n}]$ we can define a map 
\begin{align}
\label{eq:mapRhom+n0}
\rho_{N,m+n}:\mathbb{C}[S_{m+n}] \longrightarrow \operatorname{End}(V_N^{\otimes (m+n)})
\end{align} 
that allows us to find matrix representations of the considered irreducible matrix units:
\begin{align}
\label{eq:mapRhom+n}
\rho_{N,m+n}\left(Q^{\lambda}_{IJ}\right)\in \operatorname{End}(V_N^{\otimes (m+n)}).
\end{align}
In the regime when $N<m+n$, the map $\rho_{N,m+n}$ has a non-trivial kernel, i.e. we have $\operatorname{ker}(\rho_{N,m+n})\neq \{0\}$. It means there exists a  non-empty set of $Q^\lambda_{IJ}$ which is mapped to zero matrix under action of $\rho_{N,m+n}$. It happens a Young diagram $\lambda$ satisfies $c_1(\lambda)>N$.

Now, let us consider another map $P^t_{m,n}$ performing partial transposition with respect to last $n$ systems
\begin{align}
\label{eq:mapPt}
P^t_{m,n}:\mathbb{C}[S_{m+n}]\longrightarrow B(m,n),
\end{align}
where the $B(m,n)$ is a vector space spanned by all walled Brauer diagrams.
Then we have for all irreducible units $Q^\lambda_{IJ}$
\begin{align}
\mathbb{C}[S_{m+n}] \ni Q^\lambda_{IJ} \longrightarrow{\text{$P^t_{m,n}$}}(Q^\lambda_{IJ}) \in B(m,n).
\end{align} 

Having the mapping between algebras $\mathbb{C}[S_{m+n}]$ and $B(m,n)$, we can ask about matrix representation of the image of the $P^t_{m,n}$ on the space of $\operatorname{End}(V_N^{\otimes m}\otimes \overline{V}_N^{\otimes n})$, so in fact matrix elements of $\mathcal{A}_{m,n}^N$. Similarly to~\eqref{eq:mapRhom+n} we define a following map
\begin{align}
\label{eq:mapRhoBmn}
&\rho_{N,m,n}: B(m,n) \longrightarrow \operatorname{End}(V_N^{\otimes m}\otimes \overline{V}_N^{\otimes n}),\\
&B(m,n) \ni P^t_{m,n}(Q^\lambda_{IJ}) \longrightarrow{\text{$\rho_{N,m,n}$}}\left(P^t_{m,n}(Q^\lambda_{IJ})\right) \in \mathcal{A}_{m,n}^N.
\end{align}
Notice that the map $\rho_{N,m,n} \circ P^t_{m,n}$ does not map irreducible units of $\mathbb{C}[S_{m+n}]$ to irreducible units of the algebra $\mathcal{A}_{m,n}^N$. However, the result of its action still spans the whole linear space.
The map $\rho_{N,m,n}$, as it was for the map from~\eqref{eq:mapRhom+n},  also has a non-zero kernel. The kernel of this map is strongly connected with irreducible matrix units of the algebra $B_N(m,n)$. Namely, similarly to irreducible matrix units for $\mathbb{C}[S_{m+n}]$, we can define set of operators $Q^{\gamma}_{IJ}$ spanning every irreducible space labelled by $\gamma=(k,\gamma_+,\gamma_-)$.  These operators satisfy the analogous matrix multiplication rules to~\eqref{MatUnits_S_n}. In the semisimple regime, when $N\geq (m+n)$ the quotient algebra $\widehat{B}_N(m,n)$ is equal to $B_N(m,n)$ and  the map $\rho_{N,m,n}$ has a  trivial kernel. 
In the non-semisimple regime, when $N<(m+n)$, the kernel of  $\rho_{N,m,n}$ is non-trivial and composed of objects of two kinds. The objects of the first kind  are all matrix  units $Q^{\gamma}_{IJ}$, where $\gamma=(k,\gamma_+,\gamma_-)$ corresponds  to dropped $\gamma'$s. This case happens when the condition $c_1(\gamma_+)+c_1(\gamma_-)\leq N$ is not fulfilled. The objects of the second kind are matrix units  $Q^{\gamma'}_{ab}, Q^{\gamma'}_{ib}, Q^{\gamma'}_{aj}$, where $\gamma'$ corresponds to the irreps whose dimensions are modified, and $a,b$ denote inadmissible paths in the Bratteli diagram while $ i,j$ label admissible paths. The maps discussed above satisfy the commutativity relations depicted in Figure~\ref{fig2:sminusr} for all $\sigma \in \mathbb{C}[S_{m+n}]$.
\begin{figure}[h!]
\centering
\begin{tikzpicture}[scale=1.5, baseline=(current bounding box.center),
>={Stealth}, 
every node/.style={font=\small}]

\node (V) at (0,2) {$\sigma \in \mathbb{C}[S_{m+n}]$};
\node (W) at (5,2) {$\rho_{N,m+n}(\sigma) \in \operatorname{End}\left(V_N^{\otimes (m+n)}\right)$};
\node (V') at (0,0) {$P^t_{m,n}(\sigma)\in B(m,n)$};
\node (W') at (5,0) {$\rho_{N,m,n}(P^t_{m,n}(\sigma))\in \operatorname{End}\left(V_N^{\otimes m}\otimes \overline{V}_N^{\otimes n}\right)$};

\draw[->] (V) -- (W) node[midway, above] {$\rho_{N,m+n}$};
\draw[->] (V) -- (V') node[midway, left] {$P^t_{m,n}$};
\draw[->] (W) -- (W') node[midway, right] {$\mathrm{id} \otimes T_{m+1}\otimes \cdots \otimes T_{m+n}$};
\draw[->] (V') -- (W') node[midway, below] {$\rho_{N,m,n}$};

\end{tikzpicture}
\caption{Commutativity diagram for algebra elements $\sigma \in \mathbb{C}[S_{m+n}]$ under action of maps $\rho_{N,m+n},P^t_{m,n}$, and $\rho_{N,m,n}$ given through~\eqref{eq:mapRhom+n0}, \eqref{eq:mapPt}, and~\eqref{eq:mapRhoBmn} respectively.}
\label{fig2:sminusr}
\end{figure}
Notice that the both maps~\eqref{eq:mapRhom+n0} and~\eqref{eq:mapRhoBmn} when having non-trivial kernel give us useful identities between elements of the respective algebras in $\operatorname{End}(V_N^{\otimes (m+n)})$ and $\operatorname{End}(V_N^{\otimes m}\otimes \overline{V}_N^{\otimes n})$. Namely, we have identities of the form $\rho_{N,m+n}(Q^\lambda_{IJ})\equiv 0$ and $\rho_{N,m,n}(Q^\gamma_{IJ})\equiv 0$.

Having these preliminary considerations, we can formulate the following:
\begin{lemma}
	For the maps $\rho_{N,m+n},\rho_{N,m,n}$ from~\eqref{eq:mapRhom+n0},~\eqref{eq:mapRhoBmn} respectively, the following equality holds:
	\begin{align}
	\operatorname{ker}(\rho_{N,m+n})=\operatorname{ker}(\rho_{N,m,n}).
	\end{align}
\end{lemma}

\begin{proof}
First we prove that for any $A\in \mathbb{C}[S_{m+n}]$ which is in the kernel of the standard  map $\rho_{N,m +n}$ from~\eqref{eq:mapRhom+n0} we can deduce it is also in the kernel of the map from~\eqref{eq:mapRhoBmn}.
Let us take $A \in \mC [ S_{ m + n } ] $ of its the most general form:
\begin{equation}
	A = \sum_{ \sigma \in S_{m+n} } A_{ \sigma } \sigma, \qquad \forall \sigma \in S_{m+n} \quad A_{\sigma}\in\mathbb{C}.
\end{equation}
Let us assume $A$ is in the kernel of the standard  map $\rho_{N,m +n}$ from~\eqref{eq:mapRhom+n0}.
This means that 
\begin{equation}
\rho_{ N,m +n } ( A ) = 0.
\end{equation}
The above is equivalent in saying that 
\begin{equation}
\label{vanX} 
\sum_{ \sigma \in S_{m+n} } A_{ \sigma } 
\langle e^{ j_1 } \otimes e^{ j_2 }\otimes \cdots \otimes e^{ j_{m+n} } | \rho_{N,m +n}(\sigma) | e_{ i_1 } \otimes e_{ i_2 } \otimes \cdots \otimes e_{ i_{ m+n } } \rangle = 0.
\end{equation}
Multiply this this equation with $ ( X_{1} )^{ i_1}_{ j_1 } \cdots  ( X_{ m+n} )^{ i_{ m+n} }_{ j_{ m+n} } $ where these variables are matrix elements of linear operator  variables $X_a$,  $X_{a} | e_i \rangle =  \sum_{ j }  ( X_a )_i^j | e_j \rangle $, for $1\leq a \leq m+n$. The  coefficients  $ ( X_{1} )^{ i_1}_{ j_1 } \cdots  ( X_{ m+n} )^{ i_{ m+n} }_{ j_{ m+n} } $ can be thought as the  coefficients in the expansion of a vector $X_1 \otimes X_2 \otimes \cdots \otimes X_{ m+n}$ in $\operatorname{End} ( V_N^{ \otimes (m + n) } ) $ in terms of the basis  vectors $| e_{ i_1 } \otimes e_{ i_2 } \otimes \cdots \otimes e_{ i_{ m+n } } \rangle \langle e^{ j_1 } \otimes e^{ j_2 }\otimes  \cdots \otimes e^{ j_{m+n} } |$. 
We can express the vanishing in \eqref{vanX} as 
\bea\label{VanMultTr}  
\sum_{ \sigma \in S_{m+n} } A_\sigma \tr_{ V_N^{ \otimes ( m +n) } } \left[( X_{ 1} \otimes X_2\otimes \cdots \otimes X_{ m+n} ) \rho_{N,m+n}(\sigma)\right] 
= 0.
\eea
Now exploiting fact that the partial transposition does not change value of trace, we apply it to \eqref{VanMultTr} with respect to last $n$ systems. Together with the commutativity property described in Figure~\ref{fig2:sminusr} we obtain the equality for any $A \in \mC [ S_{ m +n } ] $, and any value of $M$:
\bea\label{PartTranspEq} 
\begin{split}
&\sum_{ \sigma \in S_{m+n} } A_\sigma \tr_{ V_M^{ \otimes ( m +n) } } \left[( X_{ 1} \otimes X_2\otimes \cdots \otimes X_{ m+n} )\rho_{N,m+n}(\sigma)\right]  \\
&= \sum_{ \sigma \in S_{m+n} } A_\sigma \tr_{ V_M^{ \otimes ( m +n) } }\left[( X_{ 1} \otimes X_2 \otimes  \cdots \otimes X_{m} \otimes X^T_{ m+1}  \otimes \cdots \otimes X^T_{ m+n} ) \rho_{N,m,n}(P^t_{  m, n }  (  \sigma ))\right].
\end{split}
\eea 
Specialising the above to $ M = N $ and using \eqref{VanMultTr} 
\bea 
&& \sum_{ \sigma \in S_{m+n} } A_{\sigma}\tr_{ V_N^{ \otimes ( m +n) } } \left[( X_{ 1} \otimes X_2 \otimes  \cdots \otimes X_{m} \otimes X^T_{ m+1}  \otimes \cdots \otimes X^T_{ m+n} ) \rho_{N,m,n}(P^t_{  m, n }  (  \sigma ))\right]  
= 0.  \cr 
&&
\eea
This means that the element $\sum_{ \sigma \in S_{m+n} } A_{ \sigma } P^t_{m,n} ( \sigma )$ belongs to $B(  m , n )$ and maps to zero under the map $\rho_{ N, m , n } :  B( m , n )  \rightarrow  \operatorname{End} ( V_N^{\otimes m } \otimes \overline{V}_N^{ \otimes n } )$.
	
The converse argument taking us from the kernel of $ \rho_{ N, m , n } $ to the kernel of $ \rho_{ N, m +n} $ proceeds in the same way. We start from an element $B \in B_d( m , n ) $ of the form
\bea 
B = \sum_{b\in B_d( m , n )} B_{ b } b, \qquad \forall b \in B_d( m , n ) \quad B_{b}\in\mathbb{C}
\eea
in the kernel of 
$\rho_{N, m , n} $ to obtain the element $ P^t_{ m , n } ( B ) \in \mC[S_{ m + n }]$ which is in the kernel of $ \rho_{N,m +n} $ because of the equality \eqref{PartTranspEq}.
\end{proof}

We illustrate above by the following examples. 
\begin{example}
	Consider the case when $l=1$, this means we have $N=m+n-1$. There is only one irrep $\gamma$ in the kernel, which is
	\begin{align}
	\gamma=(k=0,[1^m],[1^n]), \quad d_{m,n,N}(\gamma)=1,
	\end{align}
	with the corresponding matrix unit
	\begin{align}
	Q^{\gamma}=P^t_{m,n}(Q^{[1^{m+n}]}).
	\end{align}
	In this particular case, it is easy to see a one-to-one correspondence between dropped irreducible matrix units in $\mathcal{A}_{m,n}^N$ and $\mathbb{C}[S_{m+n}]$. 
	\begin{align}
	\rho_{N,m+n}(Q^{[1^{m+n}]})=0 \quad \text{and} \quad \rho_{N,m,n}(P^t_{m,n}(Q^{[1^{m+n}]}))=\rho_{N,m,n}(Q^{\gamma})=0.
	\end{align}
	This shows that the both kernels are equal and one-dimensional.
\end{example}

\begin{example}
	Let us consider now more complicated case when $l=2$. This means we have $N=m+n-2$. In this case vanishing matrix units of the algebra $\mathbb{C}[S_{m+n}]$ under action of the map $\rho_{N,m+n}$ are those associated with the following partitions
	\begin{align}
	\lambda_1=[1^{m+n}], \quad \lambda_2=[2,1^{m+n-2}]
	\end{align}
	with $d_{m+n}(\lambda_1)=1$, and $d_{m+n}(\lambda_2)=m+n-1$. It is clear that for $\lambda_1$ we have only one matrix unit $Q^{\lambda_1}$, while for $\lambda_2$, we have $(m+n-1)^2$ matrix units $Q^{\lambda_2}_{IJ}$. Shortly, we can write 
	\begin{align}
	\rho_{N,m+n}(Q^{\lambda_1})=0, \quad \forall I,J \ \rho_{N,m+n}(Q^{\lambda_2}_{IJ})=0.
	\end{align}
	This means that the kernel of $\rho_{N,m+n}$ is $(m+n-1)^2+1$ dimensional.
	Applying first the map $P^t_{m,n}$ to every matrix units and then the map $\rho_{N,m,n}$, we get 
	\begin{align}
	&\rho_{N,m,n}\left(P^t_{m,n}(Q^{\lambda_1})\right)=0,\cr
	\forall \ I,J \ &\rho_{N,m,n}\left(P^t_{m,n}(Q^{\lambda_2}_{IJ})\right)=0.
	\end{align}
	From the walled Brauer algebra perspective there are three irreps that must be dropped
	\begin{align}
	&\gamma_1=(k=0,[1^m],[1^n]), \gamma_2=(k=0,[2,1^{m-2}],[1^n]), \gamma_3=(k=0,[1^m],[2,1^{n-2}]),\cr
	& d_{m,n}(\gamma_1)=1, \  d_{m,n}(\gamma_2)=m-1, \  d_{m,n}(\gamma_3)=n-1.
	\end{align}
	and according to Section~\ref{Sec:l=2} one irrep whose dimension must be modified
	\begin{align}
	\gamma_4=(k=1,[1^{m-1}],[1^{n-1}]),\quad d_{m,n}(\gamma_4)=mn.
	\end{align}
	The correction in this case is simple and equal to $\delta=1$, giving $ \widehat{d}_{m,n,N}(\gamma_4)= d_{m,n}(\gamma_4)-1=mn-1$. It means the only one path in the Bratteli diagram is inadmissible, let us denote it by $a$. Then matrix units of the second kind living in the kernel are of the form:
	\begin{align}
	&Q^{\gamma_4}_{aa}\quad \text{1 operator},\cr
	&Q^{\gamma_4}_{ia}\quad \text{$(mn-1)$ operators},\cr
	&Q^{\gamma_4}_{aj}\quad \text{$(mn-1)$ operators}.
	\end{align} 
	This given in total $2mn-1$ matrix units that live in the kernel. The remaining operators $Q^{\gamma_4}_{ij}$, for $i,j\neq a$ are not in the kernel. The dimension of the kernel for $\rho_{N,m,n}$ is equal then to $2mn-1+(m-1)^2+(n-1)^2+1=2mn+(m-1)^2+(n-1)^2 = ( m+n -1)^2 +1$ which is equal to the dimension of  the kernel of $\rho_{N,m+n}$.
\end{example}

\section{The restricted Bratteli diagrams : definition and properties  }
\label{sec:introRBDB}

An irreducible representation of the walled Brauer algebra $ B_N ( m ,n ) $  is 
  specified by a triple, consisting of an integer 
 $ 0 \le k \le \min ( m , n ) $ where $ \min( m , n ) $ is the smaller of $ m , n $, along with a partition $ \gamma_+ $ of $ (m-k)$ and a partition $ \gamma_-$ of $ ( n - k)  $. These partitions can be visualised as Young diagrams with row lengths $ r_i ( \gamma_+ ) , r_i ( \gamma_- ) $ obeying $ r_i ( \gamma_{\pm} ) \ge  r_{ i+1} ( \gamma_{\pm} ) $ and $ \sum_{ i } r_i ( \gamma_+ ) = ( m -k ) , \sum_i r_i ( \gamma_- ) = ( n -k ) $.   We will thus write $ \gamma = ( k , \gamma_+ \vdash ( m-k)  , \gamma_-\vdash ( n-k)   ) $ for an irreducible representation. $ \gamma_+ , \gamma_-$ are naturally associated with Young diagrams with $(m-k)$ and $ ( n-k)$ boxes respectively.  We denote the length of the first columns as $ c_1 ( \gamma_+ ) $ and $ c_1 ( \gamma_-)$.  It is also conventional to associate a Young diagram with positive row lengths given by the parts of $ \gamma_+$ and negative row lengths given by the parts of $ \gamma_-$ 
 \bea\label{mixedYD}  
&&  R_i ( \gamma ) = r_i ( \gamma_+ )  ~~\hbox{ for } ~~    1 \le i \le c_1 ( \gamma_+ ) \cr 
&&  R_{ N -i +1 }  ( \gamma ) = - r_i ( \gamma_- )  ~~ \hbox{ for } ~~     1\le i \le c_1 ( \gamma_-)  
 \eea
 
The representation theory of $ B_N ( m , n ) $ and its Schur-Weyl duality with $ U(N) $ acting in the mixed tensor representation $ \Vmn  $  is best understood for $ N > ( m + n ) $.  In this range, the dimension of the irreducible representation \cite{BENKART1994529, bowman2018cellularsecondfundamentaltheorem} is given by~\eqref{eq:dimWBA}.  The dimension given in \eqref{eq:dimWBA} is independent of $N$ while $N$ varies in the range $ N \ge ( m +n)$, and  we will refer to this  as a large $N$ stability property. 
As discussed in section  \ref{sec:Brattdiag} this formula is equal to a counting of paths in the Bratteli diagram for $ B_N ( m,n)$, i.e. the expression for $ d_{m,n}( \gamma)  $ counts all paths in the Bratteli diagram of $B(m,n)$ from the root to a given irrep $\gamma$. 
The diagrams share the stability property of the dimension formula, i.e. they are independent of $N$ in the range $ N > (m+n)$. 

For $ m + n < N $, as discussed in \eqref{AquotB},  there is a semisimple quotient $ \cA^N_{ m , n  } $ which is isomorphic to the commutant algebra of $ U^{ \otimes m} \otimes \overline{U}^{ \otimes n } $  in $\Vmn $.  Irreps of  $ \cA^N_{ m , n  } $ are restricted by the condition $ c_1 ( \gamma_+ ) + c_1 ( \gamma_-) \le N$. A further important feature is that a subset of the irreps of $ \cA^N_{ m , n } $, labelled by $ \gamma $ obeying this condition, have a dimension which is smaller than \eqref{eq:dimWBA}.  It is also known~\cite{BENKART1994529,StollWerth2016,bowman2018cellularsecondfundamentaltheorem} that the modified dimensions can be obtained by counting paths in the Bratteli diagram of $ B ( m , n ) $ which  have the following properties: \\
{\bf Properties of  paths contributing to modified dimensions}\label{Cordimprop} 
\begin{enumerate} 

\item They terminate at $ \gamma$ in the last layer. 

\item They do not pass through any Young diagrams $ \gamma'$ at smaller levels with $ c_1 ( \gamma'_+ ) + c_1 ( \gamma'_- ) > N $.  We will refer to these as $N$-excluded diagrams, or $N$-excluded Brauer triples. 

\end{enumerate} 
 Equivalently 
\bea\label{dimcor}  
 \widehat{d}_{m,n,N} ( \gamma )=d_{m,n} ( \gamma )  - \delta_{m,n,N} (  \gamma ),  
\eea
where  $ \delta_{m,n,N} ( \gamma ) $ is the number of paths terminating at $ \gamma $ in the last layer of the $B_N(m,n)$ Bratelli diagram and  passing through one or more nodes associated $N$-excluded diagrams. 

For the non-semisimple regime  $ N < m+n $, we define the  coloured Bratteli diagrams (CBD) of 
$B_N( m , n  )$ as follows. 

\begin{definition}{Definition  of coloured Bratteli diagrams of $B_N( m , n  )  $.   } 
\label{CBDBNmn} 
\begin{enumerate} 

\item They have $ ( m+n+1)$ layers.  

\item They contain all the same nodes and paths as the $B_M(m,n)$ Bratelli-diagram from the stable regime $M \ge (m+n)$. 

\item The nodes obeying $ c_1 ( \gamma_+ ) + c_1 ( \gamma_- ) \le N $ are coloured green, while 
the nodes obeying $ c_1 ( \gamma_+ ) + c_1 ( \gamma_- ) > N$ are coloured red. The Brauer representation triples  $ \gamma$ with  $ c_1 ( \gamma_+ ) + c_1 ( \gamma_- ) > N$ are called $N$-excluded nodes. 

\end{enumerate} 
\end{definition} 

It is evident from Definition  \ref{Cordimprop} that calculating $ \delta_{m,n,N} (  \gamma ) $ and $d_{m,n} ( \gamma )$ can be accomplished by counting paths in the CBD of $ B_N ( m,n)$. 
A more efficient calculation of these quantities can be accomplished using a simpler diagram containing a subset of nodes and links of the CBD of   $ B_N( m , n ) $.
The {\bf Restricted Bratelli diagram (RBD)  of $ B_N( m , n) $ } is thus defined by the following 
properties. 

\begin{definition}{Definition of restricted Bratteli diagrams of  $B_N(m,n)$ }
\label{DefRBDB}
These diagrams have a sequence of layers of depths $d$  starting from $d=0$. 

\begin{enumerate}

\item At $ d =0$,  they only contain the subset of nodes of the CBD of  $ B_N ( m , n ) $ at the layer $ L = ( m+n)$ which are green. 

\item At depth $d>0$, they contain the red  nodes at $ L = ( m+n) -d $ of the CBD of $ B_N ( m,n)$ which connect to green nodes at $ L = (m+n)$, or $d=0$. 

\item The  RBD of  $B_N(m,n)$ also contains green nodes at $ d >0 $ which appear along paths from reds nodes at $d' >d $ to the green nodes at $d=0$. Green nodes in CBD of  $B_N ( m , n ) $ which only connect to red nodes at $ d' < d $ are removed to produce the simpler RBD. The counting of paths from such green nodes in $ \CBDB_N ( m , n ) $ can be recovered from  the RBD of $B_N(m,n)$ by applications of the stable range dimension formula \eqref{eq:dimWBA}. This will be illustrated in examples. 

\end{enumerate} 
\end{definition} 

The RBD of $B_N ( m  , n  ) $ are used, alongside the stable-range dimension formula \eqref{eq:dimWBA} to calculate the finite $N$ dimensions and dimension modifications.  They do not include paths that are counted by the stable range dimension formula, hence their simplicity compared to the  CBD  $B_N(m,n)$. 

We will show, in Observation \ref{ObsdepRBDB}  in section \ref{sec:deepest} that for $ m +n = N  - l $, these diagrams have depth $l$, i.e. the layers are labelled by $ 0 \le d \le (l-1)$.  When $ N , m , n $ are large, while $ l $ is small, these diagrams are very useful. We describe the
Mathematica code used for  constructing the RBD of $B_N ( m , n) $ in Appendix \ref{appsec:Mtca}. 
We also obtain general formulae for the corrected dimensions at $ l = 2,3,4$. We test these formulae using the understanding of the  identity  of the dimensions of the kernels of the maps $ 
\rho_{N, m,n} :  B_N ( m , n ) \rightarrow \End (\Vmn   )   $ and $ 
\rho_{ N,m+n} :  \mC ( S_{ m+n} )  \rightarrow \End ( V_N^{ m+n}    )   $ 
derived in Section \ref{sec:relkern}. Results are collected in section~\ref{sec:DimExamples}.

\section{Modified dimensions for $B_N(m,n)$ with $N=m+n-l$ and small  values of $l$}
\label{sec:DimExamples}

In this section, starting from a few explicit examples for fixed $m$, we deliver formulas for modified dimensions for irreps of $B_N(m,n)$, where $N=m+n-l$ with $l=2,3,4$. From our examples we excluded the case of $l=1$ since we have only one $N$-excluded rep in the 
final layer of the CBD of $B_N(m,n)$, and there are no irreps with modified dimension.  The algorithm for deriving the modifications is the following:
\begin{enumerate}
	\item We work in the stable regime  for the RBD of the algebra $B_N(m,n)$ with given $N=n+m-l$. In the stable regime,the RBD are  unchanged when $m,n$ are increased. This regime is $m,n \ge (2l-3) $ which is observed from the computation of the RBD using the mathematica code in Appendix A and proved in section \ref{Sec:countingReds}. 
	\item Identifying irreps $\gamma_i=(k,\gamma^{(i)}_+\vdash m-k,\gamma_-^{(i)}\vdash n-k)$ of $B_N(m,n)$ requiring dimension modification and compute their dimension $d_{m,n}(\gamma_i)$
	from the stable range formula \eqref{eq:dimWBA} 
These irreps have paths in the CBD or the RBD of $B_N(m,n)$ which terminate at $\gamma_i$ after passing through $N$-excluded nodes. 
 Then the modifications to the dimensions are
	\bea 
&&	\widehat{d}_{m,n,N}(\gamma_i)= d _{m,n}(\gamma_i)  - \delta_{m,n,N}(\gamma_i), \quad \text{for}\quad i=1,\ldots, K, \cr 
&&	\hbox{ where } K = |\Mod ( m,n, N ) | \hbox{ is the number of irreps $ \gamma $ with modified dimensions }\cr 
&& 
	\eea
and $\delta_{m,n,N}(\gamma_i)$ is  the number of paths reaching $ \gamma_i$ through 
$N$-excluded diagrams in the CBD of $B_N(m,n)$.  The number $K$ denotes the number of irreps in the last layer of the RBD of $ B_N ( m,n) $. 
	
	\item To compute numbers $\delta_{m,n,N}(\gamma_i)$ for $i=1,\ldots,K$  more efficiently using the RBD, we first calculate dimensions $d_{m,n}(\gamma')$ for the $N$-excluded Brauer representation triples $ \gamma'$  in the RBD, which are  connected to the given $\gamma_i$. Next, we calculate the number of inadmissible paths going through the  $\gamma'$ and arriving at 
$ \gamma_i$.  These numbers are used to obtain the  dimension corrections $\delta_{m,n,N}(\gamma_i)$. We iterate this procedure for all $i=1,\ldots,K$.
\end{enumerate}
To check the correctness of the number derived from the above algorithm we can proceed as follows. We introduce the following quantity:
\begin{align}
\label{eq:deficit}
\Delta(m,n;N):=\sum_{i=1}^k\left[(d_{m,n}(\gamma_i))^2-(d_{m,n}(\gamma_i)-\delta_{m,n,N}(\gamma_i))^2\right],
\end{align}
which measures the sums of the differences  between the squares of 
$ d_{ m,n} ( \gamma_i ) $ and $ \widehat{d}_{m,n,N} ( \gamma_i ) = d_{ m,n} ( \gamma_i )  - \delta_{ m,n,N} ( \gamma_i )  $.  The dimension of $\mathcal{A}^N_{m,n}$ is then given by
\begin{align}
\label{eq:Dim1}
\dim ( \mathcal{A}^N_{m,n})&&= - \Delta ( m, n ; N ) + \sum_{ \substack { \gamma \in \BRT( m,n )   \\ 
		c_1(\gamma_+)+c_1(\gamma_-)\leq N  }  }
( d_{m,n}(\gamma) )^2 .  \cr 
&&= - \Delta ( m, n ; N ) + \sum_{ \gamma \in \widehat{\BRT}( m,n )   }
( d_{m,n}(\gamma) )^2 . 
\end{align}
The last line uses the notation introduced in \eqref{defhBRTExcl}.
The dimension of $\mathcal{A}^N_{m,n}$  is also given by dimensions of the irreps $\lambda$ of the symmetric group $S_{m+n}$:
\bea 
\label{eq:Dim2}
\dim ( \mathcal{A}^N_{m,n})=\sum_{ \substack { \lambda \vdash (m+n) \\ c_1(\lambda) \le N  }} ( d_{m+n}(\lambda) )^2.
\eea
This follows from the relation between the kernels of the maps $ \rho_{ N , m, n } : B_N (m,n) \rightarrow \Vmn $ and  $ \rho_{ N , m+n} : \mC ( S_{ m+n}  ) \rightarrow V_N^{ \otimes m+n } $ explained in section \ref{sec:relkern}. 
Then by comparing right-hand side of~\eqref{eq:Dim1} with right-hand side of~\eqref{eq:Dim2}, we deduce an alternative expression for  $\Delta(m,n;N)$:
\bea\label{identityDims} 
\Delta ( m, n ; N ) 
= \sum_{  \gamma \in \widehat {\BRT}( m,n,N  )     }
( d_{m,n} (\gamma) )^2 - \sum_{ \substack { \lambda \vdash (m+n) \\ c_1(\lambda) \le N  }} ( d_{m+n}(\lambda) )^2. 
\eea
 For $N = ( m+n) -l $ with $l$ small, and $m,n$ arbitrarily large, the two sums on the RHS only involve small numbers of diagrams (controlled by the small $l$) and the summands involve representation theoretic quantities from $S_{m+n}$ and the stable large $N$ regime of $B_N (m,n)$. The symmetric group irrep dimensions can be computed using the standard hook formula, and 
the stable regime dimension is also expressed in terms of hook lengths in \eqref{eq:dimWBA}. In the following we will calculate $\Delta ( m, n ; N )$  from \eqref{eq:deficit} after working out the modifications $ \delta_{ m,n,N} ( \gamma_i ) $ using the RBD, and we will compare with the simpler formula \eqref{identityDims}. The agreement between the two calculations will give a non-trivial check of the computation of the modified dimensions.

\subsection{Modified dimensions for $l=2$}
Irrep $\gamma=(k,\gamma_+,\gamma_{-})$ in $B_N(m,n)$ which receives a modification to its dimension, see Figure~\ref{Fig2} for several examples, is
\begin{align}\label{leq2gam} 
\gamma=(1,[1^{m-1}],[1^{n-1}]),\quad d_{m,n}(\gamma)=\frac{m!n!}{1!(m-1)!(n-1)!}=mn.
\end{align}

\begin{figure}[h!] 
	\centering
	\includegraphics[scale=0.4]{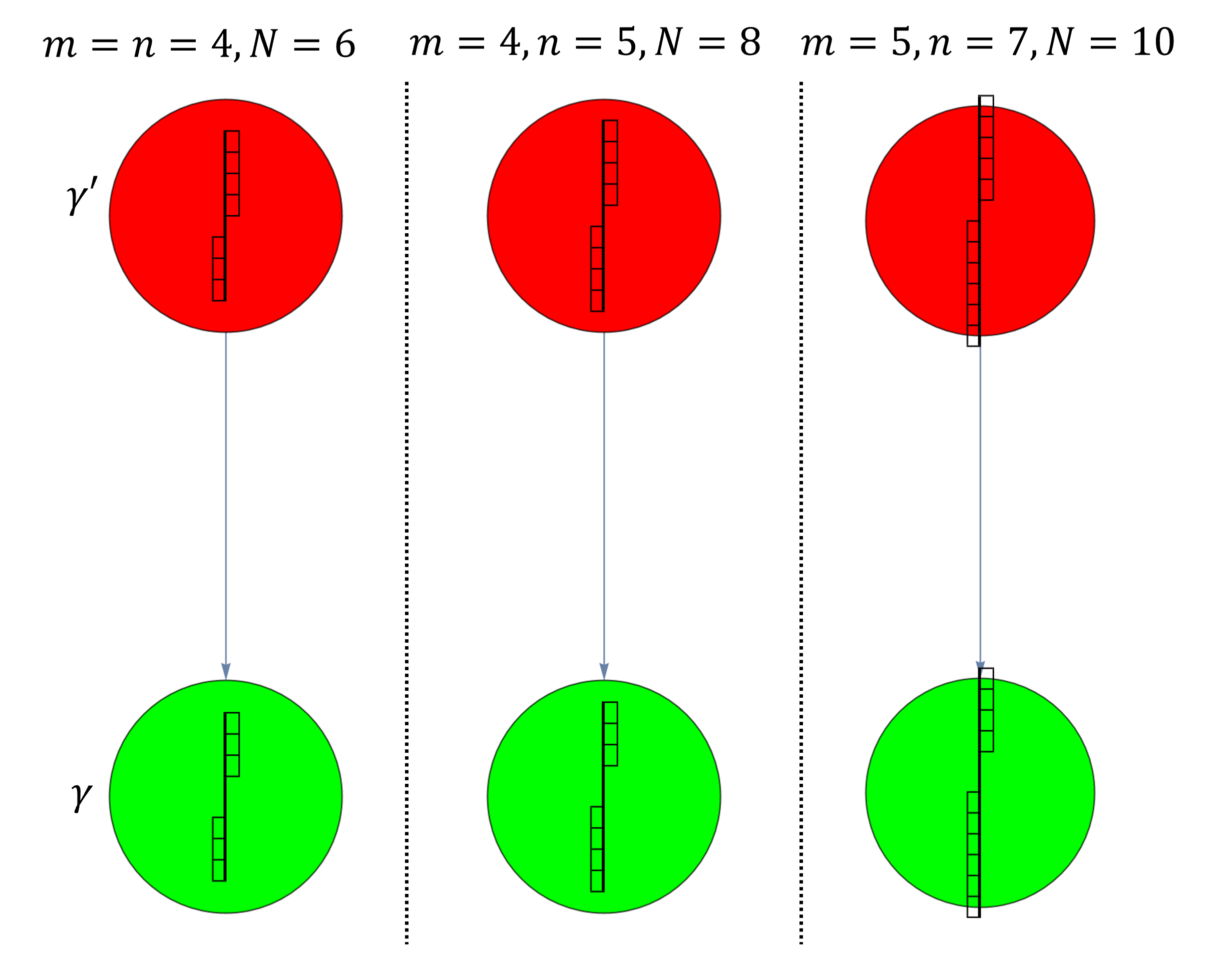}
	\caption{The graphic presents the RBDs for $m=n=4, N=6$, $m=4, n=5, N=7$, and $m=5, n=7, N=10$ counting from the left.} 
	\label{Fig2}
\end{figure} 
The modified dimension is equal to 
\bea 
\widehat{d}_{m,n,m+n-2}(\gamma)=d_{m,n}(\gamma)  - \delta_{m,n,m+n-2}(\gamma).
\eea
The dimension modification $\delta_{m,n,N}(\gamma)$ is expressed in terms of dimension $d_{m,n}(\gamma')$ of the dropped irrep $\gamma'=(0,[1^{n-1}],[1^{m-1}])$ as an irrep of $B_N(m,n-1)$:
\begin{align}\label{leq2delt}
\delta_{m,n,m+n-2}(\gamma)=d_{m,n-1}(\gamma')=\frac{(n-1)!(m-1)!}{0!(n-1)!(m-1)!}=1.
\end{align}
The dimension of the algebra $\mathcal{A}^{N=m+n-2}_{m,n}$ is equal to 
\bea 
\dim(\cA^{N=m+n-2}_{ m , n  })= - \Delta ( m, n ; N = m+n-2)
  + \sum_{ \mu \in \widehat { \BRT} ( m,n , N=m+n-2 )   } ( d_{m,n} (  \mu ))^2 ,
\eea
where $\Delta ( m , n ; N = m+n-2)= ( d_{m,n}(\gamma))^2 - ( d_{m,n}(\gamma)  - \delta_{m,n,m+n-2}(\gamma))^2$. Combining~\eqref{leq2gam} and~\eqref{leq2delt} we get an universal expression for dimension modification:
\begin{equation}
\label{eq:mod1}
\boxed{
\Delta ( m , n ; N = m+n-2)=2mn-1.}
\end{equation}
Computing this $ \Delta ( m,n ; N = m+n-2) $ by using equation~\eqref{identityDims}, we obtain a sequence starting from $ n =2$  to $ n =14$ for $m=2$:
\bea 
\Delta ( 2, 2\leq n\leq 14 ; N = m+n-2)=\{7, 11, 15, 19, 23, 27, 31, 35, 39, 43, 47, 51, 55\},
\eea
which agrees with~\eqref{eq:mod1}.

\subsection{Modified dimensions for $l=3$}

We will describe the calculations of the modified dimension for the  general $2$-parameter case $ (m,n, N  = m+n-3)$, starting from the simpler special cases cases $ (2,n, N  = n-1) $ and $  ( 3,n, N = n)$. 

\subsubsection{Case of $B_N(2,n)$ and $N=n-1$}
As it is depicted in  Figure~\ref{Gen_n2m5N4}, we have two irreps $\gamma=(k,\gamma_+,\gamma_{-})$ in $B_N(2,n)$ which receive a dimension modification:
\begin{align}\label{modleq3gams} 
\gamma_1&=(1,[1],[2,1^{n-3}]),\cr 
\gamma_2&=(2,[\emptyset],[1^{n-2}]).
\end{align}
\begin{figure} 
	\centering
	\includegraphics[scale=0.40]{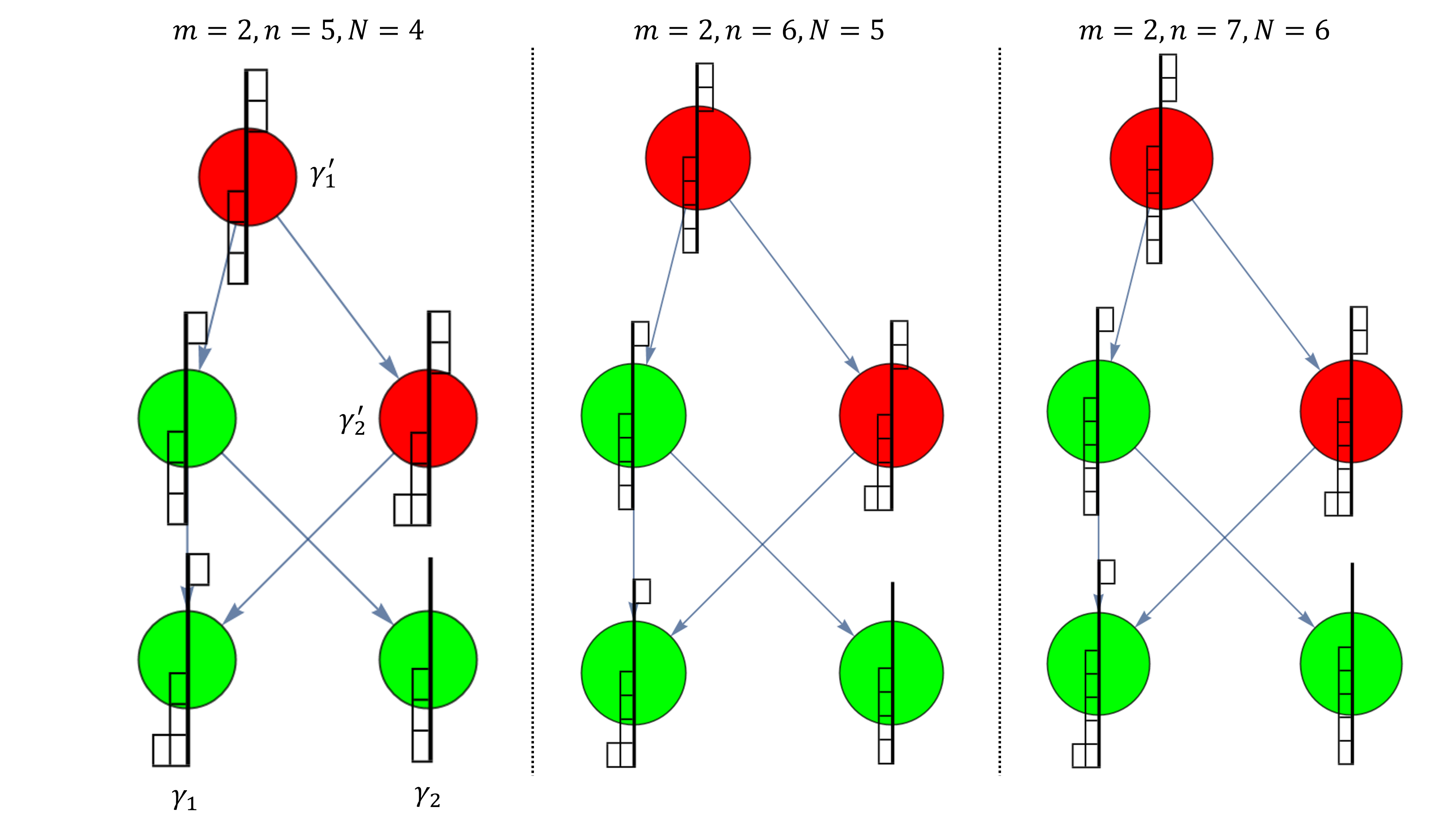}
	\caption{The graphic presents RBDs for $m=2,n=5,N=4$, $m=2, n=6, N=5$, and $m=2, n=7, N=6$ counting from the left. In this case, we have already achieved the stable regime for $m=2, n=3$.} 
	\label{Gen_n2m5N4}
\end{figure} 
Their dimensions are
\begin{align}
d_{2,n}(\gamma_1)&=\frac{2!n!}{1!1!(n-1)(n-3)!}=2n(n-2),\cr 
d_{2,n}(\gamma_2)&=\frac{2!n!}{2!0!(n-2)!}=n(n-1).
\end{align}
The modifications to the dimensions are 
\bea 
\widehat{d}_{2,n,N=n-1}(\gamma_i)=d_{2,n}(\gamma_i)  - \delta_{2,n,N=n-1}(\gamma_i), \quad \text{for}\quad i=1,2.
\eea
The modifications $\delta_{2,n,N=n-1}(\gamma_i)$ are expressed in terms of dimensions $d_{2,n-2}(\gamma'_1),\\ d_{2,n-1}(\gamma'_2)$ of the irreps in red: 
\begin{align}\label{modleq3delts} 
\delta_{2,n,N=n-1}(\gamma_1)&=d_{2,n-2}(\gamma'_1)+d_{2,n-1}(\gamma'_2)=n-1,\cr 
\delta_{2,n,N=n-1}(\gamma_2)&=d_{2,n-2}(\gamma'_1)=1.
\end{align}
The right-hand sides of the above expression are obtained by using the following formulas for dimensions $d_{2,n-2}(\gamma'_1),d_{2,n-1}(\gamma'_2)$:
\begin{itemize}
	\item $d_{2,n-2}(\gamma'_1)$ is the dimension of $\gamma_1'=(0,[1^2],[1^{n-2}])$ as an irrep of $B_N(2,n-2)$ for $N\geq n$
	\begin{align}
	d_{2,n-2}(\gamma'_1)=\frac{2!(n-2)!}{0!2!(n-2)!}=1.
	\end{align}
	\item $d_{2,n-1}(\gamma'_2)$ is the dimension of $\gamma'_2=(0,[1^2],[2,1^{n-3}])$ as an irrep of $B_N(2,n-1)$ for $N\geq n+1$
	\begin{align}
	d_{2,n-1}(\gamma'_2)=\frac{2!(n-1)!}{0!2!(n-1)(n-3)!}=n-2.
	\end{align}
\end{itemize}
The dimension of the algebra $\mathcal{A}^{N=n-1}_{2,n}$ is equal to
\bea 
\dim(\mathcal{A}^{N=n-1}_{2,n})=- \Delta ( m, n ; N = n-1)
+ \sum_{ \mu \in \widehat { \BRT} ( 2,n , N=n-1 )   } ( d_{2,n} (  \mu ))^2 , 
\eea
where 
\begin{equation}
\label{eq:delta2}
\boxed{\begin{split}
\Delta ( 2 , n ; N = n-1)&= \sum_{ i =1}^{ 2 }\left[( d_{2,n} (\gamma_i))^2 - ( d_{2,n} (\gamma_{i})  - \delta_{2,n,N=n-1} )^2\right]\\
&=4n^3-11n^2+8n-2.
\end{split}}
\end{equation}
Computing this $ \Delta ( 2,n ; N = n-1) $ by using equation~\eqref{identityDims}, we obtain a sequence starting from $ n =3$  to $ n =13$ for $m=2$:
\bea 
\begin{split}
&\Delta ( 2, 3\leq n\leq 13 ; N = n-1)\\
&=\{31, 110, 263, 514, 887, 1406, 2095, 2978, 4079, 5422, 7031\}
\end{split}
\eea
which agrees with equation~\eqref{eq:delta2}.
\subsubsection{Case of $B_N(3,n)$ and $N=n$}
\label{sec:m3nNn}
As it is depicted in Figure~\ref{Gen_n3m5N5}, we have three irreps $\gamma=(k,\gamma_+,\gamma_{-})$ in $B_N(3,n)$ which receive a modification to their dimension:
\begin{align}
\gamma_1&=(1,[2],[1^{n-1}]),\cr
\gamma_2&=(1,[1^2],[2,1^{n-3}]),\cr
\gamma_3&=(2,[1],[1^{n-2}]).
\end{align}
Their dimensions are
\begin{align}
d_{3,n}(\gamma_1)&=\frac{3!n!}{1!2!(n-1)!}=3n,\cr
d_{3,n}(\gamma_2)&=\frac{3!n!}{1!2!(n-1)(n-3)!}=3n(n-2),\cr
d_{3,n}(\gamma_3)&=\frac{3!n!}{2!1!(n-2)!}=3n(n-1).
\end{align}
The modifications to the dimensions are 
\bea 
\widehat{d}_{3,n,N=n} (\gamma_i)= d_{3,n}(\gamma_i)  - \delta_{3,n,N=n}(\gamma_i), \quad \text{for}\quad i=1,2,3.
\eea
The modifications $ \delta_{3,n,N=n}(\gamma_i)$ are expressed in terms of dimensions $ d_{3,n-1}(\gamma_1'),\\ d_{3,n-2}(\gamma_2'),d_{3,n-1}(\gamma_3')$  of the irreps in red:
\begin{align}
\delta_{3,n,N=n}(\gamma_1)&=d_{3,n-1}(\gamma_1')=2,\cr
\delta_{3,n,N=n}(\gamma_2)&=d_{3,n-2}(\gamma_2')+d_{3,n-1}(\gamma_3')=n-1,\cr
\delta_{3,n,N=n}(\gamma_3)&=d_{3,n-2}(\gamma_2')=1.
\end{align}
\begin{figure} 
	\centering
	\includegraphics[scale=0.4]{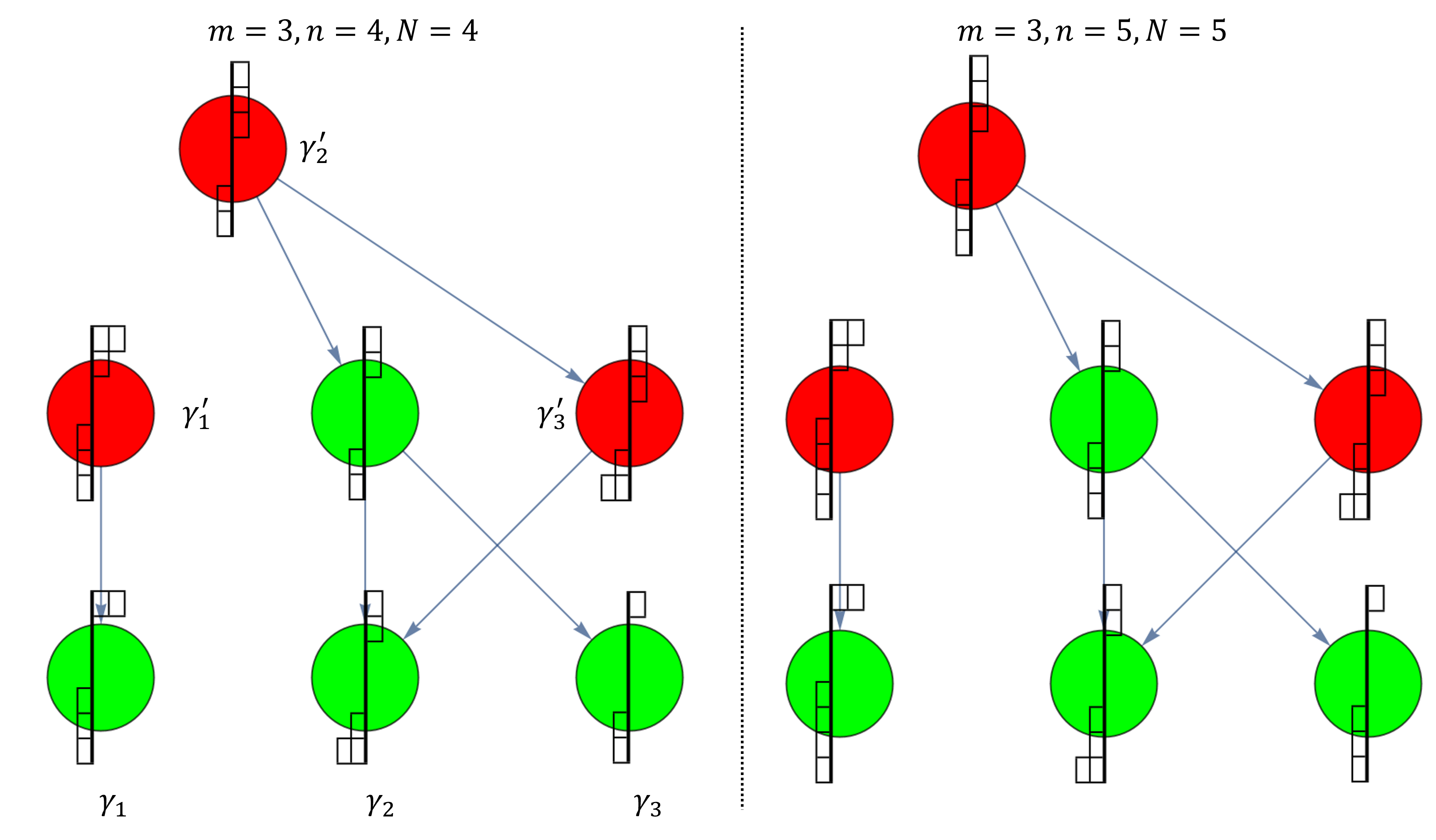}
	\caption{The graphics presents RBDs for $m=3, n=4, N=4$, and $m=3, n=5, N=5$ counting from the left. In this case, we have already achieved the stable regime for $m=3, n=3$.} 
	\label{Gen_n3m5N5}
\end{figure} 
The right-hand sides of the above expression are obtained by using the following formulas for dimensions $d_{3,n-1}(\gamma_1'),d_{3,n-2}(\gamma_2'),d_{3,n-1}(\gamma_3')$:
\begin{itemize}
	\item $d_{3,n-1}(\gamma_1')$ is the dimension of $\gamma_1'=(0,[2,1],[1^{n-1}])$ as an irrep of $ B_N (3,n-1 )$ for $ N \geq n+2$.
	\begin{align}
		d_{3,n-1}(\gamma_1')=\frac{3!(n-1)!}{0!3(n-1)!}=2.
	\end{align}
	\item $d_{3,n-2}(\gamma_2')$ is the dimension of $\gamma_2'=(0,[1^3],[1^{n-2}])$ as an irrep of $ B_N (3,n-2 )$ for $ N \geq n+1$.
	\begin{align}
	d_{3,n-2}(\gamma_2')=\frac{3!(n-2)!}{0!3!(n-2)!}=1.
	\end{align}
	\item $d_{3,n-1}(\gamma_3')$ is the dimension of $\gamma_3'=(0,[1^3],[2,1^{n-3}])$ as an irrep of $ B_N (3,n-1 )$ for $ N \geq n+2$.
	\begin{align}
	d_{3,n-1}(\gamma_3')=\frac{3!(n-1)!}{0!3!(n-1)(n-3)!}=n-2.
	\end{align}
\end{itemize}
The dimension of the algebra $\mathcal{A}^{N=n}_{3,n}$ is equal to
\bea 
\dim(\mathcal{A}^{N=n}_{3,n})=- \Delta ( 3, n ; N = n)
+ \sum_{ \mu \in \widehat { \BRT} ( 3,n , N=n )   } ( d_{3,n} (  \mu ))^2 ,
\eea
where 
\begin{equation}
	\label{eq:3nN=n}
	\boxed{
	\begin{split}
\Delta ( 3 , n ; N = n)&= \sum_{ i =1}^{ 3 }\left[( d_{3,n}(\gamma_i))^2 - ( d_{3,n}(\gamma_i)  - \delta_{3,n,N=n}(\gamma_i) )^2\right]\\
&= 6n^3-13n^2+20n-6.
\end{split}}
\end{equation}
Computing the $ \Delta ( 3,n ; N = n ) $ by using equation~\eqref{identityDims}, we obtain a sequence starting from $ n =3$  to $ n =13$ 
\bea 
\begin{split}
&\Delta ( 3, 3\leq n\leq 13 ; N = n)\\
 &=\{99, 250, 519, 942, 1555, 2394, 3495, 4894, 6627, 8730, 11239 \},
\end{split}
\eea
which agrees with~\eqref{eq:3nN=n}.
\subsubsection{Case of $B_N(m,n)$ and $N=m+n-3$}
As we pointed out in Figure~\ref{Gen_n3m5N5}, in this case we achieve the stable regime already for $m=3$. This allows us to present here the analysis for an arbitrary $m\geq 3$. Irreps $\gamma=(k,\gamma_+,\gamma_{-})$ in $B_N(m,n)$ which receive  modifications to their dimension are
\begin{align}
\gamma_1&=(1,[2,1^{m-3}],[1^{n-1}]),\cr 
\gamma_2&=(1,[1^{m-1}],[2,1^{n-3}]),\cr 
\gamma_3&=(2,[1^{m-2}],[1^{n-2}]).
\end{align}
Their dimensions are
\begin{align}
d_{m,n}(\gamma_1)&=\frac{m!n!}{1!(m-1)(m-3)!(n-1)!}=mn(m-2),\cr 
d_{m,n}(\gamma_2)&=\frac{m!n!}{1!(m-1)!(n-1)(n-3)!}=mn(n-2),\cr 
d_{m,n}(\gamma_3)&=\frac{m!n!}{2!(m-2)!(n-2)!}=\frac{1}{2}mn(m-1)(n-1).
\end{align}
The modifications to the dimensions are 
\bea 
\widehat{d}_{m,n}(\gamma_i) = d_{m,n}(\gamma_i)  - \delta_{m,n,N=m+n-3}(\gamma_i), \quad \text{for}\quad i=1,2,3.
\eea
The modifications $ \delta_{m,n,N=m+n-3}(\gamma_i)$ are expressed in terms of dimensions $ d_{m,n-1}(\gamma_1'),\\ d_{m,n-2}(\gamma_2'),d_{m,n-1}(\gamma_3')$ of the irreps in red, as it is depicted in Figure~\ref{Gen_n3m5N5}:
\begin{align}
\delta_{m,n,N=m+n-3}(\gamma_1)&=d_{m,n-1}(\gamma_1')=m-1,\cr 
\delta_{m,n,N=m+n-3}(\gamma_2)&=d_{m,n-2}(\gamma_2')+d_{m,n-1}(\gamma_3')=n-1,\cr 
\delta_{m,n,N=m+n-3}(\gamma_3)&=d_{m,n-2}(\gamma_2')=1.
\end{align}
The right-hand sides of the above expression are obtained by using the following formulas for dimensions $d_{m,n-1}(\gamma_1'), d_{m,n-2}(\gamma_2'), d_{m,n-1}(\gamma_3')$:
\begin{itemize}
	\item $d_{m,n-1}(\gamma_1')$ is the dimension of $\gamma_1'=(0,[2,1],[1^{n-1}])$ as an irrep of $ B_N (m,n-1 )$ for $ N \geq m+n-1$.
	\begin{align}
		d_{m,n-1}(\gamma_1')=\frac{m!(n-1)!}{0!m(m-2)!(n-1)!}=m-1.
	\end{align}
	\item $d_{m,n-2}(\gamma_2')$ is the dimension of $\gamma_2'=(0,[1^m],[1^{n-2}])$ as an irrep of $ B_N (m,n-2 )$ for $ N \geq m+n-2$.
	\begin{align}
		d_{m,n-2}(\gamma_2')=\frac{m!(n-2)!}{0!m!(n-2)!}=1.
	\end{align}
	\item $d_{m,n-1}(\gamma_3')$ is the dimension of $\gamma_3'=(0,[1^m],[2,1^{n-3}])$ as an irrep of $ B_N (m,n-1 )$ for $ N \geq m+n-1$.
	\begin{align}
		d_{m,n-1}(\gamma_3')=\frac{m!(n-1)!}{0!m!(n-1)(n-3)!}=n-2.
	\end{align}
\end{itemize}
The dimension of the algebra $\mathcal{A}^{N=m+n-3}_{m,n}$ is equal to
\bea 
\dim(\mathcal{A}^{N=m+n-3}_{m,n})=- \Delta ( m, n ; N = m+n-3)
+ \sum_{ \mu \in \widehat { \BRT} ( m,n , N=m+n-3 )   } ( d_{m,n} (  \mu ))^2 ,
\eea
where 
\begin{equation}
	\boxed{
	\begin{split}
&\Delta ( m , n ; N = m+n-3)=\\
&= \sum_{ i =1}^{ 3 }\left[( d_{m,n} (\gamma_i))^2 - ( d_{m,n}(\gamma_i)  - \delta_{m,n,N=m+n-3}(\gamma_i))^2\right]\\
&=m(2n^3-7n^2+9n+2)+m^2(n^2-7n-1)+2m^3n+2n-n^2-3.
\end{split}}
\end{equation}
Computing the $ \Delta ( m =4 ,n ; N = m+n-3 )$ by using equation~\eqref{identityDims}, we obtain the a sequence starting  from $ n =3$  to $ n =13$: 
\bea 
\begin{split}
&\Delta ( 4, 3\leq n\leq 13 ; N =n+1)\\  &=\{250, 509, 934, 1573, 2474, 3685, 5254, 7229, 9658, 12589, 16070 \}.
\end{split}
\eea
which agrees with ~\eqref{eq:mod1}.

Substituting in all above expressions $m=3$ we recover all results from subsection~\ref{sec:m3nNn}.

\subsection{Modified dimensions  for $l=4$}

We will describe the calculations of the modified dimension for the  general $2$-parameter case $ (m,n, N  = m+n-4)$, starting from the simpler special cases cases $ (2,n, N  = n-2) $ and $  ( 3,n, N = n-1)$. 

\subsubsection{Case of $B_N(2,n)$ and $N=n-2$}
As it is depicted in Figure~\ref{Gen_n2m7N5}, we have four irreps $\gamma=(k,\gamma_+,\gamma_{-})$ in $B_N(2,n)$ which receive modifications to their dimension:
\begin{align}
&\gamma_1 = (1,[1] , [3,1^{n-4}] ),\cr 
&\gamma_2 = (1,[1] , [2^2,1^{n-3} ]),\cr 
&\gamma_3 = (2,[\emptyset] , [2,1^{n-4}]),\cr 
&\gamma_4 = (2,[\emptyset] , [1^{n-2}]).
\end{align}

\begin{figure} 
	\centering
	\includegraphics[scale=0.4]{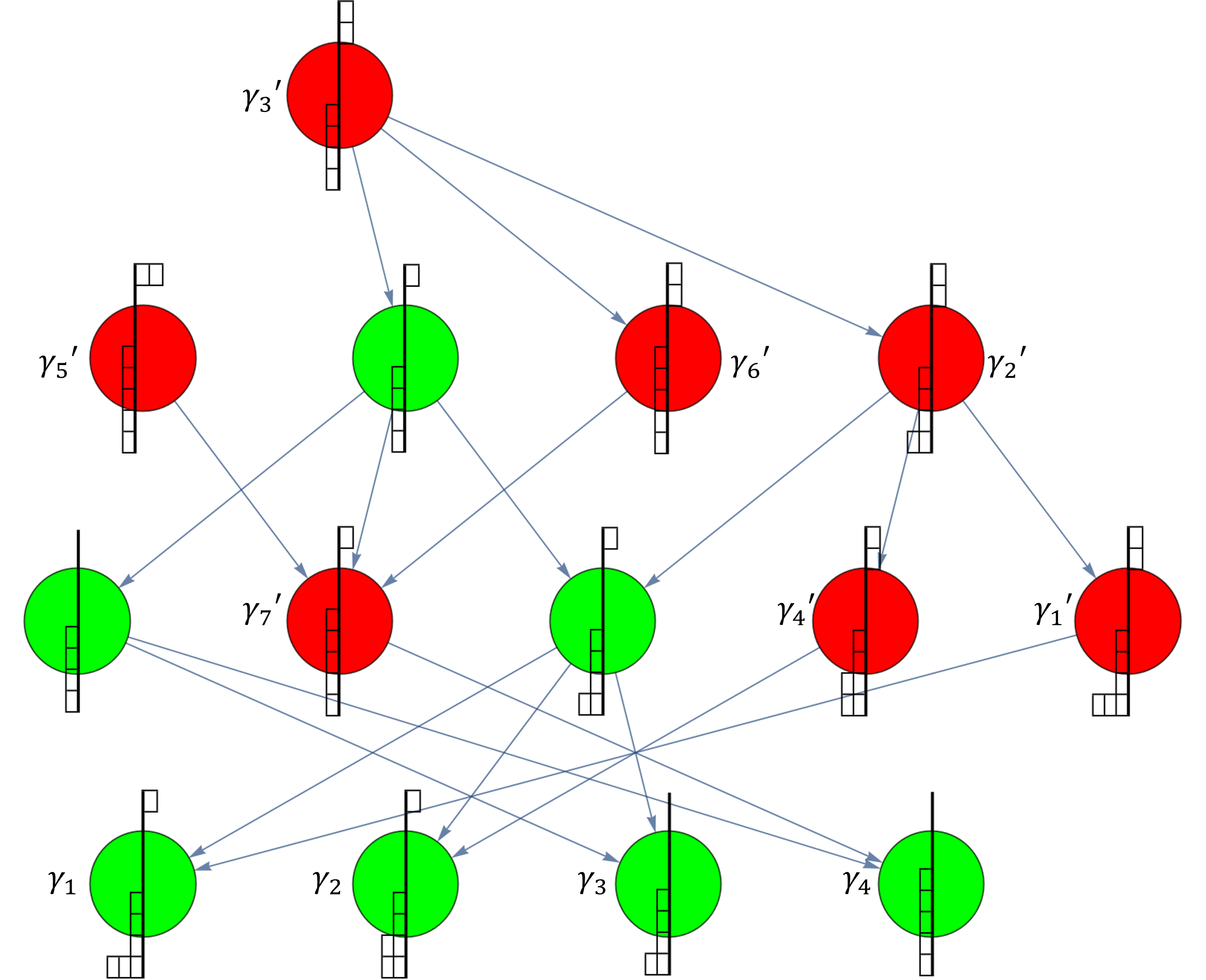}
	\caption{The graphics presents RBD for $m=2, n=7, N=5$. In this case, we have already achieved the stable regime for $m=2, n=5, N=3$.} 
	\label{Gen_n2m7N5}
\end{figure} 

Their dimensions are
\begin{align} 
&d_{2,n}( \gamma_1) =  { 2! n! \over 1!1! ( n-4)! 2 ( n-1) } = n ( n-2) ( n-3),\cr 
&d_{2,n} ( \gamma_2 ) =  { 2 ! n! \over 1!1! ( n-5)! 2 ( n-2) ( n-3) } = n ( n -1) ( n-4),\cr 
&d_{2,n} ( \gamma_3 ) = {2! ( n)! \over 2! ( n-4)! ( n-2) } = n ( n-1) (n-3),\cr 
&d_{2,n}( \gamma_4 )  =  { 2! n! \over 2! ( n-2)! } = n ( n-1).
\end{align}
The modifications to their dimensions are 
\bea 
\widehat{d}_{2,n,N=n-2}(\gamma_i) = d_{2,n}(\gamma_i)  - \delta_{2,n,N=n-2}(\gamma_i), \quad \text{for}\quad i=1,\ldots,4.
\eea
The $\delta_{2,n,N=n-2}(\gamma_i)$ are expressed in terms of dimensions $d_{2,n-1}(\gamma_1'), d_{2,n-2}(\gamma_2'),\\ \ldots,d_{2,n-1}(\gamma_7')$ of the irreps in red from Figure~\ref{Gen_n2m7N5}: 
\begin{align}
\delta_{2,n,N-2}(\gamma_1) &= ( d_{2,n-1}( \gamma_1') - d_{2,n-2}(\gamma_2')  ) + ( d_{2,n-2}(\gamma_2') - d_{2,n-3}(\gamma_3') )+\\
 &+ d_{2,n-3}(\gamma_3') + d_{2,n-3}(\gamma_3') + d_{2,n-2}(\gamma_2')\\
 &= d_{2,n-1}( \gamma_1') +d_{2,n-2}(\gamma_2') + d_{2,n-3}(\gamma_3')=\frac{1}{2}(n^2-3n+2),\\
\delta_{2,n,N-2}(\gamma_2) &= ( d_{2,n-2}(\gamma_2') - d_{2,n-3}(\gamma_3')) + d_{2,n-3}(\gamma_3') +\\
&+ (d_{2,n-1}(\gamma_4') - d_{2,n-2}(\gamma_2') ) + d_{2,n-3}(\gamma_3') + d_{2,n-2}(\gamma_2') \\
& = d_{2,n-2}(\gamma_2') + d_{2,n-3}(\gamma_3') + d_{2,n-1}(\gamma_4')=\frac{1}{2}n(n-3),\\ 
\delta_{2,n,N-2}(\gamma_3)& = d_{2,n-2}(\gamma_2')+2 d_{2,n-3}(\gamma_3')=n-1, \\
\delta_{2,n,N-2}(\gamma_4)&= ( d_{2,n-1}(\gamma_7') - d_{2,n-1}(\gamma_5') - d_{2,n-1}(\gamma_5') - d_{2,n-3}(\gamma_3') ) +\\
&+ d_{2,n-1}(\gamma_5') + ( d_{2,n-1}(\gamma_5') - d_{2,n-2}(\gamma_2') ) + 3d_{2,n-3}(\gamma_3')\\
&= d_{2,n-3}(\gamma_3') + d_{2,n-1}(\gamma_7')=2n-1.
\end{align}
The right-hand sides of the above expression are obtained by using the following formulas for the respective dimensions:
\begin{itemize}
	\item $d_{2,n-1}(\gamma_1')$ is the dimension of $\gamma_1'=(0, [1^2] , [3,1^{n-4} ]) $ as an irrep of $ B_N ( 2, n-1)$ for $ N \ge n+1 $. 
	\bea 
	d_{2,n-1}(\gamma_1') = { 2! ( n-1)! \over 0! 2! ( n-4)! 2 ( n-1) } = { ( n -2 ) ( n-3) \over 2 } .
	\eea 
	\item $d_{2,n-2}(\gamma_2')$ is the dimension of $\gamma_2'=(0, [1^2 ] , [ 2,1^{ n-4} ] ) $ as an irrep of $ B_N ( 2,n-2) $ for $ N \ge n$. 
	\bea 
	d_{2,n-2}(\gamma_2') = { 2! ( n-2)! \over 0! 2! ( n-4)! ( n-2) } = ( n-3). 
	\eea
	\item $d_{2,n-3}(\gamma_3')$ is the dimension of $ \gamma_3'=(0,[1^2] , [1^{n-3}] ) $ as an irrep of $ B_N  ( 2,n-3)$ for $ N \ge n-1 $. 
	\bea 
	d_{2,n-3}(\gamma_3') = { 2! ( n-3)! \over 0!2! ( n-3)! } = 1.
	\eea
	\item $d_{2,n-1}(\gamma_4')$ is the dimension from $ \gamma_4'=(0, [1^2] , [ 2^2 , 1^{n-5} ] $ as an irrep of $ B_N ( 2, n-1) $
	for $ N \ge  n+1$.
	\bea 
	d_{2,n-1}(\gamma_4') = { 2! ( n-10! \over 0!2! ( n-5)! 2 ( n-2) ( n-3) } = { ( n-1) ( n-4 ) \over 2 }.
	\eea 
	\item $d_{2,n-1}(\gamma_5')$ is the dimension of $ \gamma_5'=(0,[2], [1^{ n-5} ])$ as an irrep of $ B_N ( 2 , n-1 )$ for $ N \ge n+1 $
	\bea 
	d_{2,n-1}(\gamma_5') = { 2! ( n-2)! \over 0!2! ( n-2)! } = 1.
	\eea 
	\item $d_{2,n-2}(\gamma_6')$ is the dimension of $\gamma_6'=( 0,[1^2] , [1^{ n-2}] ) $ as an irrep of $ B ( 2, n-2) $ for $ N \ge n$ 
	\bea 
	d_{2,n-2}(\gamma_6') = { 2! ( n-2)! \over 0! 2! ( n-2)! } = 1.
	\eea
	\item $d_{2,n-1}(\gamma_7')$ is the dimension of $\gamma_7'= (1,[1] , [1^{ n-2} ] )$ as an irrep of $ B_N ( 2,n-1) $ for $ N \ge n+1$ 
	\bea 
	d_{2,n-1}(\gamma_7') = { 2! ( n-1)! \over 1!1! ( n-2)! } = 2 ( n-1).
	\eea
\end{itemize}
The dimension of the algebra $\mathcal{A}^{N=n-2}_{2,n}$ is equal to
\bea 
\dim(\mathcal{A}^{N=n-2}_{2,n})=- \Delta ( 2, n ; N = n-2)
+ \sum_{ \mu \in \widehat { \BRT} ( 2,n , N=n-2 )   } ( d_{2,n} (  \mu ))^2 
\eea
where

 \begin{equation}
	\label{eq:2nN=n-2}
	\boxed{
	\begin{split}
\Delta ( 2 , n ; N = n-2)&= \sum_{ i =1}^{ 4 }\left[( d_{2,n}(\gamma_i))^2 - ( d_{2,n}(\gamma_i)  - \delta_{2,n,N=n-2}(\gamma_i) )^2\right]\\
&= \frac{1}{2} \left(4 n^5-29 n^4+78 n^3-85 n^2+34 n-6\right).
\end{split}}
\end{equation}
Computing the $ \Delta ( 2,n ; N = n-2 ) $ by using equation~\eqref{identityDims}, we obtain  a sequence starting from $ n =5$  to $ n =14$: 
\bea 
\begin{split}
&\Delta ( 2,n ; N = n-2 )\\
&=\{1082, 3753, 10210, 23525, 48102, 89917, 156758, 258465, 407170, \
617537 \},
\end{split}
\eea
which agrees with~\eqref{eq:2nN=n-2}.

\subsubsection{Case of $B_N(3,n)$ and $N=n-1$}
As it is depicted in Figure~\ref{Gen_3n}, we have six irreps $\gamma=(k,\gamma_+,\gamma_{-})$ in $B_N(3,n)$ which receive modifications to their dimension:
\begin{align}
&\gamma_1=(1,[2],[2,1^{n-3}]),\cr
&\gamma_2=(1,[1^2],[3,1^{n-4}]),\cr
&\gamma_3=(1,[1^2],[2^2,1^{n-5}]),\cr
&\gamma_4=(2,[1],[2,1^{n-4}]),\cr
&\gamma_5=(2,[1],[1^{n-2}]),\cr
&\gamma_6=(3,\emptyset,[1^{n-3}]).
\end{align}
Their dimensions are:
\begin{align}
&d_{3,n}(\gamma_1)=\frac{3!n!}{2!(n-3)!(n-2)}=3n(n-2),\cr
&d_{3,n}(\gamma_2)=\frac{3!n!}{2!(n-4)!(n-1)2}=\tfrac{3}{2}n(n-2)(n-3),\cr
&d_{3,n}(\gamma_3)=\frac{3!n!}{2!(n-5)!(n-2)(n-1)2}=\tfrac{3}{2}n(n-1)(n-4),\cr
&d_{3,n}(\gamma_4)=\frac{3!n!}{2!(n-4)!(n-2)}=3n(n-1)(n-3),\cr
&d_{3,n}(\gamma_5)=\frac{3!n!}{2!(n-2)!}=3n(n-1),\cr
&d_{3,n}(\gamma_6)=\frac{3!n!}{3!(n-3)!}=n(n-1)(n-2).
\end{align}

\begin{figure} 
	\includegraphics[scale=0.45]{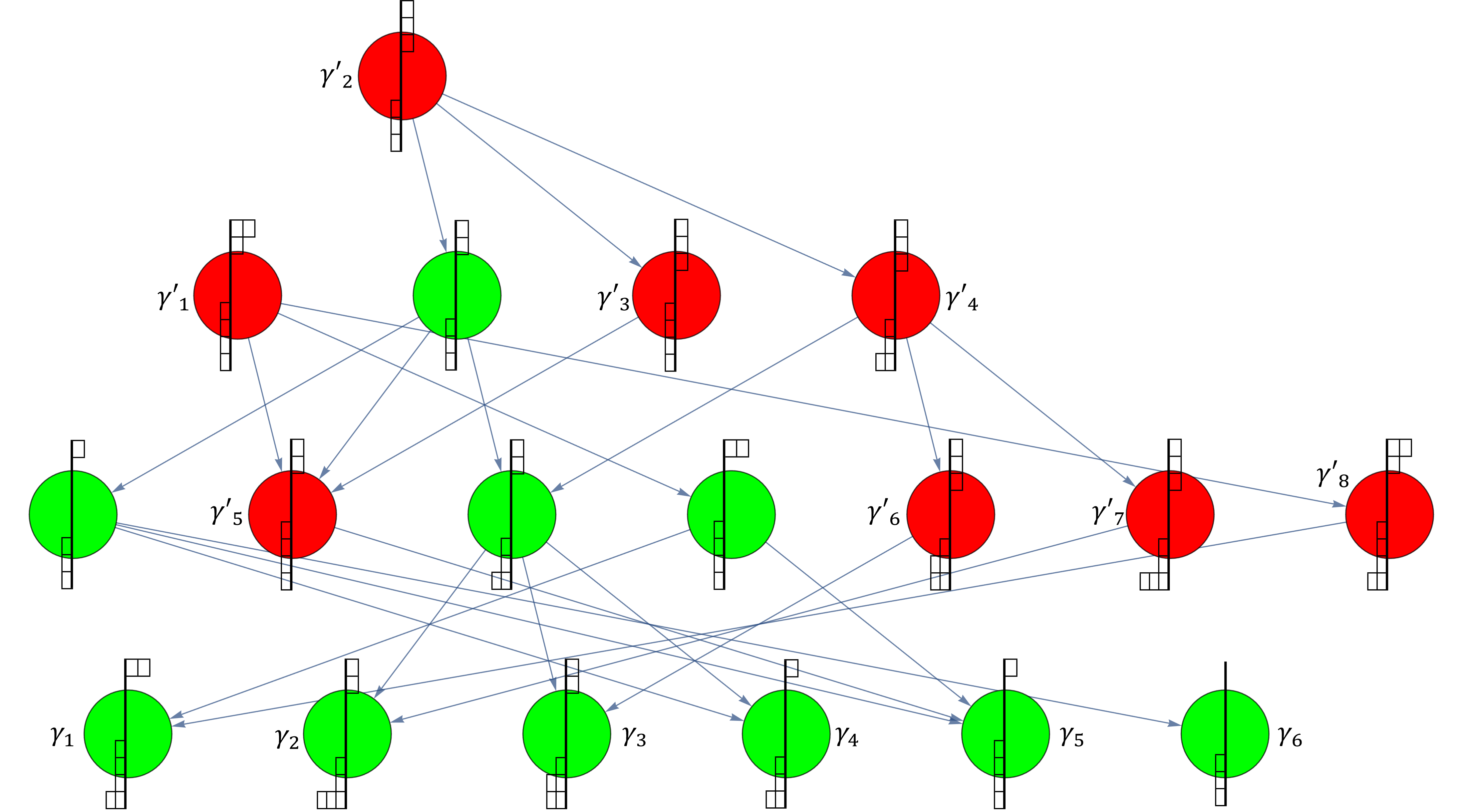}
	\caption{The graphic presents RBD for $m=3, n=6$ and $N=5$. The stable region is achieved for $m=3, n=5$ and $N=4$.} 
	\label{Gen_3n}
\end{figure}

The modifications to the dimensions are 
\bea 
\widehat{d}_{3,n}(\gamma_i)=d_{3,n} (\gamma_i)  - \delta_{3,n,N=n-1}(\gamma_i), \quad \text{for}\quad  i=1,\ldots,6.
\eea
The $\delta_{3,n,N=n-1}(\gamma_i)$ are expressed in terms of dimensions of the irreps in red:
\begin{align}
\delta_{3,n,N=n-1}(\gamma_1)&=d_{3,n-2}(\gamma_1')+(d_{3,n-1}(\gamma_8')-d_{3,n-2}(\gamma_1'))+d_{3,n-2}(\gamma_1')=d_{3,n-2}(\gamma_1')+d_{3,n-1}(\gamma_8')\cr 
&=2(n-1),\cr 
\delta_{3,n,N=n-1}(\gamma_2)&=3d_{3,n-3}(\gamma_2')+2(d_{3,n-2}(\gamma_4')-d_{3,n-3}(\gamma_2'))+(d_{3,n-1}(\gamma_7')-d_{3,n-2}(\gamma_4'))\cr 
&=d_{3,n-3}(\gamma_2')+d_{3,n-2}(\gamma_4')+d_{3,n-1}(\gamma_7')\cr 
&=\frac{1}{2}(n-1)(n-2),\cr 
\delta_{3,n,N=n-1}(\gamma_3)&=3d_{3,n-3}(\gamma_2')+2(d_{3,n-2}(\gamma_4')-d_{3,n-3}(\gamma_2'))+(d_{3,n-1}(\gamma_6')-d_{3,n-2}(\gamma_4'))\cr 
&=d_{3,n-3}(\gamma_2')+d_{3,n-2}(\gamma_4')+d_{3,n-1}(\gamma_6')\cr 
&=\frac{1}{2}n(n-3),\cr 
\delta_{3,n,N=n-1}(\gamma_4)&=d_{3,n-3}(\gamma_2')+(d_{3,n-2}(\gamma_4')-d_{3,n-3}(\gamma_2'))+d_{3,n-3}(\gamma_2')+d_{3,n-3}(\gamma_2')\cr 
&=2d_{3,n-3}(\gamma_2')+d_{3,n-2}(\gamma_4')=n-1,\cr 
\delta_{3,n,N=n-1}(\gamma_5)&=(d_{3,n-1}(\gamma_5')-d_{3,n-2}(\gamma_1')-d_{3,n-3}(\gamma_2')-d_{3,n-2}(\gamma_3'))\cr 
&+(d_{3,n-2}(\gamma_3')-d_{3,n-3}(\gamma_2'))+3d_{3,n-3}(\gamma_2')+2d_{3,n-2}(\gamma_1')\cr 
&=d_{3,n-2}(\gamma_1')+d_{3,n-3}(\gamma_2')+d_{3,n-1}(\gamma_5')\cr 
&=3n,\cr 
\delta_{3,n,N=n-1}(\gamma_6)&=d_{3,n-3}(\gamma_2')=1.
\end{align}

The right-hand sides of the above expression are obtained by using the following formulas for the dimensions:
\begin{itemize}
	\item $d_{3,n-2}(\gamma_1')$ is the dimension from $\gamma_1'=(0,[2,1],[1^{n-2}])$ as an irrep of $B_N(3,n-2)$, for $N\geq n+1$
	\begin{align}
	d_{3,n-2}(\gamma_1')=\frac{3!(n-2)!}{3(n-2)!}=2.
	\end{align}
	\item $d_{3,n-3}(\gamma_2')$ is the dimension from $\gamma_2'=(0,[1^3],[1^{n-3}])$ as an irrep of $B_N(3,n-3)$, for $N\geq n$
	\begin{align}
	d_{3,n-3}(\gamma_2')=\frac{3!(n-3)!}{3!(n-3)!}=1.
	\end{align}
	\item $d_{3,n-2}(\gamma_3')$ is the dimension from $\gamma_3'=(0,[1^3],[1^{n-2}])$ as an irrep of $B_N(3,n-2)$, for $N\geq n+1$
	\begin{align}
	d_{3,n-2}(\gamma_3')=\frac{3!(n-2)!}{3!(n-2)!}=1.
	\end{align}
	\item $d_{3,n-2}(\gamma_4')$ is the dimension from $\gamma_4'=(0,[1^3],[2,1^{n-4}])$ as an irrep of $B_N(3,n-2)$, for $N\geq n+1$
	\begin{align}
	d_{3,n-2}(\gamma_4')=\frac{3!(n-2)!}{3!(n-4)!(n-2)}=n-3.
	\end{align}
	\item $d_{3,n-1}(\gamma_5')$ is the dimension from $\gamma_5'=(1,[1^2],[1^{n-2}])$ as an irrep of $B_N(3,n-1)$, for $N\geq n+2$
	\begin{align}
	d_{3,n-1}(\gamma_5')=\frac{3!(n-1)!}{2!(n-2)!}=3(n-1).
	\end{align}
	\item $d_{3,n-1}(\gamma_6')$ is the dimension from $\gamma_6'=(0,[1^3],[2^2,1^{n-5}])$ as an irrep of $B_N(3,n-1)$, for $N\geq n+2$
	\begin{align}
	d_{3,n-1}(\gamma_6')=\frac{3!(n-1)!}{3!(n-5)!(n-3)(n-2)2}=\frac{1}{2}(n-1)(n-4).
	\end{align}
	\item $d_{3,n-1}(\gamma_7')$ is the dimension from $\gamma_7'=(0,[1^3],[3,1^{n-4}])$ as an irrep of $B_N(3,n-1)$, for $N\geq n+2$
	\begin{align}
	d_{3,n-1}(\gamma_7')=\frac{3!(n-1)!}{3!(n-4)!(n-1)2}=\frac{1}{2}(n-2)(n-3).
	\end{align}
	\item $d_{3,n-1}(\gamma_8')$ is the dimension from $\gamma_8'=(0,[2,1],[2,1^{n-3}])$ as an irrep of $B_N(3,n-1)$, for $N\geq n+2$
	\begin{align}
	d_{3,n-1}(\gamma_8')=\frac{3!(n-1)!}{3(n-3)!(n-1)}=2(n-2).
	\end{align}
\end{itemize}

The dimension of the algebra $\mathcal{A}^{N=n-1}_{3,n}$ is equal to
\bea 
\dim(\mathcal{A}^{N=n-1}_{3,n})=- \Delta ( 3, n ; N = n-1) + \sum_{ \mu \in \widehat { \BRT} ( 3,n , N=n-1 )   }
( d_{3,n}(\mu) )^2 , 
\eea
where 
\begin{equation}
\label{DiagFormula2}
\boxed{
\begin{split}  
\Delta ( 3 , n ; N = n-1)  &= \sum_{ i =1}^{ 6 }\left[( d_{3,n}(\gamma_i))^2 - ( d_{3,n}(\gamma_i)  - \delta_{3,n,N=n-1}(\gamma_i) )^2\right]\\
&=-7 + 41 n - \frac{195}{2}n^2 + 68 n^3 - \frac{37}{2} n^4 + 3 n^5.
\end{split}}
\end{equation}
Computing the $ \Delta ( 3,n ; N = n-1 ) $ by using equation~\eqref{identityDims}, we obtain a sequence starting from $n=5$ to $n=10$:
\bea 
\Delta ( 3 , n ; N = n-1)=\{4073,10769,24829,51425,97805 \}
\eea
which agrees with~\eqref{DiagFormula2}.

\subsubsection{Case of $B_N(m,n)$ and $N=m+n-4$}
\label{subSub}
As it is depicted in Figure~\ref{Gen_mnk4}, we have nine irreps $\gamma=(k,\gamma_+,\gamma_{-})$ in $B_N(m,n)$ which receive modifications to their dimension:
\begin{align}\label{leq4modgams} 
&\gamma_1=(1,[3,1^{m-4}],[1^{n-1}]),\cr 
&\gamma_2=(1,[2^2,1^{m-5}],[1^{n-1}]),\cr 
&\gamma_3=(1,[2,1^{m-3}],[2,1^{n-3}]),\cr 
&\gamma_4=(2,[2,1^{m-4}],[1^{n-2}]),\cr 
&\gamma_5=(1,[1^{m-1}],[3,1^{n-4}]),\cr 
&\gamma_6=(1,[1^{m-1}],[2^2,1^{n-5}]),\cr 
&\gamma_7=(2,[1^{m-2}],[2,1^{n-4}]),\cr 
&\gamma_8=(2,[1^{m-2}],[1^{n-2}]),\cr
&\gamma_9=(3,[1^{m-3}],[1^{n-3}]).
\end{align}
Their dimensions are:
\begin{align}
&d_{m,n}(\gamma_1)=\frac{m!n!}{(m-4)!2(m-1)(n-1)!}=\frac{1}{2}mn(m-2)(m-3),\cr 
&d_{m,n}(\gamma_2)=\frac{m!n!}{(m-5)!2(m-3)(m-2)(n-1)!}=\frac{1}{2}nm(m-1)(m-4),\cr 
&d_{m,n}(\gamma_3)=\frac{m!n!}{(m-3)!(m-1)(n-3)!(n-1)}=mn(m-2)(n-2),\cr 
&d_{m,n}(\gamma_4)=\frac{m!n!}{2!(m-4)!(m-2)(n-2)!}=\frac{1}{2}mn(m-3)(m-1)(n-1),\cr 
&d_{m,n}(\gamma_5)=\frac{m!n!}{(m-1)!(n-4)!(n-1)2}=\frac{1}{2}mn(n-3)(n-2),\cr 
&d_{m,n}(\gamma_6)=\frac{m!n!}{(m-1)!(n-5)!2(n-3)(n-2)}=\frac{1}{2}mn(n-1)(n-4),\cr 
&d_{m,n}(\gamma_7)=\frac{m!n!}{2!(m-2)!(n-4)!(n-2)}=\frac{1}{2}mn(m-1)(n-1)(n-3),\cr 
&d_{m,n}(\gamma_8)=\frac{m!n!}{2!(m-2)!(n-2)!}=\frac{1}{2}mn(m-1)(n-1),\cr 
&d_{m,n}(\gamma_9)=\frac{m!n!}{3!(m-3)!(n-3)!}=\frac{1}{6}mn(m-1)(n-1)(m-2)(n-2).
\end{align}
The modifications to the dimensions are 
\bea 
\widehat{d}_{m,n}(\gamma_i)= d_{m,n} (\gamma_i)  - \delta_{m,n,N=m+n-4}(\gamma_i), \quad \text{for}\quad  i=1,\ldots,9.
\eea
These $\widehat{d}_{m,n}(\gamma_i)$ are expressed in terms of dimensions of the irreps in red:
\begin{align} 
\delta_{m,n,N=m+n-4}(\gamma_1)&=d_{m,n-1}(\gamma_2')=\frac{1}{2}(m-1)(m-2),\cr 
\delta_{m,n,N=m+n-4}(\gamma_2)&=d_{m,n-1}(\gamma_1')=\frac{1}{2}m(m-3),\cr 
\delta_{m,n,N=m+n-4}(\gamma_3)&=d_{m,n-2}(\gamma_3')+(d_{m,n-1}(\gamma_{10}')-d_{m,n-2}(\gamma_3'))+d_{m,n-2}(\gamma_3')\cr
&=d_{m,n-2}(\gamma_3')+d_{m,n-1}(\gamma_{10}')=(m-1)(n-1),\cr 
\delta_{m,n,N=m+n-4}(\gamma_4)&=d_{m,n-2}(\gamma_3')=m-1,\cr 
\delta_{m,n,N=m+n-4}(\gamma_5)&=d_{m,n-3}(\gamma_4')+d_{m,n-2}(\gamma_7')+d_{m,n-1}(\gamma_9')\cr
&=\frac{1}{2}(n^2-3n+2)=\frac{1}{2}(n-1)(n-2),\cr
\delta_{m,n,N=m+n-4}(\gamma_6)&=d_{m,n-3}(\gamma_4')+d_{m,n-2}(\gamma_7')+d_{m,n-1}(\gamma_8')=\frac{1}{2}n(n-3),\cr
\delta_{m,n,N=m+n-4}(\gamma_7)&=d_{m,n-3}(\gamma_4')+(d_{m,n-2}(\gamma_7')-d_{m,n-3}(\gamma_4'))\cr
&+d_{m,n-3}(\gamma_4')+d_{m,n-3}(\gamma_4')\cr
&=2d_{m,n-3}(\gamma_4')+d_{m,n-2}(\gamma_7')=n-1,\cr
\delta_{m,n,N=m+n-4}(\gamma_8)&=d_{m,n-2}(\gamma_3')+d_{m,n-3}(\gamma_4')+d_{m,n-1}(\gamma_5')=mn,\cr
\delta_{m,n,N=m+n-4}(\gamma_9)&=d_{m,n-3}(\gamma_4')=1.\cr 
& \label{leq4moddelts2} 
\end{align}
\begin{figure} 
	\includegraphics[scale=0.5]{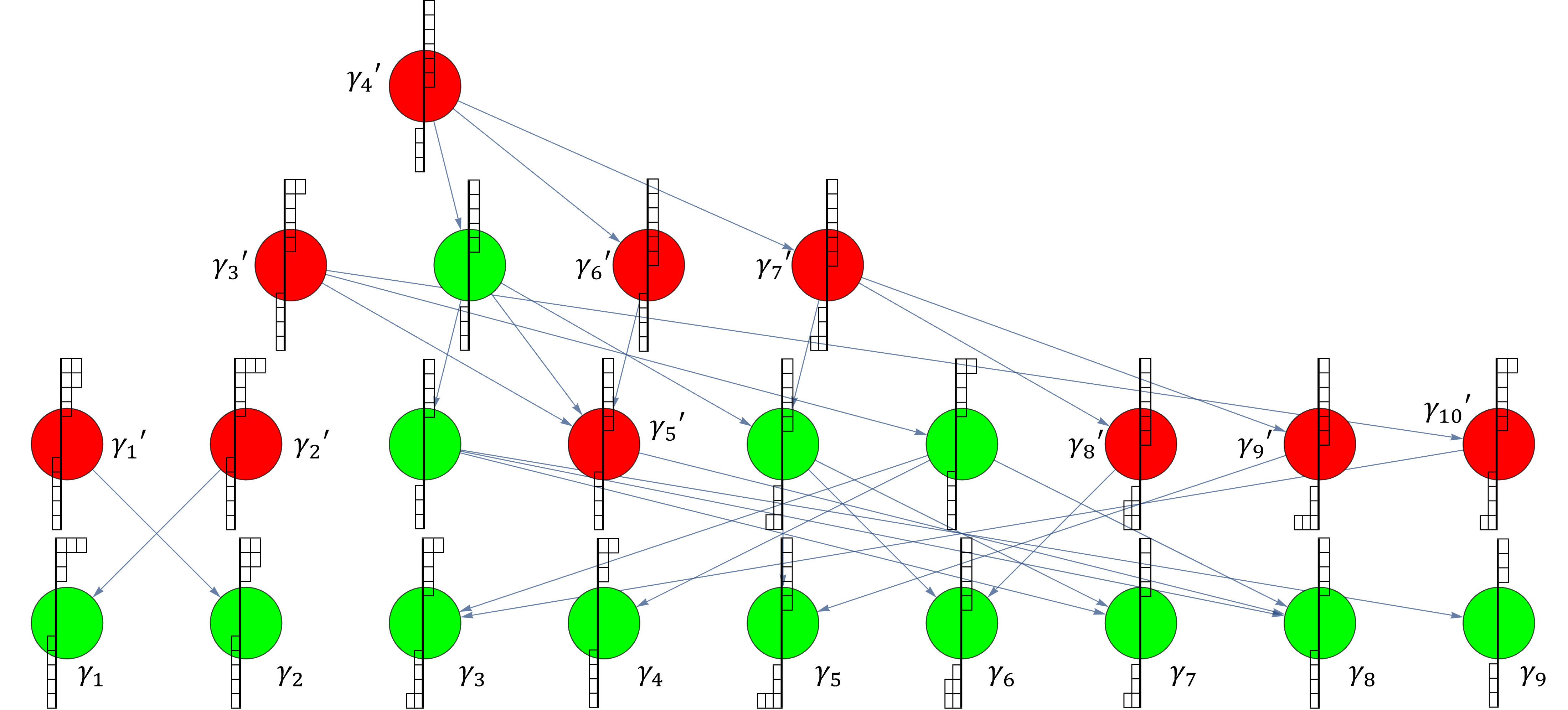}
	\caption{The graphic presents RBD for $m=n=6$ and $N=8$. The stable regime is achieved for $m=n=5$ and $N=6$.} 
	\label{Gen_mnk4}
\end{figure} 
The right-hand sides of the above expression are obtained by using the following formulas for dimensions $\gamma'_i$:
\begin{itemize}
	\item $d_{m,n-1}(\gamma_1')$ is the dimension from $\gamma_1'=(0,[2^2,1^{m-4}],[1^{n-1}])$ as an irrep of $B_N(m,n-1)$, for $N\geq n+m-1$
	\begin{align}
	d_{m,n-1}(\gamma_1')=\frac{m!(n-1)!}{(m-4)!(m-2)(m-1)2(n-1)!}=\frac{1}{2}m(m-3).
	\end{align}
	\item $d_{m,n-1}(\gamma_2')$ is the dimension from $\gamma_2'=(0,[3,1^{m-3}],[1^{n-1}])$ as an irrep of $B_N(m,n-1)$, for $N\geq n+m-1$
	\begin{align}
	d_{m,n-1}(\gamma_2')=\frac{m!(n-1)!}{(m-3)!m2(n-1)!}=\frac{1}{2}(m-1)(m-2).
	\end{align}
	\item $d_{m,n-2}(\gamma_3')$ is the dimension from $\gamma_3'=(0,[2,1^{m-2}],[1^{n-2}])$ as an irrep of $B_N(m,n-2)$, for $N\geq n+m-2$
	\begin{align}
	d_{m,n-2}(\gamma_3')=\frac{m!(n-2)!}{(m-2)!m(n-2)!}=m-1.
	\end{align}
	\item $d_{m,n-3}(\gamma_4')$ is the dimension from $\gamma_4'=(0,[1^m],[1^{n-3}])$ as an irrep of $B_N(m,n-3)$, for $N\geq m+n-3$
	\begin{align}
	d_{m,n-3}(\gamma_4')=\frac{m!(n-3)!}{m!(n-3)!}=1.
	\end{align}
	\item $d_{m,n-1}(\gamma_5')$ is the dimension from $\gamma_5'=(1,[1^{m-1}],[1^{n-2}])$ as an irrep of $B_N(m,n-1)$, for $N\geq m+n-1$
	\begin{align}
	d_{m,n-1}(\gamma_5')=\frac{m!(n-1)!}{(m-1)!(n-2)!}=m(n-1).
	\end{align}
	\item $d_{m,n-2}(\gamma_6')$ is the dimension from $\gamma_6'=(0,[1^m],[1^{n-2}])$ as an irrep of $B_N(m,n-2)$, for $N\geq m+n-2$
	\begin{align}
	d_{m,n-2}(\gamma_6')=\frac{m!(n-2)!}{m!(n-2)!}=1.
	\end{align}
	\item $d_{m,n-2}(\gamma_7')$ is the dimension from $\gamma_7'=(0,[1^m],[2,1^{n-4}])$ as an irrep of $B_N(m,n-2)$, for $N\geq m+n-2$
	\begin{align}
	d_{m,n-2}(\gamma_7')=\frac{m!(n-2)!}{m!(n-4)!(n-2)}=n-3.
	\end{align}
	\item $d_{m,n-1}(\gamma_8')$ is the dimension from $\gamma_8'=(0,[1^m],[2^2,1^{n-5}])$ as an irrep of $B_N(m,n-1)$, for $N\geq m+n-1$
	\begin{align}
	d_{m,n-1}(\gamma_8')=\frac{m!(n-1)!}{m!(n-5)!(n-3)(n-2)2}=\frac{1}{2}(n-1)(n-4).
	\end{align}
	\item $d_{m,n-1}(\gamma_9')$ is the dimension from $\gamma_9'=(0,[1^m],[3,1^{n-4}])$ as an irrep of $B_N(m,n-1)$, for $N\geq m+n-1$
	\begin{align}
	d_{m,n-1}(\gamma_9')=\frac{m!(n-1)!}{m!(n-4)!(n-1)2}=\frac{1}{2}(n-2)(n-3).
	\end{align}
	\item $d_{m,n-1}(\gamma_{10}')$ is the dimension from $\gamma_{10}'=(0,[2,1^{m-2}],[2,1^{n-3}])$ as an irrep of $B_N(m,n-1)$, for $N\geq m+n-1$
	\begin{align}
	d_{m,n-1}(\gamma_{10}')=\frac{m!(n-1)!}{(m-2)!m(n-3)!(n-1)}=(m-1)(n-2).
	\end{align}
\end{itemize}

The dimension of the algebra $\mathcal{A}^{N=m+n-4}_{m,n}$ is equal to 
\bea 
\dim(\mathcal{A}^{N=m+n-4}_{m,n})=- \Delta ( m, n ; N = m+n-4) + \sum_{ \mu \in \widehat { \BRT} ( m,n , N=m+n-4 )   }
( d_{m,n}(\mu) )^2,
\eea
where 
\begin{equation}
\label{DiagFormula}  
\boxed{
\begin{split}
&\Delta ( m , n ; N = m+n-4)= \sum_{ i =1}^{ 9 } ( d_{m,n}(\gamma_i))^2 - ( d_{m,n}(\gamma_i)  - \delta_{m,n,N=m+n-4}(\gamma_i) )^2=\\
&=\frac{1}{6}(6 m^5 n + m^4 (-3 - 54 n + 6 n^2) + 
2 m^3 (9 + 92 n - 39 n^2 + 10 n^3) -\\
&+3 (12 - 14 n + 15 n^2 - 6 n^3 + n^4) +3 m^2 (-15 - 84 n + 80 n^2 - 33 n^3 + 3 n^4)\\
& +m (42 + 140 n - 252 n^2 + 148 n^3 - 33 n^4 + 3 n^5)).
\end{split}}
\end{equation}
Computing~\eqref{DiagFormula} for $m=5$ and $5\leq n\leq  14$ we get
\bea
\begin{split}
&\text{\phantom{x}}\\
&\Delta ( m , n ; N = m+n-4)=\\
&=\{26564,54996,104624,186488,315336,510224,795116,1199484,1758908,2515676 \}.
\end{split}
\eea
Computing the $ \Delta ( m,n ; N = m+n-4 ) $ by using equation~\eqref{identityDims}, we obtain a sequence for $m=5$ and $5\leq n< 14$
\bea 
\begin{split}
	&\text{\phantom{x}}\\
	&\Delta ( m , n ; N = m+n-4)=\\
&=\{26564,54996,104624,186488,315336,510224,795116,1199484,1758908,2515676 \},
\end{split}
\eea
which agrees with~\eqref{DiagFormula}.

\section{Structure of  restricted Bratteli diagrams and counting of the red nodes } 
\label{Sec:countingReds}

In this section, we derive basic structural properties of the RBD for $B_{ N } (m,n)$, with $ N = m +n -l $, which lead to counting formulae for the red and green nodes in the RBD. 
In section \ref{sec:deepest} we show that the RBD have $ (l$ layers, labelled by depths $ 0 \le d \le (l-1)$.  Section \ref{sec:bdsxs} derives an upper bound on the excess height $ \Delta = c_1 ( \gamma_+ + c_1 ( \gamma_- ) - N $ in the RBD. Section \ref{mnstability} shows that the form of the RBD is independent of $(m,n)$ in the region $ (m,n) \ge (2l-3)$, the property we refer to as  degree-stability. Section \ref{gencountreds} derives the general formula \eqref{Redcount} for the counting  of red nodes.

\subsection{Depth and structure of  the deepest red }
\label{sec:deepest} 

The restricted Bratteli-diagrams have a layer for representations of $ B_{ N } ( m , n ) $, and subsequent rows at increasing depth. We define the initial row to have depth $d=0$. At depth $d$ we have irreps of  $ B ( m , n - d )$. We set $ N = m +n - l $ : for the non-semisimple regime of interest here, we have $ l\ge 1$. The depth $0$ diagrams are all green, which are connected to some red diagram at a depth $ d >0$. 

\begin{proposition}\label{depstruc} 
For $ N = ( m + n - l ) $, the largest depth $d$  which admits an $N$-excluded triple $ ( k , \gamma_+  , \gamma_-) \in \BRT  ( m , n-d )  $ is denoted  by $ d_{ \max} $ and  is given by 
\bea\label{dmaxres} 
d_{\max} = l-1 
\eea
The unique $N$-excluded triple at $ d_{\max } $ is 
\bea\label{deepestred} 
( k=0 , [1^m] , [ 1^{ n-l +1} ] ) \in \BRT  ( m   , n - l +1 ) 
\eea 
which exists for $ n \ge (l-1)$.
\end{proposition}

\noindent 
\begin{observation}\label{ObsdepRBDB} An immediate consequence is that the RBD of  $B_N ( m , n ) $ has exactly $ l$ layers, for $ n \ge (l-1)$,  with depth label $d$  ranging from $ 0\le  d \le (l-1)$. Any green mode in the RBD must connect to a red at a greater depth (Definition \ref{DefRBDB}), so there cannot be any greens at depths greater than $ ( l-1)$. 
\end{observation} 

Consider diagrams at depth $d$, irreps of $B_N( m , n - d ) $, with labels $ ( k , \gamma_+ , \gamma_- )$. Let $| \gamma_+ |$ and $ | \gamma_- | $ be the number of boxes in $ \gamma_+ , \gamma_-$ respectively. 
\bea\label{defk} 
&& | \gamma_+ | = m - k \cr 
&& | \gamma_- | = n - k  -d 
\eea
Let $ c_1 ( \gamma_{ \pm} )$ be the length of the first column of $ \gamma_{ \pm} $. Red diagrams obey 
\bea\label{ctff} 
c_1 ( \gamma_+ ) + c_1 ( \gamma_- ) > m + n - l 
\eea
We also have the inequalities 
\bea\label{boxhtin}
&& c_1 ( \gamma_+ )  \le | \gamma_+ | \cr 
&& c_1 ( \gamma_- )  \le | \gamma_- |
\eea
From \eqref{defk} we have 
\bea\label{defk1}  
| \gamma_+ | + | \gamma_- |  = ( m+n - d - 2k ) 
\eea
From \eqref{boxhtin} we have 
\bea 
| \gamma_+ | + | \gamma_- |  \ge c_1 ( \gamma_+ )  + c_1 ( \gamma_- ) 
\eea
which implies, using \eqref{ctff} that 
\bea\label{gpgmgt}
 | \gamma_+ | + | \gamma_- |  > (m+n-l) 
\eea
Therefore 
\bea 
&& ( m+n - d - 2k )  > ( m +n - l ) \cr 
&& \implies d + 2k < l 
\eea 
$d$ is maximised when $ k =0$. Hence the upper bound  on $d$ is $ l-1$. 
\bea\label{dmaxder} 
d_{ \max } = l-1 
\eea

Now define
\bea 
&& | \gamma_+ | - c_1 ( \gamma_+ )  =  | \gamma_+ \setminus c_1 | \cr
&& | \gamma_- | - c_1 ( \gamma_-)  =  | \gamma_- \setminus c_1 |
\eea
These are the numbers of boxes left after we remove the first column of $ \gamma_+ , \gamma_-$ respectively. They are greater or equal to $0$ : 
\bea 
&& | \gamma_+ \setminus c_1 | \ge 0  \cr 
&& | \gamma_-\setminus c_1 | \ge 0 
\eea 
 From \eqref{defk1} 
\bea 
&& c_1 ( \gamma_+ ) + c_1 ( \gamma_- ) + | \gamma_+ \setminus c_1 | +  | \gamma_-\setminus c_1 |  
= ( m +n - d - 2k ) \cr 
&& \implies c_1 ( \gamma_+ ) + c_1 ( \gamma_- ) 
= ( m +n - d - 2k ) - | \gamma_+ \setminus c_1 | -  | \gamma_-\setminus c_1 |  
\eea
Now the inequality \eqref{ctff} becomes 
\bea\label{ctff1} 
( m +n - d - 2k ) - | \gamma_+ \setminus c_1 | -  | \gamma_-\setminus c_1 |   > m + n - l 
\eea 
which simplifies to 
\begin{equation}\label{simpImp} 
\boxed{ 
~~~ d + 2k + | \gamma_+ \setminus c_1 | +  | \gamma_-\setminus c_1 |  < l ~~~
} 
\end{equation}

To maximise $d$, we must have $k=0$, and 
\bea 
&& | \gamma_+ \setminus c_1 | = 0 \cr 
&& | \gamma_-\setminus c_1 |  = 0 
\eea 
This gives the additional information that the maximal depth irrep which disappears at $ N = m+n - l $ has a pair of Young diagrams which each have all the boxes in the first column. Hence the diagram at $ d = d_{ \max } = l -1 $ is 
\bea\label{deepestRed}  
( k=0 , [1^m] , [ 1^{ n-l +1} ] )
\eea 
Alongside \eqref{dmaxder}, this completes the proof of the proposition. \hfill $\blacksquare$

\vskip.2cm 

\noindent 
{\bf Remarks } 
\begin{enumerate} 

\item The simple equation \eqref{simpImp} is very useful. We can also get an upper  bound on $k$ for the reds. To get a least upper bound we need $ d =1$ (which is the smallest $d$ that has the reds) and $ | \gamma_+ \setminus c_1 | =  | \gamma_-\setminus c_1 |  =0$, i.e. the Young diagrams $ \gamma_+ , \gamma_-$ are single columns. 
\bea 
2k < (l -1 )
\eea

\item When a green at depth $d$ connects to a red at $ d+1 $, it only connects to one such. To go from green to red as we go from  depth $ d$ to depth  $d+1$, the total height of $ \gamma_+ $ and $  \gamma_- $ must increase. Connections exist when, increasing the depth by $1$, we  remove a box from  $ \gamma_-$ or add a box to  $ \gamma_+$. Since we want to decrease the total height, we have to remove a box from $ \gamma_-$. For the height of $ \gamma_+$ to increase, we have to add the  box to the first column of $ \gamma_+$. This completely specifies the $ \tilde \gamma_+ $ of the red Brauer diagram $ ( \tilde \gamma_+ , \gamma_-  ) $ which connects to a green $ (  \gamma_+ , \gamma_- ) $  at depth $ d$. 

\item Our primary interest in subsequent sections is in the description of the RBD for the generic cases where $m,n$ are large enough compared to $l$. We establish the $(m,n)$-stability region $ m,n \ge ( 2l-3) $ in section \ref{mnstability}. The special cases of small $(m,n)$ outside the stable regime and for $ m,n \le (l-1)$ should generically follow from the large $(m,n)$ regime by specialising the derived formulae  for dimension corrections (of the kind computed in section \ref{sec:DimExamples}) and taking into account the vanishing of these. Appendix B studies some small $(m,n)$ cases, and the systematic study of these is left for the future. 

\end{enumerate}

\subsection{Bounds on the excess height }
\label{sec:bdsxs}  

The red diagrams have $ c_1 ( \gamma_+ ) + c_1 ( \gamma_- ) > (m +n -l ) $. Let 
\bea\label{defDeltxs} 
c_1 ( \gamma_+ ) + c_1 ( \gamma_- ) = (m +n -l ) + \Delta 
\eea
where $ \Delta $ is defined as the excess total height of $ \gamma_+ $ and $ \gamma_-$. 

\begin{proposition}\label{bndexcess}  
 The excess height of a triple $ ( k , \gamma_+ , \gamma_ - ) $ at depth $d$ is bounded by the minimum of $ d $ and $ (l-d)$
\begin{equation}\label{UPD} 
\boxed{  
 ~~~\Delta ( d  ) \le \min  ( l - d , d )  ~~~ 
 } 
\end{equation}
\end{proposition} 

The equation \eqref{ctff1} becomes : 
\bea 
( m +n - d - 2k ) - | \gamma_+ \setminus c_1 | -  | \gamma_-\setminus c_1 |   = m + n - l  + \Delta 
\eea 
The inequality \eqref{simpImp} then becomes 
\begin{equation}\label{simpEq}  
\boxed{ 
 d + 2k + | \gamma_+ \setminus c_1 | +  | \gamma_-\setminus c_1 |  + \Delta   = l ~~~
} 
\end{equation}
Since $ k \ge 0 ,| \gamma_{ \pm } \setminus c_1 | \ge 0  $, this implies that 
\bea 
\Delta \le  l - d 
\eea
At $ d = d_{\max } $, $ \Delta = 1 $, which we know from the form of the maximal depth red. 
At $ d = d_{ \max} -1 $, $ \Delta \le 2 $. 

There is another constraint on $ \Delta $ for the restricted Bratteli diagram. In this diagram, there is no red node at $ d =0$. At $ d=1$, the restricted Bratteli diagram  only keeps red nodes which connect to greens at $ d=0$, which have $ \Delta \le  0$. Going from $ d =0$ to $ d=1$ along an arrow, we add a box to $ \gamma_+$ or remove a box from $ \gamma_-$. This means that the maximum $ \Delta $ for irreps at $ d =1$ is $1$. The maximum $ \Delta $ for depth $d$ in the truncated diagram is $ d $. 

Combining with above 
\bea\label{UPD} 
 \Delta ( d  ) \le \min  ( l - d , d ) 
\eea 
This completes the proof of the proposition \ref{bndexcess}. \hfill  $\blacksquare$. 

\subsection{Stability of  the  RBD for $B_N ( m , n ) $ for large enough $m,n$ } 
\label{mnstability}

Inspection of restricted Bratelli diagrams   of  $B_N(m,n)$ has, for fixed $l$ (with $N = m+n=l$), shows that when $m,n$ are large enough compared to $l$, the diagrams take he same form independent of $m,n$. We illustrate with the diagrams  for a sequence of $m,n$ and $l=4$ in Figure~\ref{sequence} below. 

\begin{figure} 
	\centering
\includegraphics[scale=0.35]{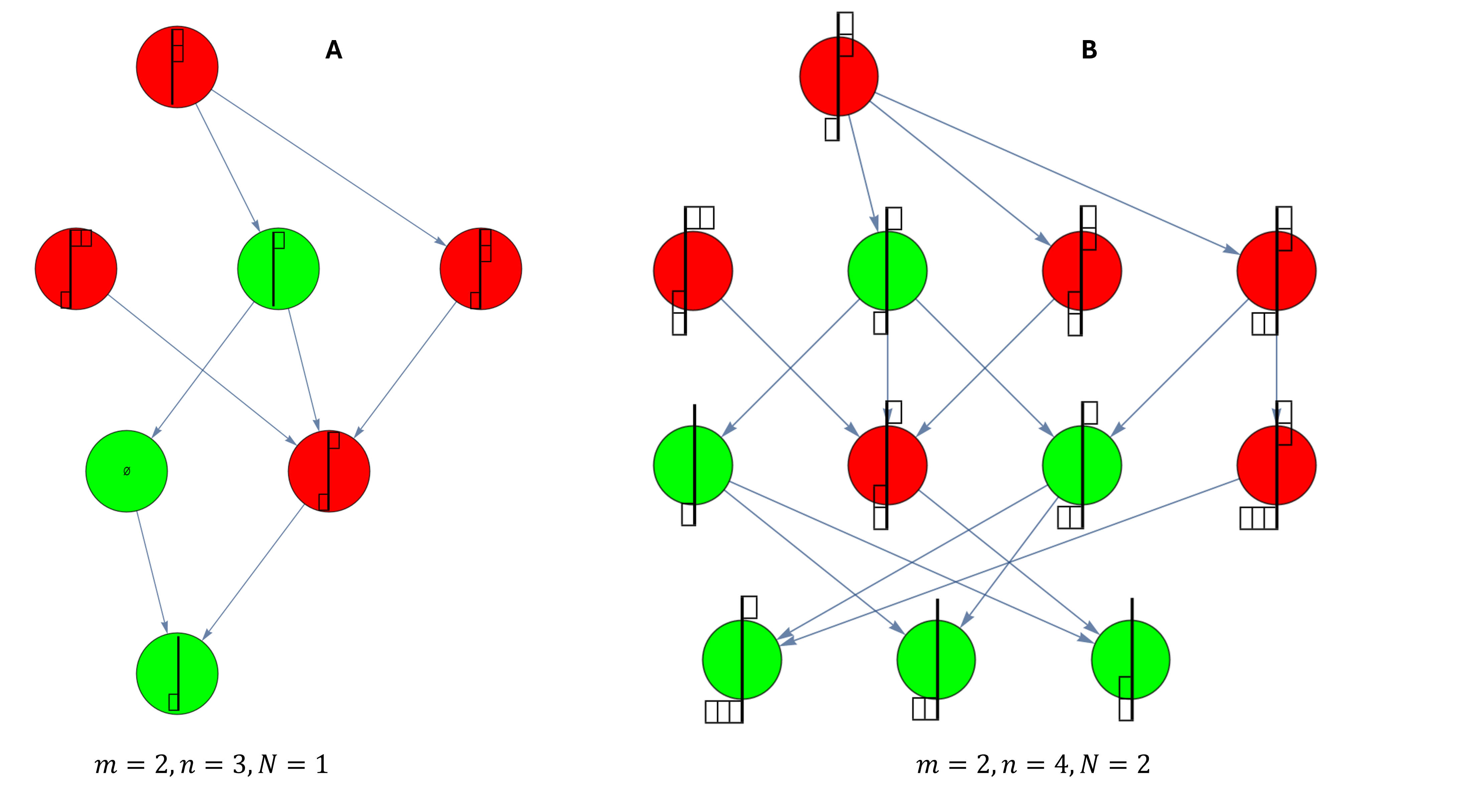}\\
\includegraphics[scale=0.38]{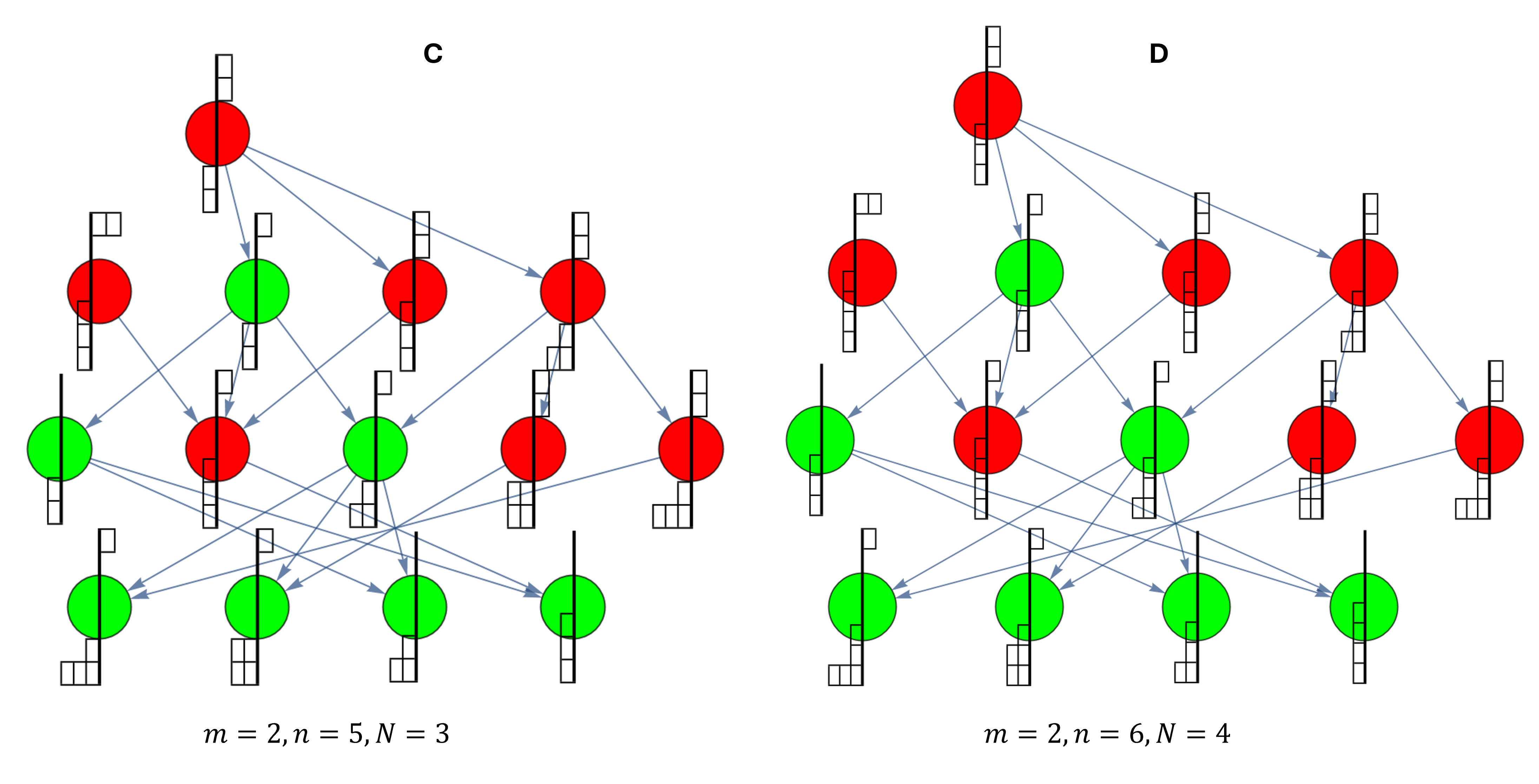}\\
\includegraphics[scale=0.38]{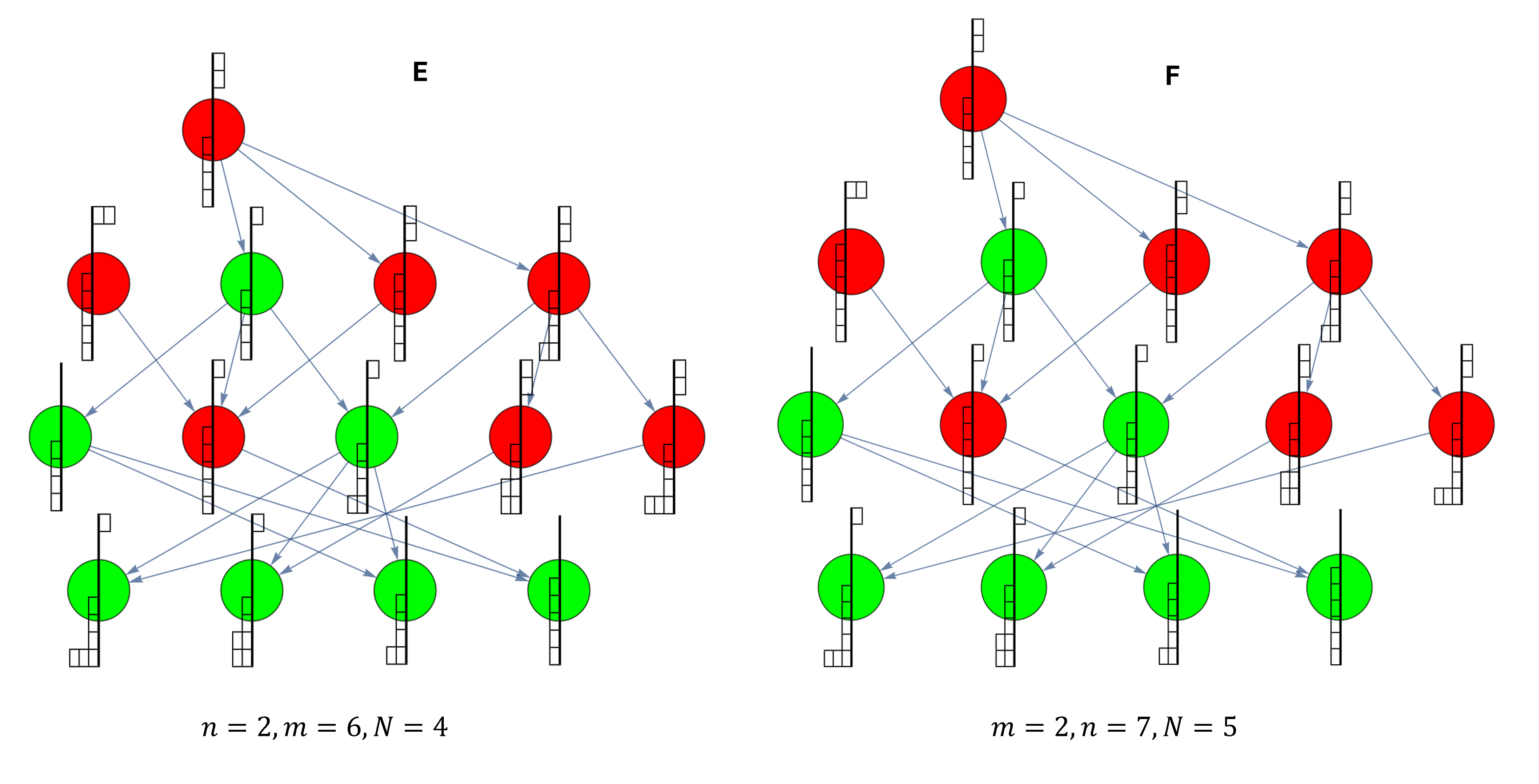}
\caption{For fixed $m=2$ and $l=4$ we present a sequence of RBDs from $n=3$ to $n=7$. We achieve the stable region for $n\geq 5$, which agrees with Proposition~\ref{mnstab}. } 
\label{sequence}
\end{figure} 

Ad $m,n$ are decreased away from the stable regime, a subset of the red Brauer triples disappear. We derive the stable range by further building on the results obtained in the previous section. Specifically, we begin by analyzing the saturated inequality~\eqref{simpImp}, as expressed in~\eqref{simpEq}. Given that all the parameters involved in the equation are positive, and taking into account the allowed Bratteli moves when transitioning between different depths $d$, we can deduce the shapes of the corresponding Young diagrams that label the walled Brauer irreducible representations satisfying the given equality.

\begin{proposition}\label{mnstab} 
The restricted Bratteli diagram of $ B_N ( m , n ) $, with $ N = (m+n-l)$  
has a form independent of $ (m,n)$ for large enough $m,n$ in the range  
\begin{equation} 
\boxed{ 
 ~~~~ m \ge ( 2l -3 )   ~~ ;~~ n \ge ( 2l -3). ~~~~ 
} 
\end{equation} 
\end{proposition} 

From $ \eqref{simpEq}$ we see, using $ k \ge  0 ,  | \gamma_- \setminus c_1  | \ge 0 $  that there is a bound  for red Brauer triples 
\bea 
| \gamma_+ \setminus c_1 |  \le  l - d - \Delta. 
\eea 
This bound is strongest for the lowest possible values   $d=1$, $ \Delta =1$ allowed for red triples, where we have 
\bea 
| \gamma_+ \setminus c_1 |  \le  l - 2. 
\eea
Choose the maximum value 
\bea 
| \gamma_+ \setminus c_1 |  =   l - 2, 
\eea 
which can only occur if $ k =0 $ and $ | \gamma_- \setminus c_1 | =0  $.  
Further  consider  the Young diagram $ \gamma_+$ where all these $ (l-2)$ boxes are in the second column, i,e.  $  ( \gamma_+ \setminus c_1)  = [ 1^{  l-2}] $. 
 Then the red triple  is $ ( k =0  , [ 2^{l-2} ,  1^{ m - 2l + 4 } ] , [ 1^{ n-1} ] ) $. 
 This red triple  at $ d =1$ connects to a green triple  at 
$d=0$ which is $ ( k=1 ,  [2^{l-2} ,  1^{ m - 2l +3  } ] , [ 1^{ n-1} ] ) $. This can only exist if 
\bea 
m \ge ( 2l - 3 ).
\eea
Similarly we can derive $ n \ge ( 2l -3 )$. 

If we consider  instead   a Young diagram $  ( \gamma_+ \setminus c_1 )  $ 
which has fewer rows than the maximum $ ( l-2)$, i.e. $ ( l - 2 - t)  $ for some $ t > 0 $, then $ \gamma_+ \setminus c_1 $ can be written as 
\bea 
&& \gamma_+ \setminus c_1   = [ 1^{ l-2-t }  ; ( \gamma_+\setminus  \{ c_1 , c_2 \}  )  ],  \cr 
&& | ( \gamma_+\setminus  \{ c_1 , c_2 \}  ) | = t, 
\eea
where  $ ( \gamma_+\setminus  \{ c_1 , c_2 \}  ) $ is the Young diagram obtained from $ \gamma_+$ by deleting the first column $ c_1 $ and the second column $ c_2 $. The 
Brauer triple  is then
\bea 
 ( k =0  , [2^{ l - 2 - t } ,  1^{  m - 2 l + 4 + t  } ;   ( \gamma_+\setminus \{ c_1 , c_2 \}  )  ]    , [ 1^{ n-1} ] ).
\eea 
This red diagram connects to   a green 
  \bea 
     ( k =1 , [2^{ l - 2 - t } ,  1^{  m - 2 l + 4 + t  } ;   ( \gamma_+\setminus \{ c_1 , c_2 \}  ) ]    , [ 1^{ n-1} ] ).  \eea 
This leads to  a weaker constraint $ m \ge ( 2l - 3 - t ) $. 

Thus   that the RBD of  $B_N(m,n)$ has, for fixed $l$,  a universal form independent of 
$m,n$ when 
\bea 
&& m \ge ( 2l - 3 ), \cr 
&& n \ge ( 2l -3 ). 
\eea
This completes the proof of the proposition \ref{mnstab}. \hfill $\blacksquare$

 This $(m,n)$-stability of the restricted Bratteli diagrams of $ B_N (m,n)$ implies corresponding stability for 
the set of irreps of $B_N ( m, n ) $ which have modified dimensions. The dimension-modifications
have $m,n,l$ dependences which are computable from the structure of the $(m,n)$-stable
restricted Bratteli diagrams of  $ B_N (m,n)$, and known dimension formulae from the $N$-stable regime of $ N \ge (m+n)$.

\subsection{General formula for counting of reds }\label{gencountreds}

The counting of Brauer representation triples for $ (m,n ) $ is 
\bea 
\hbox{BRT } ( m , n ) = \sum_{ k =0 }^{ \min (m , n ) } p ( m-k ) p ( n -k) 
\eea
where $p( m ) $ for any integer $m \ge 1$ is the number of partitions of $m$, while $p(0) $ is defined as $1$.  This is easy to understand since Brauer triples for $ (m, n ) $ are triples of the form $ ( k , \gamma_+ , \gamma_- ) $ where $ 0 \le k \le \min ( m , n ) $ and 
$ \gamma_+ \vdash (m-k) , \gamma_- \vdash (n-k) $. 

We start with the equation \eqref{simpEq} which, for convenience, we rewrite here : 
 \begin{equation}  
 d + 2k + | \gamma_+ \setminus c_1 | +  | \gamma_-\setminus c_1 |  + \Delta   = l 
\end{equation}
From the non-negativity of $k, | \gamma_+ \setminus c_1 | ,  | \gamma_- \setminus c_1 | $ we deduce that $ \Delta \le  l - d $. As explained in the proof of Proposition \ref{bndexcess} there is a stronger bound $0 \le  \Delta \le \min ( d , l-d )$ for red diagrams in the RBD of  $ B_N(m,n)$. The diagrams $\gamma_+ \setminus c_1  $ 
and  $ \gamma_- \setminus c_1$ have no constraints beyond \eqref{simpEq} in the stable regime 
$ m,n \ge (2l-3)$.

This leads to a counting formula for the reds, as a function of  for $l \equiv ( m+n - N )  $ and  depth  $d$, in the stable regime ( $ m, n \ge ( 2l -3 ) $ ). Let $l_1 = | \gamma_+ \setminus c_1 | $ and $ l_2 = | \gamma_- \setminus c_1 |$ . It follows from \eqref{simpEq}  that 
\bea 
l_2  = l - d - 2k - l_1  - \Delta 
\eea
Let $ p(l)$  be the number of partitions of $l$. The counting function for red nodes is obtained by considering the product  $ p(l_1) p( l_2) $ and summing over the allowed ranges of $l_1 , \Delta , k $ 
\bea\label{Redcount1}  
\cR ( l, d )  = \sum_{ \Delta = 1 }^{ \min ( l -d  , d ) }  ~ \sum_{ k  = 0 }^{ \lfloor { (  l - d - \Delta )  \over 2 } \rfloor  }  ~ \sum_{ l_1 =0 }^{  l  - d - 2 k - \Delta } 
p ( l_1 )   p ( l - d - 2k  - \Delta - l_1 ) 
\eea 
We can also write this as 
\bea\label{Redcount}  
\cR ( l , d )  = \sum_{ k =0 }^{  \lfloor {  ( l- d )  \over 2 } \rfloor } ~  \sum_{ \Delta =1 }^{ \min ( l - d - 2k  , d ) }  ~
\sum_{ l_1 =0}^{   l  - d - 2 k - \Delta  } 
p ( l_1 )   p ( l - d - 2k  - \Delta - l_1 ) 
\eea

\section{Simple  harmonic oscillators  and the counting of red nodes  in the RBD of $B_N ( m,n)$ }
\label{sec:derSHO}  

In section \ref{sec:SHOhdreds}, we establish an equation relating the counting of high depth  red nodes in the RBD to the partition function $\cZ_{ \rm univ} (x)$  of a tower of harmonic oscillators \eqref{ResDeepUniv}. A more non-trivial equation relates this partition function to the low-depth red nodes is found in section \ref{sec:lowdepOsc}. The counting of red nodes 
at each depth $d$ can be further refined according to the excess $ \Delta $ which measures the extent to which $ { \rm ht} ( \gamma ) $ exceeds $N$ (see the definition \eqref{defDeltxs}). 
In section \ref{sec:DelEq1Osc}  we relate the counting of the $ \Delta =1 $ red nodes to the oscillator partition function.

\subsection{Simple Harmonic oscillators and high-depth red nodes  }
\label{sec:SHOhdreds}

We start from \eqref{simpEq} and rewrite in terms of $s$ defined by  $ d = ( l - s)  $ with $ s \in \{ 1,2, \cdots , l \} $. We also define 
\bea 
l_1 = | \gamma_+ \setminus c_1 | \cr 
l_2 =  | \gamma_- \setminus c_1 | 
\eea   
The equation   \eqref{simpEq} becomes 
\bea\label{simpEq2} 
 \Delta + 2k + l_1  + l_2 = s 
\eea
The upper bound on $ \Delta $ in \eqref{UPD} is expressed as 
\bea
\Delta \le \min ( l - d , d ) = \min ( s , l - s ) 
\eea
For $ s \le \lfloor { l \over 2 } \rfloor  $, equivalently $ d \ge \lceil { l \over  2 } \rceil $,  we have  
\bea\label{simplowsbd} 
\Delta \le s 
\eea 
This inequality is  implied by \eqref{simpEq2} by taking into account $  k , l_1 , l_2 \ge 0$. 
We define $ \cR_{ \dep}  ( s) $ to be the number of red nodes in this region of parameters 
$ s \le \lfloor { l \over 2 } \rfloor  ; ( 2l-3) \le \min ( m , n ) $.  Note that neither the equation \eqref{simpEq2} nor the bound \eqref{simplowsbd} have any  dependence on $l$. This means that for small enough $s$, equivalently large enough $d$,  $ \cR( l , d ) $ is independent of $l$.

A formula for $ \cR_{ \dep } ( s ) $  is  obtained by multiplying the number of partitions $ ( \gamma_+ \setminus c_1 )  $ of $ l_1$ with the number of partitions $ ( \gamma_- \setminus c_1 )  $ of $l_2$, subject to the constraint \eqref{simpEq2} 
\bea\label{Rdeps1}  
\cR_{  \dep }  ( s ) = \sum_{ l_1 =0 }^{ s } \sum_{ l_2 = 0  }^{ s - l_2 }  \sum_{ \Delta =1 }^{ s -l_1 - l_2 } \sum_{ k =0}^{ \lfloor {  ( s - l_1 - l_2 - \Delta ) \over 2 } \rfloor  } 
p ( l_1 ) p ( l_2 ) 
\eea
By re-arranging the sums,  while obeying the constraint in \eqref{simpEq2}, this can also be written as 
\bea\label{Rdeps2} 
\cR_{  \dep }  ( s ) = \sum_{ \Delta =1 }^s 
 \sum_{ k =0}^{ \lfloor {s - \Delta \over 2 }  \rfloor   } \sum_{ l_1 =0 }^{ s  - \Delta - 2k } 
p ( l_1 ) p ( s - \Delta - 2k - l_1  )
\eea
Now  we define the generating function 
\bea 
\cR_{ \dep }  ( x ) &  = & \sum_{ s =1  }^{ \infty }  x^{ s } ~~ \cR_{ \dep }  ( s ) 
\eea 
We have 
\bea 
\cR_{ \dep }  ( x ) &= & \sum_{ s=1}^{ \infty } x^s \sum_{ l_1 =0 }^{ s } \sum_{ l_2 = 0  }^{ s - l_2 }  \sum_{ \Delta =1  }^{ s -l_1 - l_2 } \sum_{ k =0}^{ \lfloor {  ( s - l_1 - l_2 - \Delta ) \over 2 } \rfloor  } 
p ( l_1 ) p ( l_2 )  \delta ( s , l_1 + l_2 + 2k + \Delta ) \cr 
& = & \sum_{ s=1}^{ \infty } x^s \sum_{ l_1 =0 }^{ \infty  } \sum_{ l_2 = 0  }^{ \infty }  \sum_{ \Delta =1  }^{ \infty } \sum_{ k =0}^{ \infty  } 
p ( l_1 ) p ( l_2 )  \delta ( s , l_1 + l_2 + 2k + \Delta )
\eea
We have removed 
 the upper bounds on the $ l_1 , l_2 , \Delta , k $ sums since they are enforced by 
the $ \delta $  function. We change the order of summation, doing the  sum over $s$ first  and use the delta function to trivially do the sum over $s$
\bea\label{genfunOsc} 
 \cR_{ \dep }  ( x ) & = & \sum_{ l_1 =0 }^{ \infty  } \sum_{ l_2 = 0  }^{ \infty }  \sum_{ \Delta =1 }^{ \infty } \sum_{ k =0}^{ \infty  } 
p ( l_1 ) p ( l_2 ) x^{ l_1 + l_2 + \Delta + 2k } \cr 
& =  & \left ( \sum_{ l_1 = 0 }^{ \infty } x^{  l_1 } p( l_1 ) \right ) 
    \left ( \sum_{ l_2 = 0 }^{ \infty } x^{  l_2 } p( l_2  ) \right )  \sum_{ \Delta =1 }^{ \infty } x^{ \Delta }  \sum_{ k =0}^{ \infty  } x^{ 2k }  \cr 
    &  = &  { x \over ( 1 - x ) ( 1 - x^2 ) }  \prod_{ i =1  }^{ \infty }  { 1 \over ( 1 - x^i)^{ 2 }  } 
\eea
In the final step we have employed the usual generating functions for partitions. 
The result in \eqref{genfunOsc} is the generating function given in OEIS A000714 \cite{OEISA000714}. 
The product  $ \prod_{ i =1  }^{ \infty }  { 1 \over ( 1 - x^i)  } $  gives the  partition function for a tower of oscillators $ \{  A^{ \dagger}_{ i } :  i \in \{  1, 2, \cdots \} = \mathbb{N  } \}  $, where $ A^{ \dagger}_{ i} $ creates a particle of energy $i$. The full generating  function is the partition function for the following system of oscillators 
\bea 
&&  \{  A^{ \dagger}_{ i } :  i \in \{  1, 2, \cdots \} = \mathbb{N  } \}  \cr 
&& \{  B^{ \dagger}_{ i } :  i \in \{ 1, 2, \cdots \} = \mathbb{N  } \}  \cr 
&& C^{ \dagger }_{ 1 } \cr 
&& D^{ \dagger}_{ 2} 
\eea  
where the subscript of each oscillator gives its energy. This is equivalent to the description given in OEIS which describes the sequence as giving the 
``Number of partitions of $n$, with three kinds of 1 and 2 and two kinds of 3,4,5,....''.
We will refer to $ \cZ_{ \rm univ }  ( s ) $ as the { \bf universal sequence} and the known 
partition function from OEIS is 
\bea\label{Zuniv} 
\cZ_{ \rm univ } ( x ) = \sum_{ s=0 }^{ \infty } x^s ~\cZ ( s ) = { x \over ( 1 - x ) ( 1 - x^2 ) }  \prod_{ i =1  }^{ \infty }  { 1 \over ( 1 - x^i)^{ 2  }  } 
\eea 
As  we will see, $ Z_{ \rm univ } ( x  ) $ has a number of applications in the context of restricted Bratelli diagrams. To summarise the first such connection we have derived : 
\bea\label{ResDeepUniv} 
\boxed{ 
\cR_{ \dep } ( x ) = \cZ_{ \rm univ  } ( x ) 
} 
\eea
 
For the case $ l = 18 $, the list $ \cR (  l , d ) $ for $ s =  l - d $ with $s$ increasing from $ 1 $  to $ l $ is 
\bea  
 \{1, 3, 9, 21, 47, 95, 186, 344, 620, 1075, 1814, 2950, 4623, 6869, \
9489, 11523, 10409 \} 
\eea 
For $ l = 20 $, the list $ \cR (  l , d ) $ for $ s =  l - d $ with $s$ increasing from $ 1 $ to $19$ is 
\bea 
&& \{1,3,9,21,47,95,186,344,620,1078,1832,3024,4872,7603,11456,16425,21932,25815,22639\}
\cr 
&& 
\eea 
For $ l = 19 $, the list $ \cR (  l , d  ) $ for $ s =  l - d $ with $s$ increasing from $ 1 $ to $19$ is 
\bea 
&& \{1, 3, 9, 21, 47, 95, 186, 344, 620, 1077, 1826, 2998, 4781, 7327, \
10699, 14503, 17345, 15406 \} \cr 
&& 
\eea 
The universal sequence $ \cZ_{ \rm univ } (  s  ) $ up to $12$ is : 
\bea
\{ 1, 3, 9, 21, 47, 95, 186, 344, 620, 1078, 1835, 3045 \} 
\eea 
Up to the $ 9^{\rm th}$ term, this agrees with the $ l=18$ and $ 19$ sequences.  
Up to the $ 10^{\rm th}$ term, it agrees with the $l=20$ sequence. This is as expected from the derivation of the universal sequence  above as the near-high depth region defined by  $ s \le \lfloor { l \over 2 } \rfloor  $. 

It is worth noting that the partition function $ \cZ_{ \rm univ } ( x ) $ is strongly suggestive of a two-dimensional  free field theory interpretation. Connections  between the combinatorics of 
$U(N)$ representations and two dimensional field theory have been studied in the context of gauge-string duality for large $N$ two-dimensional Yang-Mills theory, see for example equation 
(3.3) in \cite{Douglas1993}, which involves the infinite product part of $ \cZ_{ \rm univ } (x ) $ without the extra factors $ {1 \over ( 1 -x ) ( 1- x^2 ) } $.  A two-dimensional field theory interpretation of $ \cZ_{ \rm univ } (x ) $, and its connections to restricted Brauer diagrams being developed here,  is an interesting problem for the future.

\subsection{Universal oscillator partition function  and low-depth red nodes } 
\label{sec:lowdepOsc} 

The counting of red nodes in \eqref{Redcount1}  is given by a formula of the form 
\bea 
\cR ( l , d ) = \sum_{ \Delta =1 }^{ \min ( l -d , d ) } \cR ( l , d, \Delta ) 
\eea 
where $ \cR ( l , d , \Delta ) $ only depends on $ ( l-d) $. 
For $ d $ near $l$, i.e.  $ \min ( l -d , d ) = (l-d)$,  equivalently $ \lceil { l \over 2 } \rceil     \le  d \le l-1  $,  
\bea 
\cR ( l , d )  \rightarrow \cR_{ \rm  deep }  ( l -  d )  =
 \sum_{ \Delta =1 }^{  l -d } \cR ( l , d, \Delta ) 
\eea
and the relation to oscillator counting derived as \eqref{ResDeepUniv} is 
\bea 
\cR_{ \rm  deep   }  ( l -  d ) = \cZ_{ \rm univ  } ( l - d  )
\eea 
Now consider the counting of low-depth red nodes, i.e. consider $ \cR_{ \sha } ( l , d ) $ for small $ d $ where $ \min ( l -d , d ) = d $ equivalently $ d \le \lfloor { l \over 2 } \rfloor  $ 
\bea 
\cR_{ \sha } ( l , d ) = \sum_{ \Delta =1 }^{  d  } \cR ( l , d, \Delta ) 
\eea 
In this range for $d$, we have $ ( l - d ) \ge  d $ and we can write 
\bea 
&& \cR_{ \sha } ( l , d ) =  \sum_{ \Delta =1 }^{  l -d  } \cR ( l , d, \Delta )  - 
\sum_{ \Delta = d +1 }^{ l - d } \cR ( l ,  d , \Delta ) 
\eea 
Defining 
\bea 
\cR_{ 12} ( l , d )  = \sum_{ \Delta = d +1 }^{ l - d } \cR ( l ,  d , \Delta ) 
\eea 
we can write 
\bea 
\cR_{ \sha } ( l , d ) =  \cR_{\rm  deep } ( l - d )   - \cR_{ 12} ( l , d ) 
\eea 
The first term is related to the oscillator count as we showed earlier  \eqref{ResDeepUniv}. 
The second term can be re-written, using \eqref{Redcount1}, as   
\bea 
&& \cR_{ 12} ( l , d  )  = \sum_{ \Delta  = d+1 }^{ l -d } \sum_{ k =0}^{ \lfloor {  l - d - \Delta \over 2 }  \rfloor   }  ~  \sum_{ l_1 = 0 }^{ l - d - 2k - \Delta } 
p (l_1 ) p ( l- d - 2k - \Delta - l_1 ) \cr 
&& = \sum_{ D =1 }^{ l - 2d } ~ \sum_{ k =0}^{ \lfloor {  l - 2 d - D \over 2 } \rfloor    }  ~ \sum_{ l_1 = 0 }^{ l - 2 d - 2k - D } 
p (l_1 ) p ( l- 2 d - 2k - D - l_1 )
\eea 
where we defined $ D = \Delta - d $.  Comparing the last line with \eqref{Redcount1}, specialised to the high-depth limit, we observe that 
\bea 
\cR_{ 12} ( l , d ) = \cR_{ \dep} ( l - 2d ) = \cZ_{ \rm univ } ( l - 2d) 
\eea 
We conclude that 
\bea 
\cR_{ \sha } ( l , d ) = \cZ_{ \rm univ } ( l - d ) - \cZ_{ \rm univ } ( l - 2d ) 
\eea
We have thus related the counting of red nodes of the restricted Bratteli diagrams,  in both the shallow and the deep regions to 
the universal oscillator partition function $ \cZ_{ \rm univ}  ( x ) $ in \eqref{genfunOsc}. We may summarise as 
\begin{equation}
\boxed{ 
\mathcal{R}(l, d) = 
\left\{
\begin{array}{ll}
\mathcal{Z}_{\text{univ}}(l - d) - \mathcal{Z}_{\text{univ}}(l - 2d) & \text{for } 1 \leq d \leq \left\lfloor \frac{l}{2} \right\rfloor \\
\mathcal{Z}_{\text{univ}}(l - d) & \text{for } \left\lceil \frac{l}{2} \right\rceil \leq d \leq (l - 1)
\end{array}
\right.
} 
\label{RedsOscSum}
\end{equation}

\subsection{Universal oscillator partition function  and  $ \Delta =1$ red nodes.    } 
\label{sec:DelEq1Osc} 

In this section we show that the counting of red nodes with $ \Delta =1$ has a simple relation to the universal oscillator partition function.
Recall the definitions 
\bea 
&& \cR ( l , d , \Delta  ) = \hbox { Number of red nodes  at depth $d$, with $ N = m+n -l$, } \cr 
&& \hbox{  and with $ c_1 ( \gamma_+ ) + c_1 ( \gamma_-) = N + \Delta $ } 
\eea
From \eqref{Redcount1}, we have 
\bea 
\cR ( l , d , \Delta  )  = 
\sum_{ k=0}^{ \lfloor { ( l - d - \Delta ) \over 2 }  \rfloor }  \sum_{ l_1 = 0 }^{ l - d - \Delta - 2k } p ( l_1) p ( l - d - \Delta - 2k - l_1 ) 
\eea 
Specialising to $ \Delta =1 $ 
\bea 
\cR ( l ,  d , \Delta =1 ) 
= \sum_{ k=0}^{ \lfloor { ( l - d - 1 ) \over 2 }  \rfloor }  \sum_{ l_1 = 0 }^{ l - d - 1 - 2k } p ( l_1) p ( l - d - 1 - 2k - l_1 ) 
\eea

Define $ \cR ( l , d^{ +} , \Delta =1 ) $ to be the number of red nodes with 
$ \Delta =1$ at depths $ d' > d $ : 
\bea 
&& \cR ( l , d^{ +} , \Delta =1 ) = \sum_{ d' = d+1 }^{ l -1 } \cR (l ,  d , \Delta =1  )  \cr 
&& = \sum_{ d' = d +1  }^{ l-1 } \sum_{ k=0}^{ \lfloor { l - d' -1 \over 2  } \rfloor  } 
 \sum_{ l_1 = 0 }^{ l - d' - 1 - 2k }p ( l_1) p ( l - d' - 1 - 2k - l_1 ) 
\eea
Defining a shifted variable 
\bea 
\tilde d = d' - d 
\eea
to get 
\bea 
 \cR ( l , d^{ +} , \Delta =1 )
= \sum_{ \tilde d = 1 }^{ l - 1 -d } \sum_{ k=0}^{ \lfloor { l - \tilde d -  d -1 \over 2  } \rfloor  } 
 \sum_{ l_1 = 0 }^{ l - d - \tilde d  - 1 - 2k }p ( l_1) p ( l - \tilde d - d  - 1 - 2k - l_1 ) 
\eea 
The universal oscillator count $ \cZ_{ \rm univ } ( s ) = \cR_{ \dep } ( s ) $ is, using \eqref{Rdeps2}
\bea 
 \cZ_{ \rm univ } ( s ) = \sum_{ \Delta =1}^{ s } \sum_{ k=0}^{ \lfloor {s -  \Delta \over 2 } \rfloor } 
 \sum_{ l_1=0}^{ s - \Delta - 2k } 
p ( l_1 ) ~  p ( s  - \Delta -  2k  - l_1 ) 
\eea 

Consider  the evaluation of $ \cZ_{ \rm univ  }  ( s ) $ at $s= l - 1 -d   ) $ while renaming $ \Delta \rightarrow \tilde d $ to find 
\bea 
 \cZ_{ \rm univ   } (  l-1  -d  ) = \sum_{ \tilde d  =1}^{ l-1 - d   } \sum_{ k=0}^{ \lfloor {l-1  - d -  \tilde d  \over 2 } \rfloor } 
 \sum_{ l_1=0}^{ l-1  - d -  \tilde d  - 2k } 
p ( l_1 ) ~  p ( l-1   - 2k -  d - \tilde d  - l_1 ) 
\eea 
We observe that 
\bea\label{Redsdpluniv}  
\cZ_{ \rm univ  } (  l-1  -d  ) =  \cR ( l , d^{ +} , \Delta =1 )
\eea 
Setting $d=0$, we conclude that the number of red nodes in the restricted Bratteli diagram with $ \Delta =1$ at all depths $  1 \le d \le (l-1) $ is given by the oscillator partition function 
\bea\label{Reds0pluniv}   
\boxed{ 
 \cR ( l , 0^{ +} , \Delta =1 ) = \cZ_{ \rm univ   } (  l-1   )
 } 
\eea

We will show  that the counting of green nodes in the restricted Bratteli diagrams in the stable range $ m,n \ge 2l -3 $, equivalently the counting of Brauer triples with modified dimensions for 
$ N  = m+n - l $, is given precisely by  $ \cR ( l , 0^{ +} , \Delta =1 )$. This gives a direct link between the counting of Brauer triples with modified dimension and the universal oscillator partition function.

\section{ Counting of green nodes }
\label{sec:greennodes} 

Let us denote the number of green nodes in the restricted Bratteli diagram for $B_N ( m , n ) $, with $ N = ( m +n - l ) $, in the $(m,n)$-stable regime of $ \min ( m , n ) \ge 2l -3 $, as $ \cG ( l  , d ) $. In this section we give a  proof that the number of green nodes at depth $0$, i.e. $ \cG ( l , 0 ) $ is equal to the number of red nodes with $ \Delta =1$ at depths $  1 \le  d \le ( l-1)$. 
\bea 
\cG ( l , d =0 )  =  \cR ( l , 0^{ +} , \Delta =1 ), 
\eea 
Using the relation between red nodes with  $ \Delta = 1 $ at $ 1 \le d \le ( l-1)$  and  the oscillator counting function in \eqref{Reds0pluniv} we conclude that 
\bea\label{GreensUniv}  
\boxed{ 
\cG ( l , d =0 )  = \cZ_{ \rm univ   } (  l-1   )
} 
\eea 
A simple argument generalises this to 
\bea\label{GreensUniv1} 
\cG ( l , d  )  =\cR ( l , 0^{ +} , \Delta =1 ) =  \cZ_{ \rm univ   } (  l-d - 1   )
\eea 
Our strategy for proving \eqref{GreensUniv} is to establish a bijection between the set of  red nodes with $ \Delta =1$ for any $d$ and the set of green nodes at $ d =0$. 

\subsection{ Bijection between the greens and the $ \Delta =1$ reds at higher depths }

Every green node at depth $d=0$ in the RBD for $B_N (  m,n)$ is connected  to some set of reds at higher depth in the diagram. For a green node $g$, let $ d_{ \min } ( g) $ be the minimal value  of $d$, such that the layers at $d$ contains a red node connected to $g$. We refer to such nodes as minimal depth red ancestors of $g$.  Green nodes have $ \Delta \le 0 $ while reds have $ \Delta \ge 1$. The 1-box Bratelli moves which connect a diagram at some depth to another diagram at the next depth can change $ \Delta $ by $1$ at most. It follows that a minimal depth red ancestor of a given green  node must have $ \Delta =1$.
 
Inspection of a few nontrivial  RBDs  shows that 
each green $g$ at $  d=0$ has a unique red ancestor at the minimal depth $ d_{ \min  } ( g ) $. 
We refer to this as the minimal-depth red ancestor of $g$. Consider Figure  
\ref{Gen_mnk4}. The unique first-red-ancestor associated with each green node at $d=0$ is: 
\bea 
&& \gamma_1 \leftrightarrow  \gamma'_2, \cr 
&& \gamma_2 \leftrightarrow  \gamma'_1,  \cr 
&& \gamma_3 \leftrightarrow \gamma'_{ 10 }, \cr 
&& \gamma_4 \leftrightarrow \gamma'_{ 3 },  \cr 
&& \gamma_5 \leftrightarrow \gamma'_{ 9 },  \cr 
&& \gamma_6 \leftrightarrow \gamma'_{ 8  },  \cr 
&& \gamma_7 \leftrightarrow \gamma'_{ 7  },  \cr 
&& \gamma_8 \leftrightarrow \gamma'_{ 5 },  \cr 
&& \gamma_9 \leftrightarrow \gamma'_{ 4}.
\eea
The diagrams appearing above all have $ \Delta =1$, i.e. height $9$. The diagram $ D_{ 6} $ has $ \Delta =2$ and does not appear above. 

Inspection of the RBDs further shows that the path from a green at $ d=0$ to its first-red ancestor proceeds through successive increases of  $ \Delta $ by $1$. We will prove both properties, in generality, in the following and these properties will serve to establish the bijection between greens at $d=0$ and red nodes with $ \Delta =1$ at $ d \ge 1 $.

\subsubsection{ Bijection : The first red ancestor and the fastest descent of $ \Delta $.  } 

Following an arrow in the restricted Bratelli diagram  from a Young  diagram at depth $d$ to one at  depth $d-1$ results in a change of $ \Delta $ by one of  $\{ -1 , 0 , 1 \} $. 
A change of $-1$ results from removing  a box from the first column of $ \gamma_+$. 
A change of $0$ results from removing a box from $ \gamma_+ \setminus c_1 $, i.e. removing a box from  the second or higher column, or from adding a box to $ \gamma_- \setminus c_1$. 
A change of $ +1$ results from adding a box to the first column of $ \gamma_-$. 

We may summarise as follows, using $ \delta^- $ with superscript $ -$ to indicate that this is in the direction of decreasing depth
\bea\label{depthdecpathtypes}  
\delta^{ -} ( \Delta ) = -1  & \iff &  \{  \delta^- ( c_1 ( \gamma_+ ) ) = -1  \} \cr 
\delta^{ -} ( \Delta ) = 0  & \iff &  \{  \delta^- ( [ \gamma_+ \setminus c_1 ]  ) = -1 ~~ \hbox{OR} ~~ 
\delta^- ( [ \gamma_- \setminus c_1 ]  ) = +1 \}  \cr 
\delta^{ -} ( \Delta ) = +1  & \iff  &  \{ \delta^- ( c_1 ( \gamma_- ) ) = +1  \} \cr 
&& 
\eea
It is convenient to assign  names to the  four different possibilities above : 
\bea\label{depthdecpathtypes1}  
A^{- }  & \equiv &  \{  \delta^- ( c_1 ( \gamma_+ ) ) = -1  \} \cr 
B_1^{ -}  & \equiv  &  \{  \delta^- ( [ \gamma_+ \setminus c_1 ]  ) = -1  \}  \cr 
B_2^{ -}  & \equiv  & \{  \delta^- ( [ \gamma_- \setminus c_1 ]  ) = +1 \}  \cr 
C^{ -}  &  \equiv &   \{  \delta^- ( c_1 ( \gamma_- ) ) = +1 \}  \cr 
&& 
\eea 
 
Conversely following an arrow in the restricted Bratelli diagram   (RBD) from a Young  diagram at depth $d$ to one at  depth $d+1$ results in a change of $ \Delta $ by one of  $\{ -1 , 0 , 1 \} $. 
A change of $-1$ results from removing  a box from the first column of $ \gamma_-$. 
A change of $0$ results from removing a box from $ [\gamma_- \setminus c_1]  $, i.e. removing a box from  the second or higher column of $ \gamma_-$, or from adding a box to $ [  \gamma_+ \setminus c_1] $. 
A change of $ +1$ results from adding a box to the first column of $ \gamma_-$. 

We may summarise as follows, using $ \delta^+ $ with superscript $ + $ to indicate that this is in the direction of decreasing depth : 
\bea\label{deltplusmvs}  
\delta^{ +} ( \Delta ) = -1  & \iff &  \delta^+ ( c_1 ( \gamma_- ) ) = -1 \cr 
\delta^{ +} ( \Delta ) = 0  & \iff &  \{  \delta^+ ( [ \gamma_+ \setminus c_1 ]  ) = +1 ~~ \hbox{OR} ~~ 
\delta^+ ( [ \gamma_- \setminus c_1 ]  ) = -1 \}  \cr 
\delta^{ +} ( \Delta ) = +1  & \iff  & \delta^+( c_1 ( \gamma_+ ) ) = +1 \cr 
&&
\eea
It is convenient to give some names to the  four different possibilities above : 
\bea 
A^{+ }  & \equiv &  \{  \delta^+ ( c_1 ( \gamma_- ) ) = -1  \} \cr 
B_1^{ + }  & \equiv  &  \{  \delta^+ ( [ \gamma_+ \setminus c_1 ]  ) = +1   \}  \cr 
B_2^{ + }  & \equiv  & \{  \delta^+  ( [ \gamma_- \setminus c_1 ]  ) = -1 \}  \cr 
C^{ + }  &  \equiv &   \{  \delta^+ ( c_1 ( \gamma_+ ) ) = +1 \}  \cr 
&& 
\eea 

Inspection of the RBDs shows that the path from a green node at $ d=0$ to its first red ancestor 
proceeds through a sequence of $C^+$ moves.  For  for example $\gamma_9$ in Figure \ref{Gen_mnk4}  which involves 
\bea 
\gamma_9 = ( [1^3 ] , [ 1^3 ]  ) \rightarrow  ( [1^4]   , [1^3 ]  ) \rightarrow ( [1^5]  , [1^3] )   \rightarrow ( [1^6]  , [1^3] )  = \gamma'_4 .
\eea 
The next Lemma proves that this is a general property. 

\noindent 
{\bf Lemma 9.1   } The path from a green at $ d=0 $ to its first red ancestor  in the restricted Bratteli diagram always proceeds by a sequence of $ C^+$ moves. 

\begin{itemize} 

\item The first red ancestor of a green at $d=0$  is defined as a red Young diagram at the smallest depth which admits a path to the specified green.

\item Since greens have, by definition,  $ \Delta \le  0$ and reds have $ \Delta \ge 1$ and the links in the TBD have $  | \delta^{ \pm} ( \Delta )  | = 1$, it follows that the first red ancestor has $ \Delta =1$. 

\item Given any diagram the $ C^+$ move, of adding a box to the first column,  is always well-defined and it results in a unique diagram. Thus any green at $ d=0$ has a first  red ancestor with $ \Delta =1$  which is obtained by a succession of $ C^+$ moves. 

\end{itemize} 

This completes the proof of the Lemma \hfill $\blacksquare$. 

\vskip.2cm 

{\bf Lemma 9.2  :  Uniqueness/injectivity } The first red ancestor of a  green node $g$ is unique. There is no other red at  depth  $ d = d_{ \min } ( g) $ which connects to the specified green.

\begin{itemize} 

\item For a specified green with $ \Delta = \Delta_* \le 0$, the red ancestor obtained by a sequence of $ C^+$ moves occurs at $d = | \Delta_* |  +1 $. 

\item Note from \eqref{deltplusmvs} that $C^+$ is the only of the $ \delta^+$  moves linking nodes at depth $d$ to nodes at depth $(d+1)$ which changes $ \Delta $ by $ +1$. 
If  a red ancestor connects to the specified green by $ \delta^+$ moves  which are not all $ C^+$
then at least one of the moves belongs to   $\{  A^+ , B_1^+ , B_2^+ \}$ which have 
$ \delta ( \Delta ) < 1$,  so this red ancestor must occur at a higher depth.

\item Therefore the  path to the first red ancestor gives an injective (one-to-one) map  between greens with at $ d=0$ with $ \Delta_* $ and reds at $ d = \Delta_* +1$ and with $ \Delta =1$.

\end{itemize} 

This completes the proof of the Lemma. \hfill $\blacksquare$. 

\vskip.2cm 

\noindent 
{\bf Lemma 9.3 :  Surjectivity. } We now  show that the map from  the set of green nodes  at $d=0$ to the set of $ \Delta =1$ reds at any depth $  1 \le d \le ( l-1)$ is surjective. 
 Given any red $Y_*$  at some depth $d = d_*$ with $ \Delta =1$, we can find a green at $ d=0$ which connects to $Y^*$ via  a sequence of $ C^+$ moves.  

\vskip.4cm

 First of all recall that by the definition of the RBD, every red diagram at some depth $d_*$ 
 in it connects to at least one green at $ d =0$, hence contributes to the modification of the  dimension of the associated mixed Young diagram. The connecting path can proceed, in general, via a combination of reds and greens  in intermediate depths $ d $ in the range $ 0 < d < d_*$. 
 Proceeding from $ d_*$ to $ d_* -1 $ along a link in the RBD. 

\vskip.1cm 
The hypothesis of surjectivity is that every such $Y_* $ has a path consisting of a sequence of $ A^-$  steps (inverse of $C^+$) leading from $ Y_*$ to a green with at $ d=0$, having $ \Delta = 1 - d_*$.  

\vskip.2cm 

 The $A^-$ step  decreases  $ \Delta  $ by $1$, decreases $d$ by $1$ and increases $ k $ by $1$. It keeps $ ( d + \Delta + 2k ) $ fixed, and as mentioned earlier consists of removing a box from the first column of $ \gamma_+$. 

\vskip.2cm 
To prove that $Y_*$ admits such a link to a green at  $ d_*-1$ with $ \Delta =0$, requires proving that the second column of $ \gamma_+ ( Y_* ) $ is strictly shorter in length than the first column of $ \gamma_+ ( Y_* ) $ : if this does not hold, removing a box from the first column of $ \gamma_+ ( Y_* )$   would not produce a valid Young diagram. 

\vskip.2cm 
Using \eqref{simpEq} we know that 
\bea 
| \gamma_+ \setminus c_1 | = l - d - 2k - \Delta - |\gamma_- \setminus c_1 |
\eea
This means 
\bea\label{upbdgpmin} 
&& | \gamma_+ \setminus c_1 | \le  l - d - 2k - \Delta \cr 
&& \implies | \gamma_+ \setminus c_1 | + d + 2k + \Delta \le l 
\eea
We also have from the definition of $ \Delta $ in equation  \eqref{defDeltxs} that 
\bea 
&& c_1 ( \gamma_+ ) + c_1 ( \gamma_- ) = m + n - l + \Delta  \cr 
&& \implies c_1 ( \gamma_+ ) = m + ( n - c_1 ( \gamma_- ) ) - l + \Delta 
\eea
and at depth $d$ the Young-diagram pairs are irreps of $ B_N ( m , n-d) $ giving 
\bea 
&& c_1 ( \gamma_- ) \le ( n - d )  \cr 
&& \implies ( n - c_1 ( \gamma_- ) ) \ge d 
\eea
The last two equations imply  
\bea 
c_1 ( \gamma_+ )   \ge m + d + \Delta - l 
\eea 
In the stable regime (Proposition \ref{mnstab}), we have  $ m \ge ( 2l -3 ) $. 
\bea\label{lowbdc1gp} 
&& c_1 ( \gamma_+ ) \ge l + d + \Delta - 3  \cr 
&& \implies c_1 ( \gamma_+ )  - d - \Delta +3 \ge l 
\eea
Combining \eqref{upbdgpmin} and \eqref{lowbdc1gp} we have 
\bea\label{c1andbeyond}  
&& c_1 ( \gamma_+ )  - d - \Delta +3   \ge | \gamma_+ \setminus c_1 | + d + 2k + \Delta   \cr 
&& \implies  c_1 ( \gamma_+ ) \ge | \gamma_+ \setminus c_1 | + 2d + 2k + 2 \Delta -3  
\eea 
For $  \Delta = 1 $, 
\bea 
&& c_1 ( \gamma_+ ) \ge | \gamma_+ \setminus c_1 | + 2k + 2 ( d -1) + 1 
\eea 
Further imposing $ d \ge 1 $, we have 
\bea 
 c_1 ( \gamma_+ )  \ge  | \gamma_+ \setminus c_1 |  + 2 k  +1 
\eea  
This means that for any $ Y_*$, the second column in $ \gamma_+$ has to be shorter than the 
 the first column, so it admits the $ A^-$ move. 

\begin{comment} 
Now if  $Y_* $ has $ d =2 , \Delta =1$, then $ A^- Y_* $ has $ k \ge 1 , \Delta =0, d =1$. With these conditions, \eqref{c1andbeyond} implies that $ A^- Y_*$ has 
\bea 
 c_1 ( \gamma_+ )  > | \gamma_+ \setminus c_1 | 
\eea 
and a box can be removed from $ c_1 ( \gamma_+ ( A^- Y_* ) ) )$. 
\end{comment} 
Now we again use the fact that the $A^-$ move decreased $ d, \Delta $ by $1$ while increasing $ k$ by $1$. For $Y^*$ at a general depth  $ d_* $, which has to obey $ d_*  \le ( l-1 )$, the diagram 
$ ( A^-)^{ p } Y_*  $, with $ p \le ( d_*  -1 ) $, obeys 
\bea 
&&  k \ge p   \cr 
&&  \Delta = ( 1 - p ) \cr 
&& d =  ( d_*  - p )  \cr 
&&  2d +  2 \Delta + 2k  = 2 d_* - 2p + 2 - 2 p + 2 k \ge 2 d_* + 2 - 2p  
\eea 
Using  $ d_* \ge ( p + 1 ) $ we have 
\bea 
2d + 2k + 2 \Delta \ge 4 
\eea
Then with \eqref{c1andbeyond}, 
\bea 
c_1 ( \gamma_+ ) \ge | \gamma_+ \setminus c_1 | +1 
\eea  
which means that the entire sequence $ ( A^-)^p Y_* $, for $ 0 \le p \le d_* -1 $  allows the removal of a box from the first column. 

This proves that every  red node $Y_*$ admits a first-ancestor-path for some green at $ d =0$,  completing the proof of Lemma 9.3. \hfill $\blacksquare$.

Together with  the injectivity proof from above, this completes the proof that the first-ancestor-paths give a bijection between the greens at $ d =0$ and the reds with $ \Delta =1$ at $ d \ge  1$. 

 The inspection of the the RBDs also suggests the generalisation of this 
result to the greens at some more general $d$. They admit a bijection to the set of reds at depths $d+1$ and higher, thus leading to the identity in equation \eqref{GreensUniv1}. 
 The result for the greens at depth $ d  $ in the $ B_N  ( m , n ) $ 
diagram follows from the $d=0$ result for $ B ( m , n - d)$.

\section{Discussions  and Outlook. }

We outline some technical generalisations of the present work and discuss broader related future research in connection with the recent literature in gauge-string duality and quantum information theory. 

A general formula for the modified dimensions $ \widehat{ d}_{ m , n, N  } $ is an interesting goal.  With $ N = m+n-l$, 
extending  the results we have derived beyond the range  $l \le 4$ up to $ l = 15 $ say should be  possible by employing more sophisticated computational techniques based on this paper. For the general case $(m,n,N)$ case, finding a good formula or efficient combinatorial rule is a worthwhile objective,  likely to have many applications. Consideration of the $(m,n)$-stable regime of $m,n \ge (2l-3)$ is likely to be a fruitful approach. 

The explicit construction of matrix basis elements $ Q^{ \gamma }_{ IJ  } $ (also called matrix units)  for the semisimple quotient $ \widehat B_N ( m, n ) $ of $ B_N( m,n)$, which is isomorphic to the commutant $   \cA_{m,n}^N$  of the unitary group action in mixed tensor space, is of interest in matrix quantum mechanics (see e.g. \cite{KR,KRT09,KimuraQuarter}) and quantum information theory ( see e.g.~\cite{grinko2024efficientquantumcircuitsportbased,StudzinskiIEEE22,grinko2025phd}). For Brauer triples $ \gamma$ with $ c_1 ( \gamma_+ ) + c_1 ( \gamma_- )  $ close to $ (m+n)$, an evident extension of the present work is the explicit construction of the kernel $ I_N( m,n ) $ of the map $ \rho_{ N  , m , n } : 
B_N ( m,n) \rightarrow \End ( \Vmn  ) $, for $ N = ( m+n -l ) $ with $l $ small. These explicit constructions will shed light on the precise form of the algebraic relations between the non-semisimple algebra $ B_N ( m, n ) $ and the semisimple quotient $\widehat B_N ( m, n )$ . A more ambitious but still plausible goal is to start from  these explicit constructions of the Kernel defined by taking $ N = (m+n -l)$, and then change the Brauer product parameter for these Kernel elements from $N$ to a value of $N'$ in the semisimple regime. This would be a new approach to the efficient construction of matrix units potentially applicable in a regime of large $(m,n)$ but for matrix units labelled  by Brauer triples $ \gamma $ having $ c_1 ( \gamma_+ ) + c_1 ( \gamma_- ) $ close to $(m+n)$. 

Non-trivial checks of the computation of modified dimensions in section \ref{sec:DimExamples}
were performed based on the relation between the Kernel of $ \rho_{ N , m , n} $ 
and $ \rho_{ N , m+n  } $ given by partial transposition of matrices $ X_1 , X_2 , \cdots , X_{ m+n} $. Similar relations  exist if we work with just two matrices instead of $(m+n)$ matrices. In this case there will be identities analogous to \eqref{identityDims} but now involving Littlewood-Richardson coefficients on one side and reduction coefficients from $ B_N ( m,n) $ to 
$ S_m \times S_n $ on the other side. This follows from the fact that orthogonal  bases of 2-matrix invariants can be constructed using matrix units  for 
the sub-algebra of $ B_N ( m,n) $ which commutes with $ S_m \times S_n $  \cite{KR} or 
the sub-algebra of $ \mC ( S_{ m+n} ) $ which commutes with $ S_m \times S_n $ \cite{BCD0801,BDS0805}.  The algebraic description of the these orthogonal bases and related $U(2) $ covariant ones \cite{BHR1,BHR2,QuivCalc}  in terms of permutation centraliser algebras was explained in \cite{MatRam2016} and have been used to develop algebraic eigenvalue-based algorithms for these bases \cite{EigVal}. Extending the calculations of this paper to the 
study of the sub-algebra of $B_N ( m,n)$ which commutes with $S_m \times S_n$, for the non-semisimple regime $ N < (m+n)$ is an interesting project.  The $S_m \times S_n$  symmetry also  appears in quantum computation in the context of  the higher order quantum operations (HOQO) as  parallel configuration between inputs and outputs \cite{taranto2025higherorderquantumoperations}.

The simplicity of the regime $m,n \ge  ( 2l-3) $, with the stable form of the restricted Bratteli diagrams independent of $(m,n)$ and with the relations of the counting of red and green nodes in terms of a tower of harmonic oscillators, has some analogies to other large-parameter simplifications which have been studied in connection with matrix invariants and their associated dual objects, notably giant gravitons, in AdS/CFT. For holomorphic invariants of one complex matrix, the orthogonal bases of matrix invariants labelled by Young diagrams allow a map 
\cite{CJR}  to giant graviton \cite{MST,GMT0008,HHI0008}  configurations. This has been extensively evidenced (see \cite{YJKW2103,BudGai} for  recent discussions  with earlier references). For large Young diagrams having large row (or column)  lengths and large differences between successive rows (or columns), which are interpretable  in terms of well-separated giant gravitons, fluctuations described by 2-matrix invariants lead to a picture of giant graviton oscillators \cite{GigraOsc,DoubCos} which arise from solving the one-loop dilatation operator in this sector. It would be very interesting to investigate if an analogous role for the oscillators uncovered here can be found in terms of fluctuations of brane-anti-brane systems related to matrix invariants of $ Z , \overline { Z } $. The rich physics of the   $ Z , \overline { Z } $ matrix system has been found to include negative specific heat capacities \cite{DOCSR} suggestive  of interpretations in terms of small black holes in AdS \cite{Han2016,Ber2018}, and hidden 2d free fields in the physics of small black holes is a fascinating prospect.

The setting of lower-dimensional gauge-string dualities is also a promising avenue for exploring the deeper physical implications of the oscillator system we have found here. The large $N$ expansion of the dimensions of $U(N)$ irreducible representations arising in the decomposition of $ \Vmn$ are used in developing the  dual string theory picture of 2d Yang Mills theory with $U(N)$ gauge group \cite{GrTa1,GrTa2} (an extensive review including exposition of the important role of Schur-Weyl duality in this expansion is in \cite{CMR}). Identities for these dimensions
 following  from walled  Brauer algebras were used to give a holomorphic reformulation of the large $N$ expansion \cite{KRHol} at the expense of introducing line defects. As mentioned the 
 partition function $ \cZ_{ \rm univ } (x) $  in \eqref{Zuniv}, stripped of the factors  ${1\over (1-x) (1-x^2) } $, has been discussed in the 
 context of 2d Yang Mills theory and the boson-fermion description of its states \cite{Douglas1993}.  The relevant boson-fermion correspondence has been studied in depth (see  for example  \cite{Mandal} \cite{MinPoly} \cite{DasJev} and refs. therein).  
 The dual string theory has a number or proposed worldsheet actions \cite{CMR1,CMR,Horava} and has  seen a revival of interest in recent years \cite{Ofer,JanYM2,ShotaYM2}. An interesting question is whether  there is a stringy interpretation, in the 2d Yang Mills context,  of the oscillator partition function studied here.   The structure of finite dimensional diagrammatic non-semisimple  associative algebras have also been  used to inform indecomposable  representations of the Virasoro algebra arising in logarithmic CFT \cite{GJRSV2013}. This may give another avenue for potential applications of the structures uncovered here to low-dimensional quantum field theory. 

Our work, along with its possible extensions, may also be useful in the context of quantum information science and optimisation problems which can be recast in the form of semidefinite programs (SDP)~\cite{10.1088/978-0-7503-3343-6}.
 
In contrast to the central interest in large $N$ in  high-energy physics,  in quantum information and quantum computing we deal with a large number of involved systems $m+n$ and every system separately often has a small dimension $N$ : for example we  can use qubits ($N=2$) for various quantum information processing tasks~\cite{RevModPhys.81.865}. Now, if our tasks exhibit symmetries induced by the algebra $B_N(m,n)$ it is natural to apply the representation theory to simplify the problem. However, we immediately land in the highly non-semisimple region for the considered algebra, where the description is more involved. In this regime having precise description of the kernel in terms of the operator units acting on the space $\Vmn$ would lead us to relaxing  the complexity of SDP problems~\cite{Grinko2024Lin,Nechita2023,StudzinskiIEEE22, Ebler_2023}. Namely, many  quantum information related SDP problems have matrix constraints which can be reduced to matrix constraints involving irreducible matrix units. By knowing their properties, specifically the dimensions in the non-semisimple regime, we can reduce the complexity of finding solutions by lowering the dimension of the space on which our problem is defined. We elaborate more on this in the context below.

In general, we hope our toolkit with its further extensions will give a chance for better understanding the theory of higher-order quantum operations (HOQO)~\cite{taranto2025higherorderquantumoperations} by exploring underlying symmetries and help to answer on other unsolved  problems. In particular, this involves addressing the open problem of obtaining a single copy of $U^T$, where $U \in SU(N)$, given $k$ calls to the black-box program implementing $U$. In this particular problem, we have to go beyond the parallel strategies and consider adaptive or even indefinite ones to construct optimal quantum combs and calculate their efficiency~\cite{Ebler_2023}. Up to know the scientific community considered quantum combs giving output of the form of $\overline{U}, U^{\dagger}$~\cite{miyazaki17,PhysRevLett.131.120602,10089837}, but even here, for $U^\dagger$, answers are known partially~\cite{yoshida2024a,mo2024parameterizedquantumcombsimpler,chen2024quantumadvantagereversingunknown}. Such protocols can be naturally extended to the situation where, for a given number of calls $U^{\otimes k}$, we want to produce some number of $(U^T)^{\otimes l}, (U^\dagger)^{\otimes l}, \overline{U}^{\otimes l}$, where $l<k$. This kind of situation can occur when we want to apply a simultaneously transformed quantum program on several systems.


Let us elaborate more on this by focusing on potential application in enhancing the efficiency of the quantum unitary programming and the storage and retrieval (SAR) mechanisms~\cite{PhysRevA.81.032324,sedlak18,lewandowska2024} for quantum programs compared to existing implementations by including multicopy quantum teleportation~\cite{grosshans2024multicopyquantumstateteleportation}. From the previous papers we know that the algebra $B_N(m,n)$ can be presented as a chain of of inclusions of ideals. Each such ideal is determined by the number of arcs in its elements~\cite{Cox1}. For example, for the highest ideal, when the number of arcs is maximal (i.e. $n$ if we assume $n<m$) the matrix units in this ideal are $S_m \times S_n$ adapted by the construction. This fact allowed for elegant description of the multi-port-based teleportation and unitary programming when $n=1$~\cite{PhysRevA.85.022330,grosshans2024multicopyquantumstateteleportation} and multi-port based teleportation for any $n>1$~\cite{StudzinskiIEEE22}.  This was possible because these quantum primitives require knowledge only about the highest ideal. To enhance the efficiency of the unitary programming/SAR schemes one needs to go to lower ideals, even with only one arc, and construct there the $S_m \times S_n$ adapted basis~\cite{StudzinskiIEEE22,WBA2024}. Only then we are able to use the full power of symmetries contained in the problem and possibly solve corresponding optimisation problems. Having them solved, one can reconstruct optimal quantum memory state to store quantum information and compute optimal efficiency of its retrieval.

Another interesting detour is strictly connected quantum circuit designing. Namely, we can ask what is the form of the global unitary transformation performing basis change from the computational basis to the respective group adapted irreducible spaces? This problem should be possible to solve in the most general framework assuming that the required Littlewood-Richardson coefficient are known. We are aware that computing of the Littlewood-Richardson coefficients is in general a $\#P$-complete problem~\cite{Narayanan2006,Liu2019}. However, once we assume that they are given for a specific initial values of parameters, we can ask is it possible to represent efficiently corresponding unitary as a quantum circuits, similarly as it was done for unitary transformation to the Schur basis~\cite{bacon2007quantum,Krovi2019efficienthigh,Kirby_2018} or the mixed Schur basis~\cite{nguyen2023mixedschurtransformefficient,fei2023efficientquantumalgorithmportbased,PRXQuantum.5.030354,grinko2024efficientquantumcircuitsportbased}. Presenting such circuits even for small number of $m+n$ and comparing with the Schur transform should be very interesting and instructive. 

Algebraic studies of quantum complexity in connection with projector detection tasks motivated by the matrix description of brane systems in AdS/CFT have been conducted in \cite{ProjDetect}. Extending these  to Brauer algebra projectors and matrix units is an interesting area for future  investigation. Further study of connections between phenomena in high-energy physics and quantum information implied by the  common appearance of walled Brauer algebras in both settings promises to be a highly fruitful area. An area of application could be, for example, to find applications of results in  the complexity of the mixed Schur transforms to the complexity of simulations of high-energy processes.

\section*{Acknowledgments}

It is a pleasure to acknowledge useful conversations with Dimitri Grinko,  Michał Horodecki, Denjoe O' Connor, Adrian Padellaro, Ryo Suzuki.  We extend our special thanks to the authors of \cite{grinko2023gelfandtsetlinbasispartiallytransposed} for sharing the mathematica code used in the paper. We acknowledge useful conversations with ChatGPT (OpenAI)  on the mathematica syntax for specified computational tasks, which helped expedite the development of the code described in Appendix A. SR is supported by the Science and Technology Facilities Council (STFC) Consolidated Grant ST/T000686/1 
``Amplitudes, strings and duality''. SR grateful for a Visiting Professorship at Dublin Institute for Advanced Studies,  held during 2024, when this project was initiated.  SR also gratefully acknowledges a visit to the Perimeter Institute in November 2024: this research was supported in part by Perimeter Institute for Theoretical Physics. Research at Perimeter Institute is supported by the Government of Canada through the Department of Innovation, Science, and Economic Development, and by the Province of Ontario through the Ministry of Colleges and Universities.
The work of MS is  is carried out under IRA Programme, project no. FENG.02.01-IP.05-0006/23, financed by the FENG program 2021-2027, Priority FENG.02, Measure FENG.02.01., with the support of the FNP.

\section*{Data availability}
Manuscript has associated data in a data repository - Restricted-Brattelli-Diagrams.nb code available on arxiv as an ancillary file to the manuscript (see: https://arxiv.org/abs/2509.04234).

\section*{Conflict of interest}
The authors declare no conflict of interest.

\newpage 

\begin{center} 
{ \LARGE \bf Appendices } 
\end{center} 

\begin{appendix}

\section{ Guide to Mathematica code for restricted Bratelli diagrams (RBD)  for $B_N ( m ,n ) $   }\label{appsec:Mtca}  

This section explains the code in the notebook Restricted-Bratteli-Diagrams.nb which is provided as auxiliary material alongside the arXiv submission~\cite{RamgoolamStudzinski2025WBAcode}. The code builds on  a selection of definitions from the mathematica code accompanying \cite{grinko2023gelfandtsetlinbasispartiallytransposed}, which constructs Bratteli diagrams for 
$B_N ( m , n ) $. These Bratteli diagrams are graphs with nodes organised according to levels starting from $0$ to $ (m+n)$. The  nodes of the final level $L$  are associated with mixed Young diagrams of $B_N ( m , n ) $. We define depth $d$ as $ d = L - m - n $. Thus the final layer is depth $d=0$. For the $N < (m+n)$, some of the diagrams are excluded by the finite $N$ constraint 
\eqref{finiteNconst}. These are kept in the RBD but their nodes are colored red. The remaining nodes are coloured green.  The code works as follows :

\begin{enumerate}

\item Identifies all the red nodes. 
Finds the highest depth $d_{ \max} $ which contains a red node. 

\item  The command {\rm BuildSkeltFwd} constructs  a list of lists. The first entry  is the list of  red nodes at $d = d_{ \max}$. 
It turns out from the theoretical considerations we have explained (section \ref{sec:deepest}) that there is a single red node at $ d_{ \max} $. The code does not use this as input. The   second entry  is built to contain the red or green nodes at $ d = d_{\max} -1 $ which are connected to the deepest reds 
at $ d = d_{ \max }  $ by the Bratteli moves of the large $N$ Bratteli diagram, and 
 any other red nodes at $d = d_{ \max} -1 $ are appended. Goes to $ d = d_{\max }  -2  $ and constructs the list of all the red and green nodes connected to the previous list of red and green $ d = d_{\max} -1$,  then adds any ther red nodes at $ d = d_{\max} -1$. 
  Repeats the procedure 
until $d =1$ is reached. The last entry is the list of   green nodes  at $ d=0$ connected to the list of nodes thus constructed at $ d=1$.

\item  The  command {\rm CBSNEW} works with the output constructed above, and starts with the $d=0$ nodes from above. Then steps  to $ d =1$, removes any red nodes unconnected to the green nodes at $d=0$. Next steps up  to $d=2$ and removes any  nodes not connected to the above constructed list at $ d =1$. Iterates reaching all the way to $ d = d_{\max} $. 

\item The command  which produces, using the above steps the restricted Bratteli diagrams 
for $ B_N ( m ,n ) $ with $ N  = m +n -l$ and $ l \ge 2$ is  $ {\rm RestrictedBratteliDiag} [ m, n , N ]  $. 

\item Another command  $ { \rm RestBratDiagSpec } [ m , n , N , a  ] $ specialises the RBD to the nodes which are connected via Bratteli moves to a specific green node at depth $d=0$. This node is  identified by an integer $ a$, which ranges from $1$ to the number of green nodes at depth $0$. 

\end{enumerate} 

The notebook Restricted-Bratteli-Diagrams.nb also contains the command ColoredBratteliDiag which constructs the colored Bratteli diagram for $ B_N ( m , n )$, where the green nodes obey the constraint \eqref{finiteNconst}  while the red nodes do not obey this constraint. Unlike the RBD, the final layer is not restricted to green nodes which acquire a modified dimension. The command is used to produce Figure \ref{Fig0}. 

\section{ Checks of the decomposition of $ \Vmn $ into irreducible of $U(N)$  for small $m,n$  } 

In this section we perform checks of the modified dimensions $ \widehat{d}_{ m,n,N} $ calculated 
in section \ref{sec:DimExamples} by directly verifying, for small values of $m,n$, the identity for dimensions of the vector spaces on the LHS and the RHS of \eqref{eq:mixedSW}.

We recall,  for convenience here the decomposition of mixed tensor space in terms of irreducible representations of $U(N)$ from section \ref{Sec:l=2}. For the stable large $N$ region \eqref{eq:mixedSW1}  gives 
\bea\label{eq:mixedSW1A}
V_N^{\otimes m}\otimes \overline{V}_N^{\otimes n}=\bigoplus_{ \gamma \in {\rm BRT } ( m,n ) } V_\gamma^{U(N)} \otimes V_{\gamma}^{B_N(m,n)}.
\eea  
In the non-semisimple regime $ N < (m+n)$ \eqref{eq:mixedSW} gives 
\bea\label{eq:mixedSWA} 
V_N^{\otimes m}\otimes \overline{V}_N^{\otimes n}=\bigoplus_{ \substack { \gamma \in \widehat{\BRT } ( m,n)   }  } V_{ \gamma }^{U(N)} \otimes V_{\gamma}^{\widehat{B}_N(m,n)},
\eea

In the stable large $N$ regime where \eqref{eq:mixedSW1A} holds, the following identity for dimensions follows 
\bea\label{NstablePowerIds} 
 N^{ m + n } = \sum_{  \gamma \in \BRT (m,n)   }  \Dim V_{ \gamma }^{ U(N) } 
 d_{m,n} ( \gamma ) 
\eea 
where $ Dim ( V_{ \gamma }^{ B_N ( m , n ) } ) = d_{ m,n }  ( \gamma ) $ is independent of $N$. 
Outside the stable regime, for $ N < (m+n)$ ,  where \eqref{eq:mixedSWA} holds, 
\bea 
\sum_{ \substack{  \gamma \in \BRT (m,n ) \\ c_1 ( \gamma_+ ) + c_1 ( \gamma_- ) \le N    }}  \Dim V_{ \gamma }^{ U(N) }  \widehat{d}_{m,n, N } ( \gamma ) 
\eea 
As discussed around \eqref{DisjBRTmnN0} and \eqref{BRThat} of section \ref{Sec:l=2}, each $N < (m+n) $ in the non-semisimple regime defines a partition of $ \BRT (m,n)$ 
\bea\label{DisjBRTmnN}  
&& \BRT (m,n) = \Excl ( m,n , N ) \sqcup \Unmod (  m , n , N )  \sqcup \Mod (  m , n , N )  \cr 
&& \widehat{\BRT} ( m , n ) := \BRT (m,n) \setminus \Excl ( m,n , N ) = \Unmod (  m , n , N )  \sqcup \Mod (  m , n , N ) \cr 
&& 
\eea
The Brauer triples in $\Mod (  m , n , N )   $ appear as green nodes at depth $0$ in  the RBD for $ B_N ( m,n) $.  There is now an identity 
\bea\label{NunstablePowerIds} 
&& N^{ m+n} =  \sum_{ \substack{  \gamma \in \BRT (m,n)  \\ \gamma \in  \Unmod (  m , n, N  )  }  }  \Dim V_{ \gamma }^{ U(N) } \dim_{ m , n }  (  \gamma  ) +  \sum_{ \substack{  \gamma \in \BRT ( m,n)  \\  \gamma \in  \Mod ( N , m , n )  }  }  \Dim V_{ \gamma }^{ U(N) } \widehat{d}_{  m , n, N    }  (  \gamma  ) \cr 
&& 
\eea 
As we have discussed earlier, $ \Mod ( N , m , n ) $ is empty for $ N = m+n -1$. 
We will illustrate these identities \eqref{NstablePowerIds}  and \eqref{NunstablePowerIds} 
by using the well-known formula for $ d_{ m,n} ( \gamma ) $ and the modified dimensions we 
have calculated in section \ref{sec:DimExamples}, for some special cases with $(m,n) = (2,1) $, $ (m,n) = (2,2) $ and $ (m,n) = (3,3)$. 

In the following we will verify the identities \eqref{NstablePowerIds} and \eqref{NunstablePowerIds} explicity for small values of $m,n$. The irrep 
$V_{ \gamma }^{ U(N) } $ has a highest weight $ \lambda $ given by the positive and negative 
row lengths of the mixed Young diagram $ \Gamma ( \gamma , N ) $.
The  dimensions
$ \Dim V_{ \gamma }^{ U(N) } $   are calculated using the Weyl dimension formula, which 
for $ \lambda = \Gamma ( \gamma , N )  = [ \lambda_1 , \lambda_2 , \cdots , \lambda_N ] $ reads as 
\bea 
\Dim V_{ \gamma }^{ U(N) }  = \prod_{ 1 \le i <  j \le N }  \frac{\lambda_i - \lambda_j - i + j }{j-i} 
\eea

\subsection{ The case $m=2,n=1$ } 

\begin{table}[h]
\centering
\begin{tabular}{|c|c|c|c|}
\hline
 $ \gamma $ &  mixed Young diagram $\Gamma ( \gamma , N )  $  & $ \Dim V_{ \gamma }^{ U(N) } $ & $ d_{ m,n} ( \gamma ) $  \\
\hline
 $( 0 , [2] ,[1] )$  & \rule{0pt}{0.7em}$[\,2,0^{\,N-2},-1\,]$ 
 & $\displaystyle \frac{N(N+2)(N-1)}{2}$ 
& 1 \\[0.4em]
$ ( 0 , [1^2] ,[1] )$ & $[\,1,1,0^{\,N-3},-1\,]$ 
& $\displaystyle \frac{N(N+1)(N-2)}{2}$ 
& 1 \\[0.4em]
 $ ( 0 , []  , []  ) $  & $[\,1,0^{\,N-1}\,]$ 
& $\displaystyle N$ 
& 2  \\
\hline
\end{tabular}
\caption{List of irreps of $U(N)$, with highest weight $\Gamma ( \gamma , N ) $  appearing in decomposition of $\Vmn$ for $(m,n)=(2,1)$. The  stable range multiplicities are the 
Brauer irrep dimensions $ d_{ 2,1} ( \gamma )$.}
\label{DecompDatameq2neq1}
\end{table}
Using the data given in the table \ref{DecompDatameq2neq1}, the stable range identity \eqref{NstablePowerIds} is 
\bea 
N^3 =  \frac{N(N+2)(N-1)}{2} +  \frac{N(N+1)(N-2)}{2} + 2 N  
\eea 
Now consider what happens when $ N = 2$, the first value in the unstable range. 
This $ N = m+n - l $ with $ l=1$, for which there are excluded representations but no modified dimensions. The excluded representation is  $ ( 0 , [1^2] ,[1] )$ and the identity is 
\bea 
2^3 =  \biggl[ \frac{N(N+2)(N-1)}{2}  + 2N \biggr ]_{ N =2 } = 4 +4 
\eea 
Next we have $ N=1$ or $ l=2$ for which we have excluded representations $ \gamma = (  0 , [2] ,[1] ) $ and $ \gamma = ( 0 , [1^2] ,[1] )$, and one representation $ \gamma =   ( 0 , [] ,[]  ) $ with dimension modification. The modification as given 
by  \eqref{leq2gam}   and \eqref{leq2delt}
\bea 
2 \rightarrow 2-1 =1 
\eea 
The modified identity \eqref{NunstablePowerIds} then specialises to 
\bea 
1^3 =\biggl[  N \times (  2 - 1 ) \biggr ]_{ N =1  } 
\eea

\subsection{ For the case $ m,n =2$ }

\begin{table}[h]
\centering
\begin{tabular}{|c|c|c|c|}
\hline
 $ \gamma = ( k , \gamma_+  , \gamma_- ) $ &  mixed Young diagram $\Gamma ( \gamma , N )  $  & 
 $ \Dim V_{ \gamma }^{ U(N) } $ & $ d_{ m,n} ( \gamma ) $  \\
\hline
 $ \gamma_1 = ( 0 , [2] , [2] ) $   & $[\,2,0^{\,N-2},-2\,]$ 
& \rule{0pt}{0.7em} $\displaystyle \frac{N^{2}(N-1)(N+3)}{4}$ 
&  $1$  \\[0.4em]
$\gamma_2 = (0 , [2] , [1^2] )$ & $[\,2,0^{\,N-3},-1,-1\,]$ 
& $\displaystyle \frac{(N-2)(N-1)(N+1)(N+2)}{4}$ 
& $1$ \\[0.4em]
$ \gamma_3 = (0 , [1^2] , [2] )$ &  $[\,1,1,0^{\,N-3},-2\,]$ 
& $\displaystyle \frac{(N-2)(N-1)(N+1)(N+2)}{4}$ 
& $1$ \\[0.4em]
$ \gamma_4 = (0, [1^2] , [1^2] )$  & $[\,1,1,0^{\,N-4},-1,-1\,]$ 
& $\displaystyle \frac{N^{2}(N-3)(N+1)}{4}$ 
& $1$  \\[0.4em]
\hline
$ \gamma_5= (1, [1] , [1] )$  & $[\,1,0^{\,N-2},-1\,]$ 
& $\displaystyle N^{2}-1$ 
& $4 $\\[0.4em]
\hline 
$ \gamma_6 = (2, [] , [] )$ &  $[\,0^{\,N}\,]$  
& $\displaystyle 1$ 
& $2$ \\
\hline
\end{tabular}
\caption{ List of irreps of $U(N)$, with highest weight $\Gamma ( \gamma , N ) $  appearing in decomposition of $\Vmn$ for $(m,n)=(2,2)$. The  stable range multiplicities are the 
Brauer irrep dimensions $ d_{ 2,2} ( \gamma )$.   }
\label{Datameq2neq2} 
\end{table}

Using the data in table \ref{Datameq2neq2}, the identity \eqref{NstablePowerIds} now specialises to 
\bea\label{stabid22}  
&& N^4 = \frac{N^{2}(N-1)(N+3)}{4} + \frac{(N-2)(N-1)(N+1)(N+2)}{4} + \frac{(N-2)(N-1)(N+1)(N+2)}{4} 
\cr 
&& + \frac{N^{2}(N-3)(N+1)}{4} + 4 (  N^{2}-1 )  + 2 
\eea 
For $ N =3$, we have $ l = m+n -N =1$ and there is a single excluded irrep  $ \gamma_4$, no Brauer triples with modified dimensions,  while $ \Dim V_{ \gamma_4 }^{ U(N) } $ continued to $N=3$ is zero. It then follows from \eqref{stabid22} that the identity \eqref{NunstablePowerIds} holds.  
For $ N =2$, i.e. $l=2$,  the excluded triples are $ \gamma_2 , \gamma_3 , \gamma_4 $, the surviving reps are 
$ \gamma_1 , \gamma_5 , \gamma_6 $. From the equation \eqref{leq2gam} and \eqref{leq2delt}, $ \gamma_5$ has a modified dimension $ 4 \rightarrow 3$. The identity \eqref{NunstablePowerIds} is thus 
\bea 
2^4 && = \biggl [  \frac{N^{2}(N-1)(N+3)}{4}  +  3 (  N^{2}-1 )  + ( 2 ) \biggr ]_{ N=2} \cr 
&& = 5 + 9 + 2 = 16 
\eea
For $ N=1$, where $ l=3$, the surviving irrep is $  \gamma_7$ and as given by  \eqref{modleq3gams} and \eqref{modleq3delts} it has a modified dimension $ 2 \rightarrow 1 $. 
Thus the identity is $ 1^3 = 1 $, where the RHS is the contribution from $ \gamma_7$. 
Since the $m,n$ are smaller than the  $(m,n)$-stability bound  of $ 2l-3$ here, not all the modified dimensions calculated for $l=3$ are relevant here.

\subsection{ The case $ m=3,n=3$ } 

\begin{table}[h]
\centering
\scriptsize
\setlength{\tabcolsep}{6pt}
\renewcommand{\arraystretch}{1.17}
\begin{tabular}{|c|c|c|c|}
\hline
 $(k, \gamma_+,\gamma_-) \in {\rm BRT} ( 3,3) $&  mixed Young diagram $\Gamma ( \gamma , N )  $   & $\Dim (V_{ \gamma  }^{\tiny{ U(N} )} )$ &   $ d_{ m,n} ( \gamma ) $  \\
\hline
$(0,[3],[3])$ & $( 3,0^{N-2} , -3 )$  &  $\displaystyle \frac{N^{2}(N-1)(N+1)^{2}(N+5)}{36}$ &  $1$ \\
$(0,[3],[2,1])$ & $( 3,0^{N-3} , -1,-2 )$  & $\displaystyle \frac{N^{2}(N-2)(N-1)(N+2)(N+4)}{18}$ &  $2$ \\
$(0,[3],[1,1,1])$ & $ ( 3,0^{N-3} , -1,-2 )$  & $\displaystyle \frac{(N-3)(N-2)(N-1)(N+1)(N+2)(N+3)}{36}$ & $1$ \\
$(0,[2,1],[3])$ & $ ( 2,1,0^{N-3} , -3 )$ & $\displaystyle \frac{N^{2}(N-2)(N-1)(N+2)(N+4)}{18}$ & $2$ \\
$(0,[2,1],[2,1])$ & $ ( 2,1,0^{N-4} , -1,-2  )$ &  $\displaystyle \frac{(N-3)(N-1)^{2}(N+1)^{2}(N+3)}{9}$ & $4$ \\
$(0,[2,1],[1,1,1])$ & $ ( 2,1,0^{N-5} , -1,-1,-1  )$  & $\displaystyle \frac{N^{2}(N-4)(N-2)(N+1)(N+2)}{18}$ & $2$ \\
$(0,[1,1,1],[3])$ & $ ( 1,1,1 ,0^{N-4} , -3  )$  & $\displaystyle \frac{(N-3)(N-2)(N-1)(N+1)(N+2)(N+3)}{36}$ & $1$ \\
$(0,[1,1,1],[2,1])$ & $ ( 1,1,1 ,0^{N-5} , -1,-2  )$   & $\displaystyle \frac{N^{2}(N-4)(N-2)(N+1)(N+2)}{18}$ & $2$ \\
$(0,[1,1,1],[1,1,1])$ & $ ( 1,1,1 ,0^{N-6} , -1,-1,-1 )$  & $\displaystyle \frac{N^{2}(N-5)(N-1)^{2}(N+1)}{36}$ &  $1$ \\
\hline
$(1,[2],[2])$ &  $ ( 2,0^{N-2} , -2  )$   & $\displaystyle \frac{N^{2}(N-1)(N+3)}{4}$ &  $9$ \\
$(1,[2],[1,1])$ & $ ( 2,0^{N-3} , -1,-1  )$   &  $\displaystyle \frac{(N-2)(N-1)(N+1)(N+2)}{4}$ &  $9$ \\
$(1,[1,1],[2])$ & $ ( 1,1,0^{N-3} , -2  )$   &  $\displaystyle \frac{(N-2)(N-1)(N+1)(N+2)}{4}$ &  $9$ \\
$(1,[1,1],[1,1])$ & $ ( 1,1,0^{N-4} , -1,-1  )$  & $\displaystyle \frac{N^{2}(N-3)(N+1)}{4}$ &  $9$ \\
\hline
$(2,[1],[1])$ & $ ( 1 ,0^{N-2} , -1  )$  &  $\displaystyle (N-1)(N+1)=N^{2}-1$ & $18$ \\
\hline
$(2,[],[])$ & $ ( 0^{N}   )$  & $1$ &  $6$ \\
\hline
\end{tabular}
\caption{ List of irreps of $U(N)$, with highest weight $\Gamma ( \gamma , N ) $  appearing in decomposition of $\Vmn$ for $(m,n)=(3,3)$. The irreps are labelled by Brauer representation triples $ \gamma \in {\rm BRT } (3,3)$ and the stable range multiplicities are the 
Brauer irrep dimensions $ d_{ 3,3} ( \gamma )$.   }
\label{tab:Vmndecomp33} 
\end{table}

Using the data in  the table  \ref{tab:Vmndecomp33} we verify that 
\bea 
\sum_{ \gamma \in { \rm BRT} ( 3,3 )  } \Dim_N ( \gamma ) d_{ 3,3 } ( \gamma )  = N^6
\eea
For $ l=1$, i.e. $ N=5$ the identity continues to  work, because $ \gamma_9$ is the only  excluded Brauer triple and the dimension  $\Dim ( V^{ U(N)}_{ \gamma_9 }   $, evaluated for $ N=5$ gives $0$. 

For $l=2$, at $N=4$, $ \gamma_6 , \gamma_8 , \gamma_9 $ are excluded. 
The Brauer triples in \\ $ {\rm BRT } ( 3,3 ) \setminus \Excl ( 3,3, 4  ) $  listed according to the decomposition in  \eqref{DisjBRTmnN} are 
\bea 
&& \Unmod ( 3,3,4) = \{    \gamma_1 , \gamma_2 , \gamma_3 , \gamma_4 , \gamma_7 , \gamma_{ 10 } , 
\gamma_{ 11} , \gamma_{ 12} , \gamma_{ 14} , \gamma_{15} \} \cr
&& \Mod (3,3,4) = \{ \gamma_{13} \} 
\eea 
 Using \eqref{leq2gam} and  \eqref{leq2delt}, 
\bea\label{modif13}  
 \widehat{d}_{ 3,3 ,  4 } ( \gamma_{13} )  =   \widehat{d}_{ 3,3 } ( \gamma_{13} )  -1  = 8  
\eea
The lists of $ \gamma_i  
\in \widehat{\BRT} ( 3, 3 )$ along with the $ \Dim V_{ \gamma_i }^{ U(4)} $ are 
\bea\label{DimU334} 
 \{ \gamma_1  , \gamma_2 , \gamma_3 ,  \gamma_4  , 
\gamma_{7}  , \gamma_{10} , \gamma_{11} , \gamma_{12} ,  \gamma_{13} ,\gamma_{14}  , \gamma_{15}   \} \cr 
 \Dim V_{ \gamma_i }^{ U(4) }  = \{ 300 , 256 , 35 , 256 , 175 , 35 , 84 , 45 , 20 , 15 , 1 \} 
\eea 
The list of $ \widehat{d}_{ 3,3,4} ( \gamma )  $, including the one modification 
\eqref{modif13}  for $ \gamma_{13} $ is 
\bea\label{brtdim334}  
\widehat {d}_{ 3,3} (\gamma_i ) =  \{ 1, 2, 1,  2, 4,  1,  9, 9, 9, 8, 18  , 6 \}
\eea
Using the data in \eqref{DimU334} \eqref{brtdim334} , we verify 
\bea 
 \sum_{ \gamma \in \widehat{ \rm BRT } ( 3,3 )  }  \widehat {d}_{ 3,3} (\gamma ) ~~ \Dim ( V^{ U(4)}_{\gamma  } ) = 4^6 = 4096. 
\eea

For $ l =3$, i.e. $ N = 3$, the excluded Brauer triples are $ \{ \gamma_3 , \gamma_5 , \gamma_6 , \gamma_7 ,  \gamma_8 , \gamma_9 \}  $.  Using \eqref{modleq3gams} and \eqref{modleq3delts} the irreps with modified dimensions are $\{  \gamma_{11} , \gamma_{ 12} , \gamma_{ 14} \} $ and the modifications are $ \delta_{ 11} = \delta_{ 12} = 2, \delta_{ 14} = 1 $. Thus  Brauer triples $ \gamma_i \in \widehat{\BRT} ( m,n ) $, with the associated  dimensions of $ U(3)$ and ${\widehat B}_{ 3,3, 3} $ representations are 
\bea 
&& \{ \gamma_1 , \gamma_2, \gamma_3, \gamma_{ 10} , \gamma_{11} , \gamma_{ 12} , \gamma_{14} , \gamma_{ 15} \} \cr 
&& \Dim ( V_{ \gamma_i  }^{ U(N)}  )  = \{ 64, 35, 35,  27, 10, 10,  8, 1 \} \cr 
&& \widehat {d}_{ 3,3 } ( \gamma_i )  = \{ 1,2, 2, 9 , 7, 7 , 17 , 6 \}   
\eea 
Using the above data 
\bea 
\sum_{ \gamma \in \widehat{\BRT} ( 3,3,3)  } \widehat {d}_{ 3,3,3  } ( \gamma ) \Dim ( V_{ \gamma  }^{ U(3)}  )  = 729 = 3^6 
\eea 

For $ l=4$, i.e $N=2$,  the Brauer  triples in $ \widehat{\BRT} ( 3,3 , 2)$ and associated $U(2)$ dimensions are 
\bea\label{DimU2}  
&& \{ \gamma_1, \gamma_{10}  , \gamma_{14}  , \gamma_{15  } \} \cr 
&& \Dim ( V_{ \gamma_i  }^{ U(2)}  )   = \{  7, 5 , 3 , 1   \} 
\eea 
From the equations  \eqref{leq4modgams} and \eqref{leq4moddelts2} the  triples with modified Brauer dimensions are $ \{  \gamma_{10}  , \gamma_{14}  , \gamma_{15  } \} \} $ and the modifications are  $  ( 4 , 9 , 1 ) $. 
The dimensions of irreps of $ \widehat{B}_{ 3,3, 2 } $ are thus 
\bea\label{DimBhat}  
\widehat {d}_{ 3,3 } ( \gamma_i ) = \{ 1,9-4 = 5 ,18-9 =9 , 6-1 = 5 \} 
\eea
Using the data in \eqref{DimU2} and \eqref{DimBhat} we verify 
\bea 
 \sum_{ \gamma \in \widehat{ \rm BRT } ( 3,3 , 2 )  }  \widehat {d}_{ 3,3,2 } (\gamma ) ~~ \Dim ( V^{ U(2)}_{\gamma  } ) = 7 + 25 + 27 +5 = 64 = 2^6 
\eea 

\section{Auxiliary figures for computing dimensions modifications}

In this appendix, we collect the RBDs in a more convenient form for analysis, for the case of $B_N(m,n)$ with $N=m+n-4$, discussed in Subsection~\ref{subSub}. The collected RBD diagrams are presented in such a way that they display only those connections in the Bratteli diagram which are necessary for computing the dimension modifications $\delta_{m,n,N=m+n-4}(\gamma'_i)$ from \eqref{leq4moddelts2}. We exclude the dimensions $\gamma_1,\gamma_2$, since their modifications can be easily inferred from Figure~\ref{Gen_mnk4}.
\begin{figure}[h!]
	\centering
	\includegraphics[scale=0.4]{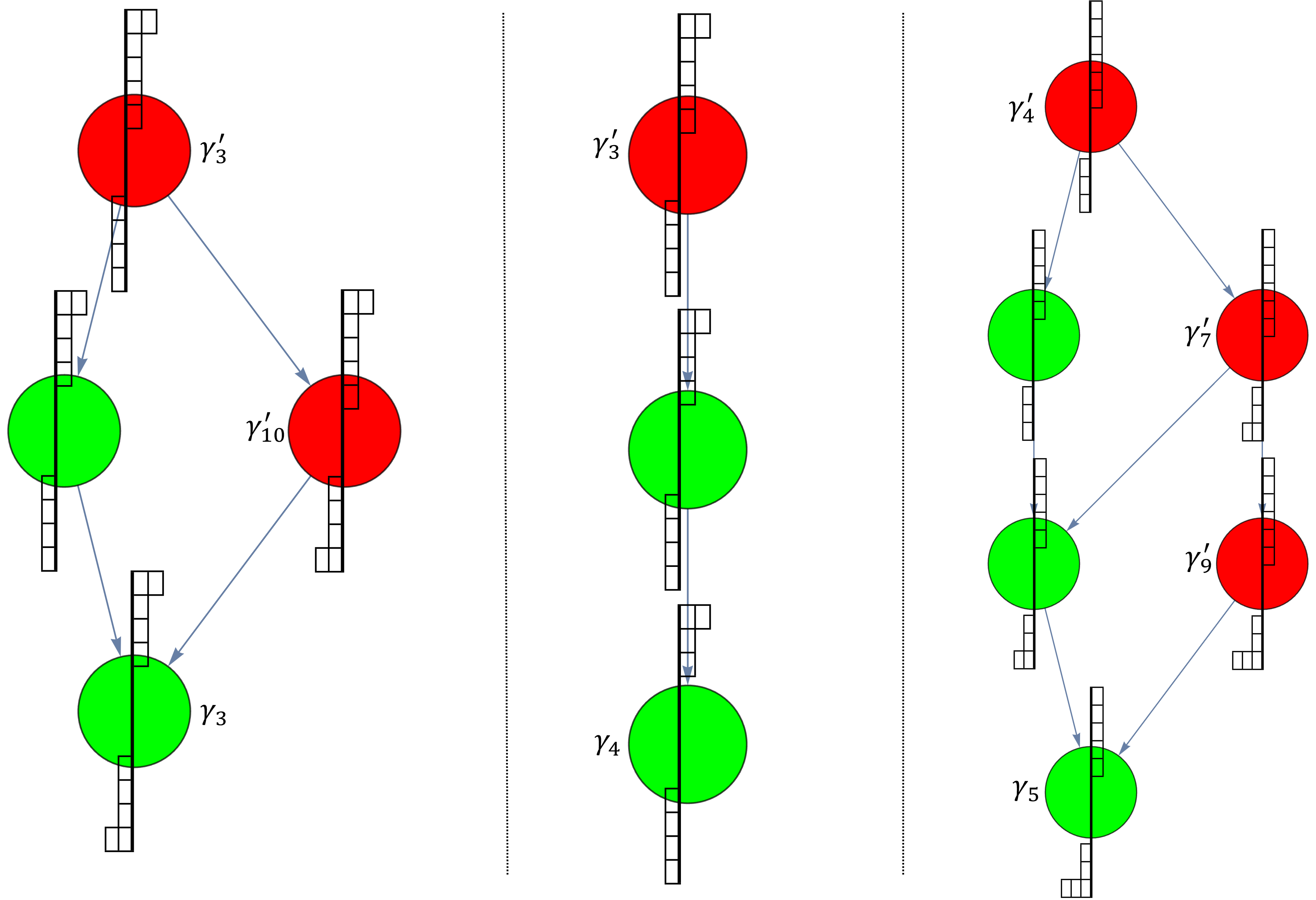}
	\caption{Parts of the RBD from Figure~\ref{Gen_mnk4} for computing dimension modifications $\gamma_{m,n,N=m+n-4}(\gamma_3),\gamma_{m,n,N=m+n-4}(\gamma_4),\gamma_{m,n,N=m+n-4}(\gamma_5)$ counting from the left.} 
	\label{AppFig1}
\end{figure} 

\begin{figure}[h!] 
	\centering
	\includegraphics[scale=0.4]{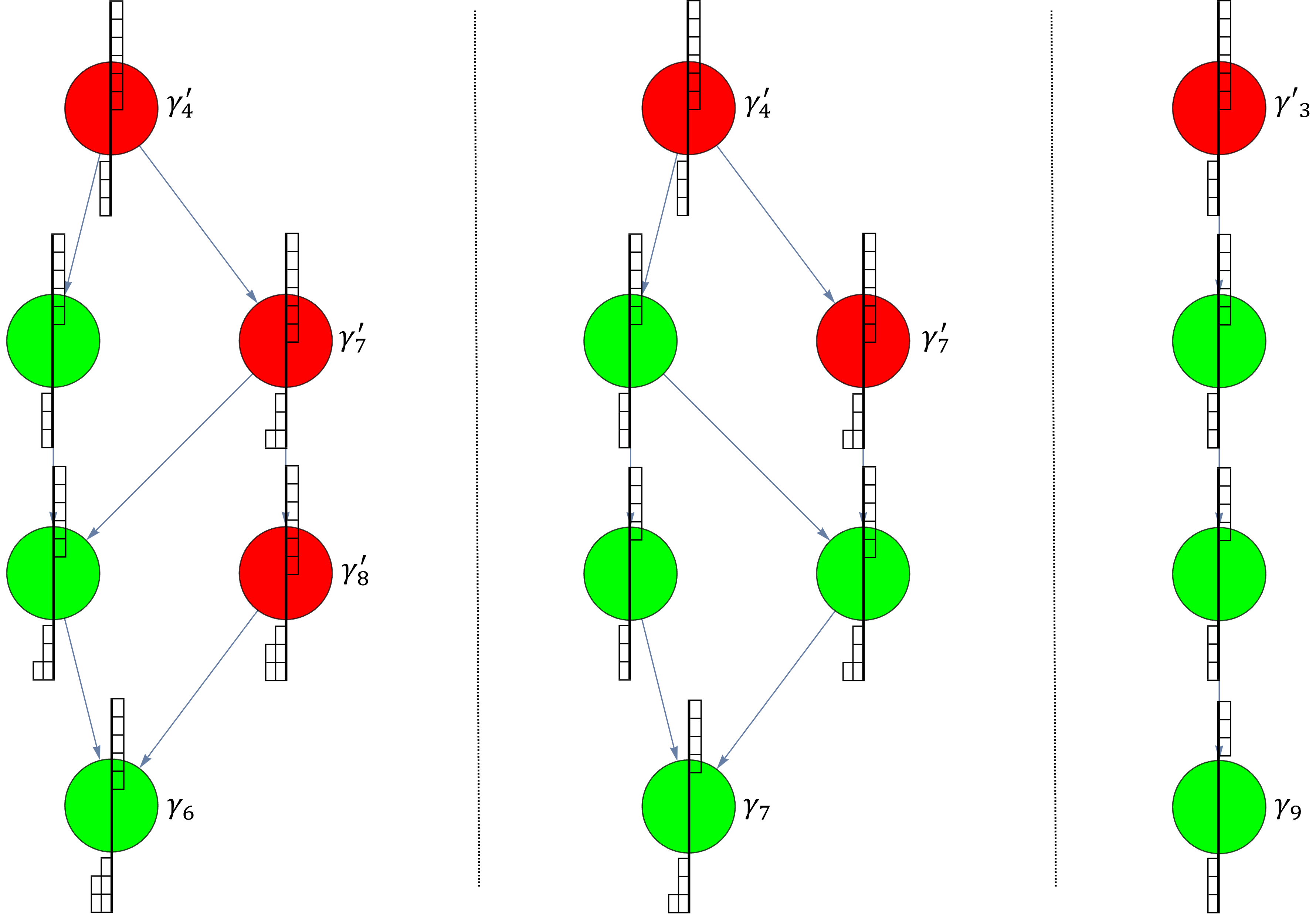}
	\caption{Parts of the RBD from Figure~\ref{Gen_mnk4} for computing dimension modifications $\gamma_{m,n,N=m+n-4}(\gamma_6),\gamma_{m,n,N=m+n-4}(\gamma_7),\gamma_{m,n,N=m+n-4}(\gamma_9)$ counting from the left.} 
	\label{AppFig2}
\end{figure}

\begin{figure}[h!] 
	\centering
	\includegraphics[scale=0.4]{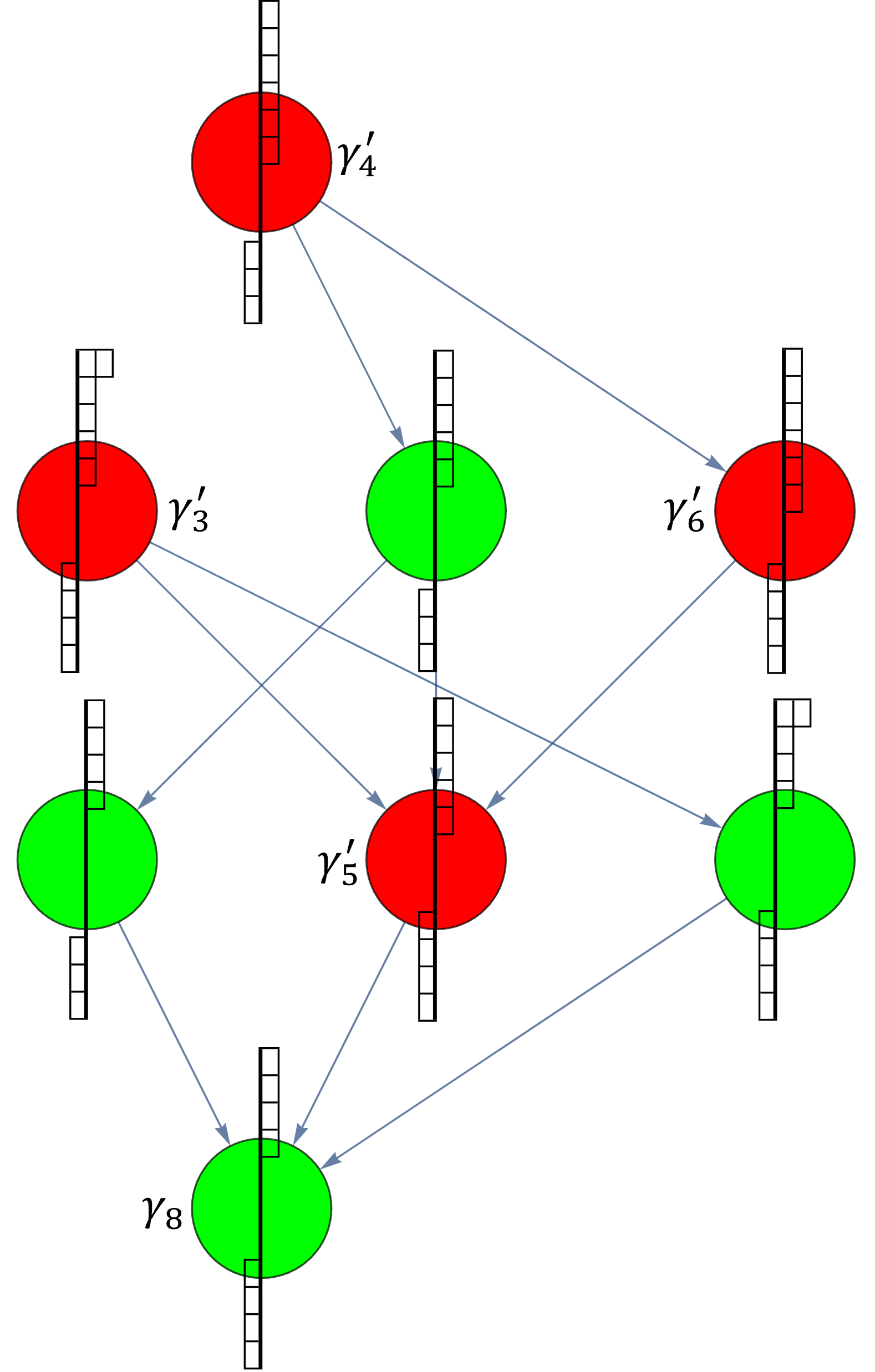}
	\caption{Part of the RBD from Figure~\ref{Gen_mnk4} for computing dimension modification $\gamma_{m,n,N=m+n-4}(\gamma_8)$.} 
	\label{AppFig3}
\end{figure}

\end{appendix}

\newpage

\bibliographystyle{unsrt}
\bibliography{bibliography}

@misc{RamgoolamStudzinski2025WBAcode,
  author = {Sanjaye Ramgoolam and Micha{\l} Studzi{\'n}ski},
  title  = {Mathematica code for restricted Bratteli diagrams and modified dimensions},
  note   = {Ancillary Mathematica notebook to arXiv:2509.04234},
  year   = {2025},
  url    = {https://arxiv.org/abs/2509.04234}
}

@article{ProjDetect,
    author = "Geloun, Joseph Ben and Ramgoolam, Sanjaye",
    title = "{The quantum detection of projectors in finite-dimensional algebras and holography}",
    eprint = "2303.12154",
    archivePrefix = "arXiv",
    primaryClass = "quant-ph",
    reportNumber = "QMUL-PH-23-04",
    doi = "10.1007/JHEP05(2023)191",
    journal = "JHEP",
    volume = "05",
    pages = "191",
    year = "2023"
}

@article{Han2016,
    author = "Hanada, Masanori and Maltz, Jonathan",
    title = "{A proposal of the gauge theory description of the small Schwarzschild black hole in AdS$_5\times$S$^5$}",
    eprint = "1608.03276",
    archivePrefix = "arXiv",
    primaryClass = "hep-th",
    reportNumber = "SU-ITP-16-14, YITP-16-95",
    doi = "10.1007/JHEP02(2017)012",
    journal = "JHEP",
    volume = "02",
    pages = "012",
    year = "2017"
}

@article{Ber2018,
    author = "Berenstein, David",
    title = "{Negative specific heat from non-planar interactions and small black holes in AdS/CFT}",
    eprint = "1810.07267",
    archivePrefix = "arXiv",
    primaryClass = "hep-th",
    doi = "10.1007/JHEP10(2019)001",
    journal = "JHEP",
    volume = "10",
    pages = "001",
    year = "2019"
}

@article{DOCSR,
    author = "O'Connor, Denjoe and Ramgoolam, Sanjaye",
    title = "{Permutation invariant matrix quantum thermodynamics and negative specific heat capacities in large N systems}",
    eprint = "2405.13150",
    archivePrefix = "arXiv",
    primaryClass = "hep-th",
    reportNumber = "DIAS-2024-05, QMUL-PH-24-07",
    doi = "10.1007/JHEP12(2024)161",
    journal = "JHEP",
    volume = "12",
    pages = "161",
    year = "2024"
}

@article{DasJev,
    author = "Das, Sumit R. and Jevicki, Antal",
    editor = "Brezin, E. and Wadia, S. R.",
    title = "{String Field Theory and Physical Interpretation of $D=1$ Strings}",
    reportNumber = "BROWN-HET-750, TIFR-TH-90-26",
    doi = "10.1142/S0217732390001888",
    journal = "Mod. Phys. Lett. A",
    volume = "5",
    pages = "1639--1650",
    year = "1990"
}

@article{MinPoly,
    author = "Minahan, Joseph A. and Polychronakos, Alexios P.",
    title = "{Equivalence of two-dimensional QCD and the C = 1 matrix model}",
    eprint = "hep-th/9303153",
    archivePrefix = "arXiv",
    reportNumber = "CERN-TH-6843-93, UVA-HET-93-02",
    doi = "10.1016/0370-2693(93)90504-B",
    journal = "Phys. Lett. B",
    volume = "312",
    pages = "155--165",
    year = "1993"
}

@article{GJRSV2013,
   title={Logarithmic conformal field theory: a lattice approach},
   volume={46},
   ISSN={1751-8121},
   url={http://dx.doi.org/10.1088/1751-8113/46/49/494012},
   DOI={10.1088/1751-8113/46/49/494012},
   number={49},
   journal={Journal of Physics A: Mathematical and Theoretical},
   publisher={IOP Publishing},
   author={Gainutdinov, A M and Jacobsen, J L and Read, N and Saleur, H and Vasseur, R},
   year={2013},
   month=nov, pages={494012} }

@article{StollWerth2016,
  author    = {Florian Stoll and Markus Werth},
  title     = {A cell filtration of mixed tensor space},
  journal   = {Mathematische Zeitschrift},
  year      = {2016},
  volume    = {282},
  number    = {3-4},
  pages     = {769--798},
  doi       = {10.1007/s00209-015-1564-y},
  url       = {https://doi.org/10.1007/s00209-015-1564-y},
  publisher = {Springer},
  issn      = {1432-1823}
}

@article{MalStrom,
    author = "Maldacena, Juan Martin and Strominger, Andrew",
    title = "{AdS(3) black holes and a stringy exclusion principle}",
    eprint = "hep-th/9804085",
    archivePrefix = "arXiv",
    reportNumber = "HUTP-98-A016",
    doi = "10.1088/1126-6708/1998/12/005",
    journal = "JHEP",
    volume = "12",
    pages = "005",
    year = "1998"
}

@article{Horava,
    author = "Horava, Petr",
    title = "{Topological rigid string theory and two-dimensional QCD}",
    eprint = "hep-th/9507060",
    archivePrefix = "arXiv",
    reportNumber = "PUPT-1547",
    doi = "10.1016/0550-3213(96)00036-3",
    journal = "Nucl. Phys. B",
    volume = "463",
    pages = "238--286",
    year = "1996"
}

@article{JanYM2,
    author = "Benizri, Lior and Troost, Jan",
    title = "{The string dual to two-dimensional Yang-Mills theory}",
    eprint = "2502.02662",
    archivePrefix = "arXiv",
    primaryClass = "hep-th",
    doi = "10.1007/JHEP08(2025)017",
    journal = "JHEP",
    volume = "08",
    pages = "017",
    year = "2025"
}

@article{ShotaYM2,
    author = {Komatsu, Shota and Maity, Pronobesh},
    title = {String Duals of Two-Dimensional Yang-Mills and Symmetric Product Orbifolds},
    journal = {Phys. Rev. Lett.},
    volume = {135},
    pages = {231603},
    year = {2025},
    eprint = {2506.21663},
    archivePrefix = {arXiv},
    primaryClass = {hep-th},
    doi = {10.1103/zkm5-smzx}
}

@article{Ofer,
    author = "Aharony, Ofer and Kundu, Suman and Sheaffer, Tal",
    title = "{A string theory for two dimensional Yang-Mills theory. Part I}",
    eprint = "2312.12266",
    archivePrefix = "arXiv",
    primaryClass = "hep-th",
    doi = "10.1007/JHEP07(2024)063",
    journal = "JHEP",
    volume = "07",
    pages = "063",
    year = "2024"
}

@article{BCD0801,
    author = "Bhattacharyya, Rajsekhar and Collins, Storm and de Mello Koch, Robert",
    title = "{Exact Multi-Matrix Correlators}",
    eprint = "0801.2061",
    archivePrefix = "arXiv",
    primaryClass = "hep-th",
    reportNumber = "WITS-CTP-036",
    doi = "10.1088/1126-6708/2008/03/044",
    journal = "JHEP",
    volume = "03",
    pages = "044",
    year = "2008"
}

@article{GigraOsc,
    author = "de Mello Koch, Robert and Dessein, Matthias and Giataganas, Dimitrios and Mathwin, Christopher",
    title = "{Giant Graviton Oscillators}",
    eprint = "1108.2761",
    archivePrefix = "arXiv",
    primaryClass = "hep-th",
    reportNumber = "WITS-CTP-078",
    doi = "10.1007/JHEP10(2011)009",
    journal = "JHEP",
    volume = "10",
    pages = "009",
    year = "2011"
}

@article{CMR,
    author = "Cordes, Stefan and Moore, Gregory W. and Ramgoolam, Sanjaye",
    title = "{Lectures on 2-d Yang-Mills theory, equivariant cohomology and topological field theories}",
    eprint = "hep-th/9411210",
    archivePrefix = "arXiv",
    reportNumber = "YCTP-P11-94",
    doi = "10.1016/0920-5632(95)00434-B",
    journal = "Nucl. Phys. B Proc. Suppl.",
    volume = "41",
    pages = "184--244",
    year = "1995"
}

@article{CMR1,
    author = "Cordes, Stefan and Moore, Gregory W. and Ramgoolam, Sanjaye",
    title = "{Large N 2-D Yang-Mills theory and topological string theory}",
    eprint = "hep-th/9402107",
    archivePrefix = "arXiv",
    reportNumber = "YCTP-P23-93, RU-94-20",
    doi = "10.1007/s002200050102",
    journal = "Commun. Math. Phys.",
    volume = "185",
    pages = "543--619",
    year = "1997"
}

@article{KRHol,
    author = "Kimura, Yusuke and Ramgoolam, Sanjaye",
    title = "{Holomorphic maps and the complete 1/N expansion of 2D SU(N) Yang-Mills}",
    eprint = "0802.3662",
    archivePrefix = "arXiv",
    primaryClass = "hep-th",
    reportNumber = "QMUL-PH-08-03",
    doi = "10.1088/1126-6708/2008/06/015",
    journal = "JHEP",
    volume = "06",
    pages = "015",
    year = "2008"
}

@article{GrTa2,
    author = "Gross, David J. and Taylor, Washington",
    title = "{Twists and Wilson loops in the string theory of two-dimensional QCD}",
    eprint = "hep-th/9303046",
    archivePrefix = "arXiv",
    reportNumber = "CERN-TH-6827-93, PUPT-1382, LBL-33767, UCB-PTH-93-09",
    doi = "10.1016/0550-3213(93)90042-N",
    journal = "Nucl. Phys. B",
    volume = "403",
    pages = "395--452",
    year = "1993"
}

@article{GrTa1,
    author = "Gross, David J. and Taylor, Washington",
    title = "{Two-dimensional QCD is a string theory}",
    eprint = "hep-th/9301068",
    archivePrefix = "arXiv",
    reportNumber = "LBL-33458, PUPT-1376, UCB-PTH-93-02",
    doi = "10.1016/0550-3213(93)90403-C",
    journal = "Nucl. Phys. B",
    volume = "400",
    pages = "181--208",
    year = "1993"
}

@article{EigVal,
    author = "Padellaro, Adrian and Ramgoolam, Sanjaye and Suzuki, Ryo",
    title = "{Eigenvalue systems for integer orthogonal bases of multi-matrix invariants at finite N}",
    eprint = "2410.13631",
    archivePrefix = "arXiv",
    primaryClass = "hep-th",
    doi = "10.1007/JHEP02(2025)111",
    journal = "JHEP",
    volume = "02",
    pages = "111",
    year = "2025"
}

@article{QuivCalc,
    author = "Pasukonis, Jurgis and Ramgoolam, Sanjaye",
    title = "{Quivers as Calculators: Counting, Correlators and Riemann Surfaces}",
    eprint = "1301.1980",
    archivePrefix = "arXiv",
    primaryClass = "hep-th",
    reportNumber = "QMUL-PH-12-17",
    doi = "10.1007/JHEP04(2013)094",
    journal = "JHEP",
    volume = "04",
    pages = "094",
    year = "2013"
}

@article{BHR2,
    author = "Brown, Thomas William and Heslop, P. J. and Ramgoolam, S.",
    title = "{Diagonal free field matrix correlators, global symmetries and giant gravitons}",
    eprint = "0806.1911",
    archivePrefix = "arXiv",
    primaryClass = "hep-th",
    reportNumber = "QMUL-PH-08-12",
    doi = "10.1088/1126-6708/2009/04/089",
    journal = "JHEP",
    volume = "04",
    pages = "089",
    year = "2009"
}

@article{BHR1,
    author = "Brown, Thomas William and Heslop, P. J. and Ramgoolam, S.",
    title = "{Diagonal multi-matrix correlators and BPS operators in N=4 SYM}",
    eprint = "0711.0176",
    archivePrefix = "arXiv",
    primaryClass = "hep-th",
    reportNumber = "QMUL-PH-07-23",
    doi = "10.1088/1126-6708/2008/02/030",
    journal = "JHEP",
    volume = "02",
    pages = "030",
    year = "2008"
}

@article{DoubCos,
    author = "de Mello Koch, Robert and Ramgoolam, Sanjaye",
    title = "{A double coset ansatz for integrability in AdS/CFT}",
    eprint = "1204.2153",
    archivePrefix = "arXiv",
    primaryClass = "hep-th",
    reportNumber = "QMUL-PH-12-08, WITS-CTP-092",
    doi = "10.1007/JHEP06(2012)083",
    journal = "JHEP",
    volume = "06",
    pages = "083",
    year = "2012"
}

@article{HHI0008,
    author = "Hashimoto, Akikazu and Hirano, Shinji and Itzhaki, N.",
    title = "{Large branes in AdS and their field theory dual}",
    eprint = "hep-th/0008016",
    archivePrefix = "arXiv",
    reportNumber = "NSF-ITP-00-063",
    doi = "10.1088/1126-6708/2000/08/051",
    journal = "JHEP",
    volume = "08",
    pages = "051",
    year = "2000"
}

@article{MST,
    author = "McGreevy, John and Susskind, Leonard and Toumbas, Nicolaos",
    title = "{Invasion of the giant gravitons from Anti-de Sitter space}",
    eprint = "hep-th/0003075",
    archivePrefix = "arXiv",
    reportNumber = "SU-ITP-00-09",
    doi = "10.1088/1126-6708/2000/06/008",
    journal = "JHEP",
    volume = "06",
    pages = "008",
    year = "2000"
}

@article{GMT0008,
    author = "Grisaru, Marcus T. and Myers, Robert C. and Tafjord, Oyvind",
    title = "{SUSY and goliath}",
    eprint = "hep-th/0008015",
    archivePrefix = "arXiv",
    reportNumber = "MCGILL-00-21, BRX-TH-472",
    doi = "10.1088/1126-6708/2000/08/040",
    journal = "JHEP",
    volume = "08",
    pages = "040",
    year = "2000"
}

@article{BudGai,
    author = "Budzik, Kasia and Gaiotto, Davide",
    title = "{Giant gravitons in twisted holography}",
    eprint = "2106.14859",
    archivePrefix = "arXiv",
    primaryClass = "hep-th",
    doi = "10.1007/JHEP10(2023)131",
    journal = "JHEP",
    volume = "10",
    pages = "131",
    year = "2023"
}

@article{YJKW2103,
    author = "Yang, Peihe and Jiang, Yunfeng and Komatsu, Shota and Wu, Jun-Bao",
    title = "{D-branes and orbit average}",
    eprint = "2103.16580",
    archivePrefix = "arXiv",
    primaryClass = "hep-th",
    reportNumber = "CERN-TH-2021-043, USTC-ICTS/PCFT-21-15",
    doi = "10.21468/SciPostPhys.12.2.055",
    journal = "SciPost Phys.",
    volume = "12",
    number = "2",
    pages = "055",
    year = "2022"
}

@article{BDS0805,
    author = "Bhattacharyya, Rajsekhar and de Mello Koch, Robert and Stephanou, Michael",
    title = "{Exact Multi-Restricted Schur Polynomial Correlators}",
    eprint = "0805.3025",
    archivePrefix = "arXiv",
    primaryClass = "hep-th",
    reportNumber = "WITS-CTP-037",
    doi = "10.1088/1126-6708/2008/06/101",
    journal = "JHEP",
    volume = "06",
    pages = "101",
    year = "2008"
}

@misc{OEISA000714,
  author       = {{OEIS Foundation Inc.}},
  title        = {{The On-Line Encyclopedia of Integer Sequences, Sequence A000714}},
  year         = {2025},
  howpublished = {\url{https://oeis.org/A000714}},
  note         = {Accessed: 2025-07-28}
}

@article{MatRam2016,
    author = "Mattioli, Paolo and Ramgoolam, Sanjaye",
    title = "{Permutation Centralizer Algebras and Multi-Matrix Invariants}",
    eprint = "1601.06086",
    archivePrefix = "arXiv",
    primaryClass = "hep-th",
    reportNumber = "QMUL-PH-15-25",
    doi = "10.1103/PhysRevD.93.065040",
    journal = "Phys. Rev. D",
    volume = "93",
    number = "6",
    pages = "065040",
    year = "2016"
}

@article{KRT09,
    author = "Kimura, Yusuke and Ramgoolam, Sanjaye and Turton, David",
    title = "{Free particles from Brauer algebras in complex matrix models}",
    eprint = "0911.4408",
    archivePrefix = "arXiv",
    primaryClass = "hep-th",
    reportNumber = "QMUL-PH-09-20",
    doi = "10.1007/JHEP05(2010)052",
    journal = "JHEP",
    volume = "05",
    pages = "052",
    year = "2010"
}

@article{EHS,
    author = "Kimura, Yusuke and Ramgoolam, Sanjaye",
    title = "{Enhanced symmetries of gauge theory and resolving the spectrum of local operators}",
    eprint = "0807.3696",
    archivePrefix = "arXiv",
    primaryClass = "hep-th",
    reportNumber = "QMUL-PH-08-13",
    doi = "10.1103/PhysRevD.78.126003",
    journal = "Phys. Rev. D",
    volume = "78",
    pages = "126003",
    year = "2008"
}

@article{KimuraQuarter,
    author = "Kimura, Yusuke",
    title = "{Quarter BPS classified by Brauer algebra}",
    eprint = "1002.2424",
    archivePrefix = "arXiv",
    primaryClass = "hep-th",
    reportNumber = "FPAUO-10-01",
    doi = "10.1007/JHEP05(2010)103",
    journal = "JHEP",
    volume = "05",
    pages = "103",
    year = "2010"
}

@article{KimuraLin,
    author = "Kimura, Yusuke and Lin, Hai",
    title = "{Young diagrams, Brauer algebras, and bubbling geometries}",
    eprint = "1109.2585",
    archivePrefix = "arXiv",
    primaryClass = "hep-th",
    reportNumber = "FPAUO-11-12",
    doi = "10.1007/JHEP01(2012)121",
    journal = "JHEP",
    volume = "01",
    pages = "121",
    year = "2012"
}

@article{Mandal,
    author = "Mandal, Gautam and Mohan, Ajay",
    title = "{Exact lattice bosonization of finite N matrix quantum mechanics and c = 1}",
    eprint = "2406.07629",
    archivePrefix = "arXiv",
    primaryClass = "hep-th",
    reportNumber = "TIFR/TH/24-6",
    doi = "10.1007/JHEP03(2025)210",
    journal = "JHEP",
    volume = "03",
    pages = "210",
    year = "2025"
}

@inproceedings{Douglas1993,
    author = "Douglas, Michael R.",
    title = "{Conformal field theory techniques in large N Yang-Mills theory}",
    booktitle = "{NATO Advanced Research Workshop on New Developments in String Theory, Conformal Models and Topological Field Theory}",
    eprint = "hep-th/9311130",
    archivePrefix = "arXiv",
    reportNumber = "RU-93-57",
    month = "5",
    year = "1993"
}

@article{KR,
    author = "Kimura, Yusuke and Ramgoolam, Sanjaye",
    title = "{Branes, anti-branes and brauer algebras in gauge-gravity duality}",
    eprint = "0709.2158",
    archivePrefix = "arXiv",
    primaryClass = "hep-th",
    reportNumber = "QMUL-PH-07-19",
    doi = "10.1088/1126-6708/2007/11/078",
    journal = "JHEP",
    volume = "11",
    pages = "078",
    year = "2007"
}

@article{malda,
    author = "Maldacena, Juan Martin",
    title = "{The Large $N$ limit of superconformal field theories and supergravity}",
    eprint = "hep-th/9711200",
    archivePrefix = "arXiv",
    reportNumber = "HUTP-97-A097, HUTP-98-A097",
    doi = "10.4310/ATMP.1998.v2.n2.a1",
    journal = "Adv. Theor. Math. Phys.",
    volume = "2",
    pages = "231--252",
    year = "1998"
}

@article{gkp,
    author = "Gubser, S. S. and Klebanov, Igor R. and Polyakov, Alexander M.",
    title = "{Gauge theory correlators from noncritical string theory}",
    eprint = "hep-th/9802109",
    archivePrefix = "arXiv",
    reportNumber = "PUPT-1767",
    doi = "10.1016/S0370-2693(98)00377-3",
    journal = "Phys. Lett. B",
    volume = "428",
    pages = "105--114",
    year = "1998"
}

@article{witten,
    author = "Witten, Edward",
    title = "{Anti de Sitter space and holography}",
    eprint = "hep-th/9802150",
    archivePrefix = "arXiv",
    reportNumber = "IASSNS-HEP-98-15",
    doi = "10.4310/ATMP.1998.v2.n2.a2",
    journal = "Adv. Theor. Math. Phys.",
    volume = "2",
    pages = "253--291",
    year = "1998"
}

@article{CJR,
    author = "Corley, Steve and Jevicki, Antal and Ramgoolam, Sanjaye",
    title = "{Exact correlators of giant gravitons from dual N=4 SYM theory}",
    eprint = "hep-th/0111222",
    archivePrefix = "arXiv",
    reportNumber = "BROWN-HET-1292",
    doi = "10.4310/ATMP.2001.v5.n4.a6",
    journal = "Adv. Theor. Math. Phys.",
    volume = "5",
    pages = "809--839",
    year = "2002"
}

@ARTICLE{StudzinskiIEEE22,
  author={Studziński, Michał and Mozrzymas, Marek and Kopszak, Piotr and Horodecki, Michał},
  journal={IEEE Transactions on Information Theory}, 
  title={Efficient Multi Port-Based Teleportation Schemes}, 
  year={2022},
  volume={68},
  number={12},
  pages={7892-7912},
  doi={10.1109/TIT.2022.3187852},
url = {https://doi.org/10.1109/TIT.2022.3187852}}

@misc{nguyen2023mixedschurtransformefficient,
      title={The mixed Schur transform: efficient quantum circuit and applications}, 
      author={Quynh T. Nguyen},
      year={2023},
      eprint={2310.01613},
      archivePrefix={arXiv},
      primaryClass={quant-ph},
      url={https://arxiv.org/abs/2310.01613}, 
}

@article{Halverson1996CharactersOT,
  title={Characters of the centralizer algebras of mixed tensor representations of $GL(r, \mathbb{C})$ and the quantum group $q(gl(r, \mathbb{C}))$},
  author={Tom Halverson},
  journal={Pacific Journal of Mathematics},
  year={1996},
  volume={174},
  pages={359-410},
  url={https://api.semanticscholar.org/CorpusID:53374509}
}

@article{Yongzhang2021permutation,
author = {Zhang, Yong and Kauffman, Louis and Werner, Reinhard},
year = {2007},
title = {Permutation and its partial transpose},
volume = {05},
journal = {International Journal of Quantum Information},
doi = {10.1142/S021974990700302X},
url = {https://doi.org/10.1142/S021974990700302X}
}

@article{MozJPA,
   title={A simplified formalism of the algebra of partially transposed permutation operators with applications},
   volume={51},
   ISSN={1751-8121},
   url={http://dx.doi.org/10.1088/1751-8121/aaad15},
   DOI={10.1088/1751-8121/aaad15},
   number={12},
   journal={Journal of Physics A: Mathematical and Theoretical},
   publisher={IOP Publishing},
   author={Mozrzymas, Marek and Studziński, Michał and Horodecki, Michał},
   year={2018},
 pages={125202} }

@article{studzinski2017port,
   title={Port-based teleportation in arbitrary dimension},
   volume={7},
   ISSN={2045-2322},
   url={http://dx.doi.org/10.1038/s41598-017-10051-4},
   DOI={10.1038/s41598-017-10051-4},
   number={1},
   journal={Scientific Reports},
   publisher={Springer Science and Business Media LLC},
   author={Studziński, Michał and Strelchuk, Sergii and Mozrzymas, Marek and Horodecki, Michał},
   year={2017}
}

@article{Ebler_2023,
   title={Optimal Universal Quantum Circuits for Unitary Complex Conjugation},
   volume={69},
   ISSN={1557-9654},
   url={http://dx.doi.org/10.1109/TIT.2023.3263771},
   doi={10.1109/tit.2023.3263771},
   number={8},
   journal={IEEE Transactions on Information Theory},
   publisher={Institute of Electrical and Electronics Engineers (IEEE)},
   author={Ebler, Daniel and Horodecki, Michał and Marciniak, Marcin and Młynik, Tomasz and Quintino, Marco Túlio and Studziński, Michał},
   year={2023},
pages={5069–5082} }

@article{ishizaka_asymptotic_2008,
	Author = {Ishizaka, Satoshi and Hiroshima, Tohya},
	Doi = {10.1103/PhysRevLett.101.240501},
	Journal = {Physical Review Letters},
	Number = {24},
	Pages = {240501},
	Title = {Asymptotic {Teleportation} {Scheme} as a {Universal} {Programmable} {Quantum} {Processor}},
	Url = {http://link.aps.org/doi/10.1103/PhysRevLett.101.240501},
	Urldate = {2011-10-19},
	Volume = {101},
	Year = {2008}}

@article{Moz1,
  author       = {Mozrzymas, Marek and Horodecki, Micha{\l} and Studzi{\'n}ski, Micha{\l}},
  title        = {Structure and properties of the algebra of partially transposed permutation operators},
  journal      = {Journal of Mathematical Physics},
  volume       = {55},
  number       = {3},
  pages        = {032202},
  year         = {2014},
  doi          = {10.1063/1.4869027},
  url          = {http://scitation.aip.org/content/aip/journal/jmp/55/3/10.1063/1.4869027},
  issn         = {0022-2488, 1089-7658}
}

@misc{marcinska24,
      title={A memory and gate efficient algorithm for unitary mixed Schur sampling}, 
      author={Enrique Cervero-Martín and Laura Mančinska and Elias Theil},
      year={2024},
      eprint={2410.15793},
      archivePrefix={arXiv},
      primaryClass={quant-ph},
      url={https://arxiv.org/abs/2410.15793}, 
}

@article{Quintino2022deterministic,
  doi = {10.22331/q-2022-03-31-679},
  url = {https://doi.org/10.22331/q-2022-03-31-679},
  title = {Deterministic transformations between unitary operations: {E}xponential advantage with adaptive quantum circuits and the power of indefinite causality},
  author = {Quintino, Marco T{\'{u}}lio and Ebler, Daniel},
  journal = {{Quantum}},
  issn = {2521-327X},
  publisher = {{Verein zur F{\"{o}}rderung des Open Access Publizierens in den Quantenwissenschaften}},
  volume = {6},
  pages = {679},
  year = {2022}
}

@Article{Grinko2024Lin,
author={Grinko, Dmitry
and Ozols, Maris},
title={Linear Programming with Unitary-Equivariant Constraints},
journal={Communications in Mathematical Physics},
year={2024},
month={Nov},
day={06},
volume={405},
number={12},
pages={278},
issn={1432-0916},
doi={10.1007/s00220-024-05108-1},
url={https://doi.org/10.1007/s00220-024-05108-1}
}

@Article{Nechita2023,
author={Nechita, Ion
and Pellegrini, Cl{\'e}ment
and Rochette, Denis},
title={The asymmetric quantum cloning region},
journal={Letters in Mathematical Physics},
year={2023},
volume={113},
number={3},
pages={74},
issn={1573-0530},
doi={10.1007/s11005-023-01694-8},
url={https://doi.org/10.1007/s11005-023-01694-8}
}

@article{mozrzymas2021optimal,
   title={Optimal Multi-port-based Teleportation Schemes},
   volume={5},
   ISSN={2521-327X},
   url={http://dx.doi.org/10.22331/q-2021-06-17-477},
   doi={10.22331/q-2021-06-17-477},
   journal={Quantum},
   publisher={Verein zur Forderung des Open Access Publizierens in den Quantenwissenschaften},
   author={Mozrzymas, Marek and Studziński, Michał and Kopszak, Piotr},
   year={2021},
   pages={477} 
}

@misc{grinko2024efficientquantumcircuitsportbased,
      title={Efficient quantum circuits for port-based teleportation}, 
      author={Dmitry Grinko and Adam Burchardt and Maris Ozols},
      year={2024},
      eprint={2312.03188},
      archivePrefix={arXiv},
      primaryClass={quant-ph},
      url={https://arxiv.org/abs/2312.03188}, 
}

@misc{fei2023efficientquantumalgorithmportbased,
      title={Efficient Quantum Algorithm for Port-based Teleportation}, 
      author={Jiani Fei and Sydney Timmerman and Patrick Hayden},
      year={2023},
      eprint={2310.01637},
      archivePrefix={arXiv},
      primaryClass={quant-ph},
      url={https://arxiv.org/abs/2310.01637}, 
}

@article{PRXQuantum.5.030354,
  title = {Efficient Algorithms for All Port-Based Teleportation Protocols},
  author = {Wills, Adam and Hsieh, Min-Hsiu and Strelchuk, Sergii},
  journal = {PRX Quantum},
  volume = {5},
  issue = {3},
  pages = {030354},
  numpages = {27},
  year = {2024},
  publisher = {American Physical Society},
  doi = {10.1103/PRXQuantum.5.030354},
  url = {https://link.aps.org/doi/10.1103/PRXQuantum.5.030354}
}

@misc{bowman2018cellularsecondfundamentaltheorem,
      title={The cellular second fundamental theorem of invariant theory for classical groups}, 
      author={Christopher Bowman and John Enyang and Frederick Goodman},
      year={2018},
      eprint={1610.09009},
      archivePrefix={arXiv},
      primaryClass={math.RT},
      url={https://arxiv.org/abs/1610.09009}, 
}

@article{ANDERSEN_STROPPEL_TUBBENHAUER_2017, 
title={SEMISIMPLICITY OF HECKE AND (WALLED) BRAUER ALGEBRAS}, 
volume={103}, DOI={10.1017/S1446788716000392}, 
number={1}, 
journal={Journal of the Australian Mathematical Society}, 
author={Andersen, Henning Haahr and Stroppel, Catharina and Tubbenhauer, Daniel}, 
year={2017},
pages={1–44}}

@article{PhysRevA.63.042111,
  title = {Separability properties of tripartite states with $U\ensuremath{\bigotimes}U\ensuremath{\bigotimes}U$ symmetry},
  author = {Eggeling, T. and Werner, R. F.},
  journal = {Phys. Rev. A},
  volume = {63},
  issue = {4},
  pages = {042111},
  numpages = {15},
  year = {2001},
  publisher = {American Physical Society},
  doi = {10.1103/PhysRevA.63.042111},
  url = {https://link.aps.org/doi/10.1103/PhysRevA.63.042111}
}

@misc{grinko2023gelfandtsetlinbasispartiallytransposed,
      title={Gelfand-Tsetlin basis for partially transposed permutations, with applications to quantum information}, 
      author={Dmitry Grinko and Adam Burchardt and Maris Ozols},
      year={2023},
      eprint={2310.02252},
      archivePrefix={arXiv},
      primaryClass={quant-ph},
      url={https://arxiv.org/abs/2310.02252}, 
}

@article{Krovi2019efficienthigh,
  doi = {10.22331/q-2019-02-14-122},
  url = {https://doi.org/10.22331/q-2019-02-14-122},
  title = {An efficient high dimensional quantum {S}chur transform},
  author = {Krovi, Hari},
  journal = {{Quantum}},
  issn = {2521-327X},
  publisher = {{Verein zur F{\"{o}}rderung des Open Access Publizierens in den Quantenwissenschaften}},
  volume = {3},
  pages = {122},
  year = {2019}
}

@inproceedings{bacon2007quantum,
  title={The quantum Schur and Clebsch-Gordan transforms: I. Efficient qudit circuits},
  author={Bacon, Dave and Chuang, Isaac L and Harrow, Aram W},
  booktitle={Proceedings of the eighteenth annual ACM-SIAM symposium on Discrete algorithms},
  pages={1235--1244},
    eprint={quant-ph/0601001},
      archivePrefix={arXiv},
      primaryClass={quant-ph},
      url={https://arxiv.org/abs/quant-ph/0601001}, 
  year={2007}
}

@article{Kirby_2018,
   title={A practical quantum algorithm for the Schur transform},
   volume={18},
   ISSN={1533-7146},
   url={http://dx.doi.org/10.26421/QIC18.9-10-1},
   DOI={10.26421/qic18.9-10-1},
   number={9 \& 10},
   journal={Quantum Information and Computation},
   publisher={Rinton Press},
   author={Kirby, William M. and Strauch, Frederick W.},
   year={2018},
   pages={721–742} }

@article{VGTuraev_1990,
doi = {10.1070/IM1990v035n02ABEH000711},
url = {https://dx.doi.org/10.1070/IM1990v035n02ABEH000711},
year = {1990},
publisher = {},
volume = {35},
number = {2},
pages = {411},
author = {V G Turaev},
title = {OPERATOR INVARIANTS OF TANGLES, AND R-MATRICES},
journal = {Mathematics of the USSR-Izvestiya}
}

@article{KOIKE198957,
title = {On the decomposition of tensor products of the representations of the classical groups: By means of the universal characters},
journal = {Advances in Mathematics},
volume = {74},
number = {1},
pages = {57-86},
year = {1989},
issn = {0001-8708},
doi = {https://doi.org/10.1016/0001-8708(89)90004-2},
url = {https://www.sciencedirect.com/science/article/pii/0001870889900042},
author = {Kazuhiko Koike}
}

@article{BENKART1994529,
title = {Tensor Product Representations of General Linear Groups and Their Connections with Brauer Algebras},
journal = {Journal of Algebra},
volume = {166},
number = {3},
pages = {529-567},
year = {1994},
issn = {0021-8693},
doi = {https://doi.org/10.1006/jabr.1994.1166},
url = {https://www.sciencedirect.com/science/article/pii/S0021869384711665},
author = {G. Benkart and M. Chakrabarti and T. Halverson and R. Leduc and C.Y. Lee and J. Stroomer}
}

@phdthesis{bulgakova:tel-02554375,
  TITLE = {{Some Aspects Of Representation Theory Of Walled Brauer Algebras}},
  AUTHOR = {Bulgakova, Daria V.},
  URL = {https://hal.science/tel-02554375},
  SCHOOL = {{Aix Marseille Universit{\'e}}},
  YEAR = {2020},
  TYPE = {Theses},
  HAL_ID = {tel-02554375},
  HAL_VERSION = {v1},
}

@article{BEN96,
title = {Commuting actions — a tale of two groups},
author={Benkart, Georgia},
journal = {In: Lie Algebras and Their Representations. Contemporary Mathematics},
volume = {194},
number = {},
pages = {},
year = {1996},
issn = {978-0-8218-0512-1},
doi = {https://doi.org/10.1090/conm/194},
url = {https://www.ams.org/books/conm/194/},
}

@article{Cox1,
title = "On the blocks of the walled Brauer algebra",
journal = "Journal of Algebra",
volume = "320",
number = "1",
pages = "169 - 212",
year = "2008",
note = "",
issn = "0021-8693",
doi = "http://dx.doi.org/10.1016/j.jalgebra.2008.01.026",
url = "http://www.sciencedirect.com/science/article/pii/S0021869308000525",
author = "Anton Cox and Maud De Visscher and Stephen Doty and Paul Martin"
}

@phdthesis{grinko2025phd,
  title={Mixed Schur--Weyl duality in quantum information},
  author={Grinko, Dmitry},
  year={2025},
  school={Universiteit van Amsterdam}
}

@misc{WBA2024,
      title={Irreducible matrix representations for the walled Brauer algebra}, 
      author={Michał Studziński and Tomasz Młynik and Marek Mozrzymas and Michał Horodecki},
      year={2025},
      eprint={2501.13067},
      archivePrefix={arXiv},
      primaryClass={quant-ph},
      url={https://arxiv.org/abs/2501.13067}, 
}

@article{RevModPhys.81.865,
  title = {Quantum entanglement},
  author = {Horodecki, Ryszard and Horodecki, Pawe\l{} and Horodecki, Micha\l{} and Horodecki, Karol},
  journal = {Rev. Mod. Phys.},
  volume = {81},
  issue = {2},
  pages = {865--942},
  numpages = {0},
  year = {2009},
  month = {Jun},
  publisher = {American Physical Society},
  doi = {10.1103/RevModPhys.81.865},
  url = {https://link.aps.org/doi/10.1103/RevModPhys.81.865}
}

@misc{grosshans2024multicopyquantumstateteleportation,
      title={Multicopy quantum state teleportation with application to storage and retrieval of quantum programs}, 
      author={Frédéric Grosshans and Michał Horodecki and Mio Murao and Tomasz Młynik and Marco Túlio Quintino and Michał Studziński and Satoshi Yoshida},
      year={2024},
      eprint={2409.10393},
      archivePrefix={arXiv},
      primaryClass={quant-ph},
      url={https://arxiv.org/abs/2409.10393}, 
}

@article{PhysRevA.85.022330,
  title = {Optimal probabilistic simulation of quantum channels from the future to the past},
  author = {Genkina, Dina and Chiribella, Giulio and Hardy, Lucien},
  journal = {Phys. Rev. A},
  volume = {85},
  issue = {2},
  pages = {022330},
  numpages = {14},
  year = {2012},
  month = {Feb},
  publisher = {American Physical Society},
  doi = {10.1103/PhysRevA.85.022330},
  url = {https://link.aps.org/doi/10.1103/PhysRevA.85.022330}
}

@misc{taranto2025higherorderquantumoperations,
      title={Higher-Order Quantum Operations}, 
      author={Philip Taranto and Simon Milz and Mio Murao and Marco Túlio Quintino and Kavan Modi},
      year={2025},
      eprint={2503.09693},
      archivePrefix={arXiv},
      primaryClass={quant-ph},
      url={https://arxiv.org/abs/2503.09693}, 
}

@article{PhysRevLett.131.120602,
  title = {Reversing Unknown Qubit-Unitary Operation, Deterministically and Exactly},
  author = {Yoshida, Satoshi and Soeda, Akihito and Murao, Mio},
  journal = {Phys. Rev. Lett.},
  volume = {131},
  issue = {12},
  pages = {120602},
  numpages = {6},
  year = {2023},
  month = {Sep},
  publisher = {American Physical Society},
  doi = {10.1103/PhysRevLett.131.120602},
  url = {https://link.aps.org/doi/10.1103/PhysRevLett.131.120602}
}

@article{miyazaki17,
  title = {Complex conjugation supermap of unitary quantum maps and its universal implementation protocol},
  author = {Miyazaki, Jisho and Soeda, Akihito and Murao, Mio},
  journal = {Phys. Rev. Research},
  volume = {1},
  issue = {1},
  pages = {013007},
  numpages = {5},
  year = {2019},
  month = {Aug},
  publisher = {American Physical Society},
  doi = {10.1103/PhysRevResearch.1.013007},
archivePrefix = {arXiv},
       eprint = {1706.03481},
 primaryClass = {quant-ph}
}

@misc{mo2024parameterizedquantumcombsimpler,
      title={Parameterized quantum comb and simpler circuits for reversing unknown qubit-unitary operations}, 
      author={Yin Mo and Lei Zhang and Yu-Ao Chen and Yingjian Liu and Tengxiang Lin and Xin Wang},
      year={2024},
      eprint={2403.03761},
      archivePrefix={arXiv},
      primaryClass={quant-ph},
      url={https://arxiv.org/abs/2403.03761}, 
}

@misc{yoshida2024a,
      title={One-to-one Correspondence between Deterministic Port-Based Teleportation and Unitary Estimation}, 
      author={Satoshi Yoshida and Yuki Koizumi and Michał Studziński and Marco Túlio Quintino and Mio Murao},
      year={2024},
      eprint={2408.11902},
      archivePrefix={arXiv},
      primaryClass={quant-ph},
      url={https://arxiv.org/abs/2408.11902}, 
}

@misc{chen2024quantumadvantagereversingunknown,
      title={Quantum Advantage in Reversing Unknown Unitary Evolutions}, 
      author={Yu-Ao Chen and Yin Mo and Yingjian Liu and Lei Zhang and Xin Wang},
      year={2024},
      eprint={2403.04704},
      archivePrefix={arXiv},
      primaryClass={quant-ph},
      url={https://arxiv.org/abs/2403.04704}, 
}

@ARTICLE{10089837,
  author={Ebler, Daniel and Horodecki, Michał and Marciniak, Marcin and Młynik, Tomasz and Quintino, Marco Túlio and Studziński, Michał},
  journal={IEEE Transactions on Information Theory}, 
  title={Optimal Universal Quantum Circuits for Unitary Complex Conjugation}, 
  year={2023},
  volume={69},
  number={8},
  pages={5069-5082},
  doi={10.1109/TIT.2023.3263771}}

@book{10.1088/978-0-7503-3343-6,
author = {Skrzypczyk, Paul and Cavalcanti, Daniel},
title = {Semidefinite Programming in Quantum Information Science},
publisher = {IOP Publishing},
year = {2023},
series = {2053-2563},
isbn = {978-0-7503-3343-6},
url = {https://dx.doi.org/10.1088/978-0-7503-3343-6},
doi = {10.1088/978-0-7503-3343-6}
}

@ARTICLE{sedlak18,
       author = {{Sedl{\'a}k}, Michal and {Bisio}, Alessandro and {Ziman}, M{\'a}rio},
        title = "{Optimal Probabilistic Storage and Retrieval of Unitary Channels}",
      journal = {Phys. Rev. Lett.},
         year = "2019",
        month = "May",
       volume = {122},
       number = {17},
          eid = {170502},
        pages = {170502},
          doi = {10.1103/PhysRevLett.122.170502},
archivePrefix = {arXiv},
       eprint = {1809.04552},
 primaryClass = {quant-ph},
       adsurl = {https://ui.adsabs.harvard.edu/abs/2019PhRvL.122q0502S}
}

@misc{lewandowska2024,
      title={Storage and retrieval of von Neumann measurements via indefinite causal order structures}, 
      author={Paulina Lewandowska and Ryszard Kukulski},
      year={2024},
      eprint={2405.11202},
      archivePrefix={arXiv},
      primaryClass={quant-ph},
      url={https://arxiv.org/abs/2405.11202}, 
}

@article{PhysRevA.81.032324,
  title = {Optimal quantum learning of a unitary transformation},
  author = {Bisio, Alessandro and Chiribella, Giulio and D'Ariano, Giacomo Mauro and Facchini, Stefano and Perinotti, Paolo},
  journal = {Phys. Rev. A},
  volume = {81},
  issue = {3},
  pages = {032324},
  numpages = {6},
  year = {2010},
  month = {Mar},
  publisher = {American Physical Society},
  doi = {10.1103/PhysRevA.81.032324},
  url = {https://link.aps.org/doi/10.1103/PhysRevA.81.032324}
}

@article{Narayanan2006,
  author    = {Hariharan Narayanan},
  title     = {On the complexity of computing Kostka numbers and Littlewood--Richardson coefficients},
  journal   = {Journal of Algebraic Combinatorics},
  volume    = {24},
  number    = {3},
  pages     = {347--354},
  year      = {2006},
  publisher = {Springer},
  doi       = {10.1007/s10801-006-6855-4},
  url       = {https://arxiv.org/abs/math/0501176}
}

@misc{Liu2019,
  author       = {Jonathan Liu},
  title        = {A Survey of Complexity Results for Kostka Numbers and Littlewood--Richardson Coefficients},
  year         = {2019},
  howpublished = {\url{https://liujon23.github.io/src/papers/249_writeup.pdf}},
  note         = {Undergraduate survey report}
}

\end{document}